\documentclass[a4paper,UKenglish]{lipics-v2019}

\usepackage{subcaption}
\usepackage[arrow,matrix,frame,curve,cmtip]{xy}\CompileMatrices
\usepackage{graphicx,epsfig,color}

\usepackage{tikz}
\usetikzlibrary{automata,arrows,shapes.geometric,positioning}
\usepackage{tikz-cd}
\usetikzlibrary{decorations.pathmorphing}
\tikzstyle{enode}=[rectangle, draw]
\tikzstyle{anode}=[diamond, shape aspect=2, draw]

\usepackage{amsmath,stmaryrd,bm,faktor,hyperref}
\usepackage{amsthm,mathtools}
\usepackage[most]{tcolorbox}

\usepackage[appendix=append,bibliography=common]{apxproof}
\newtheoremrep{theorem}{Theorem}[section]
\newtheoremrep{lemma}[theorem]{Lemma}
\newtheoremrep{proposition}[theorem]{Proposition}
\newtheoremrep{corollary}[theorem]{Corollary}

\newcommand{\nat}{\ensuremath{\mathbb{N}}}
\renewcommand{\phi}{\varphi} 
\renewcommand{\vec}[1]{\ensuremath{\boldsymbol{#1}}}
\newcommand{\subvec}[3]{\ensuremath{\vec{#1}_{{#2},{#3}}}}

\newcommand{\sol}[2][]{\ensuremath{\mathit{sol}_{#1}({#2})}}

\newcommand{\sysupto}[2]{\ensuremath{d({#1},{#2})}}
\newcommand{\Emoves}[1]{\ensuremath{\mathbf{E}({#1})}}
\newcommand{\Amoves}[1]{\ensuremath{\mathbf{A}({#1})}}

\newcommand{\moves}[1]{\ensuremath{\mathsf{M}({#1})}}
\newcommand{\owner}[1]{\ensuremath{\mathsf{P}({#1})}}
\newcommand{\oppo}[1]{\ensuremath{\overline{#1}}}
\newcommand{\confs}{\ensuremath{\mathit{Pos}}}
\newcommand{\Econfs}{\ensuremath{\mathit{Pos}_\exists}}
\newcommand{\Aconfs}{\ensuremath{\mathit{Pos}_\forall}}
\newcommand{\priori}[1]{\ensuremath{\mathsf{i}({#1})}}

\newcommand{\prd}[1]{{#1}^{\scriptscriptstyle \times}}

\newcommand{\interval}[1]{\ensuremath{\underline{#1}}}

\newcommand{\subst}[3]{\ensuremath{{#1}[{#2}:={#3}]}}

\newcommand{\lub}{\ensuremath{\bigsqcup}}
\newcommand{\glb}{\ensuremath{\bigsqcap}}

\newcommand{\cone}[1]{\ensuremath{\mathop{\downarrow\!{#1}}}}

\newcommand{\filter}[1]{\ensuremath{\mathop{\uparrow\!{#1}}}}
\newcommand{\filtersub}[2]{\ensuremath{\mathop{\uparrow_{#2}\!{#1}}}}
\newcommand{\Pow}[1]{\ensuremath{\mathbf{2}^{#1}}}
\newcommand{\Rel}[1]{\ensuremath{\mathsf{Rel}({#1})}}
\newcommand{\asc}[1]{\ensuremath{\lambda_{#1}}}

\newcommand{\ceil}[1]{\ensuremath{\lceil{#1}\rceil}}
\newcommand{\true}{\ensuremath{\mathbf{t}}}
\newcommand{\false}{\ensuremath{\mathbf{f}}}
\newcommand{\PVar}{\ensuremath{\mathit{PVar}}}
\newcommand{\Prop}{\ensuremath{\mathit{Prop}}}
\newcommand{\sem}[2][]{\ensuremath{|\!|{#2}|\!|^{#1}}}
\newcommand{\distr}[1]{\ensuremath{\mathcal{D}({#1})}}
\newcommand{\semdia}{\ensuremath{\blacklozenge}}
\newcommand{\sembox}{\ensuremath{\blacksquare}}

\usepackage{algpseudocode}

\newcommand{\cmd}[1]{\ensuremath{\mathbf{#1}}}
\newcommand{\fnname}[1]{\textsc{#1}}

\title{Abstraction, Up-To Techniques and Games for Systems of Fixpoint Equations}
\titlerunning{Abstraction, Up-To Techniques and Games for Systems of Fixpoint Equations}

\author{Paolo Baldan}{Universit\`a Padova, Italy}{baldan@math.unipd.it}{https://orcid.org/0000-0001-9357-5599}{}
\author{Barbara K\"onig}{Universit\"at Duisburg-Essen, Germany}{barbara\_koenig@uni-due.de}{https://orcid.org/0000-0002-4193-2889}{}      
\author{Tommaso Padoan}{Universit\`a di Padova, Italy} {padoan@math.unipd.it}{}{}

\funding{Supported by the PRIN Project Analysis of Program Analyses
  (ASPRA), the University of Padova project ASTA and the DFG projects
  BEMEGA and SpeQt.}

\authorrunning{P. Baldan and B. K\"onig and T. Padoan}

\Copyright{Paolo Baldan, Barbara König, and Tommaso Padoan}



\nolinenumbers

\EventEditors{Igor Konnov and Laura Kov\'{a}cs}
\EventNoEds{2}
\EventLongTitle{31st International Conference on Concurrency Theory (CONCUR 2020)}
\EventShortTitle{CONCUR 2020}
\EventAcronym{CONCUR}
\EventYear{2020}
\EventDate{September 1--4, 2020}
\EventLocation{Vienna, Austria}
\EventLogo{}
\SeriesVolume{2017}
\ArticleNo{33}

\begin{document}

\maketitle

\begin{abstract}
  Systems of fixpoint equations over complete lattices, consisting of
  (mixed) least and greatest fixpoint equations, allow one to express
  many verification tasks such as model-checking of various kinds
  of specification logics or the check of coinductive behavioural
  equivalences.
  In this paper we develop a theory of approximation for systems of
  fixpoint equations in the style of abstract interpretation: a system
  over some concrete domain is abstracted to a system in a suitable
  abstract domain, with conditions ensuring that the abstract solution
  represents a sound/complete overapproximation of the concrete
  solution.
  Interestingly, up-to techniques, a classical approach used in
  coinductive settings to obtain easier or feasible proofs, can be
  interpreted as abstractions in a way that they naturally
  fit into our framework and extend to systems of equations.
  Additionally, relying on the approximation theory, we can 
  characterise the solution of systems of fixpoint equations
  over complete lattices in terms of a suitable parity game,
  generalising some recent work that was restricted to continuous
  lattices.
  The game view opens the way for the development of local algorithms
  for characterising the solution of such equation systems. We
  describe a local algorithm for checking the winner on specific game
  positions. This corresponds to answering the associated verification
  question (i.e., for model checking, whether a state satisfies a
  formula or, for equivalence checking, whether two states are
  behaviourally equivalent).  The algorithm can be combined with
  abstraction and up-to techniques, thus providing ways of speeding up
  the computation.
\end{abstract}

\section{Introduction}

Systems of fixpoint equations over complete lattices, consisting of
(mixed) least and greatest fixpoint equations, allow one to uniformly
express many verification tasks.  Notable examples come from the area
of model-checking. In fact, in order to express properties of infinite
computations, specification logics almost invariably rely on some
notion of recursion which leads to the use of fixpoints as key
mathematical tool.

Invariant/safety properties can be characterised as greatest
fixpoints, while liveness/reachability properties as least
fixpoints. Using both least and greatest fixpoints leads to expressive
specification logics. The $\mu$-calculus~\cite{k:prop-mu-calculus} is
a prototypical example, encompassing various other logics such as LTL
and CTL.  Another area of special interest for the present paper is
that of behavioural equivalences, which typically arise as solutions
of greatest fixpoint equations. The most famous example is
bisimilarity that can be seen
as the greatest fixpoint of a suitable operator over the lattice of
binary relations on
states (see, e.g.,~\cite{s:bisimulation-coinduction}).

In the first part of the paper we develop a theory of approximation
for systems of equations in the style of abstract interpretation.
The general idea of abstract
interpretation~\cite{cc:ai-unified-lattice-model,CC:SDPAF} consists of
extracting properties of programs by defining an approximated program
semantics over a so-called abstract domain, usually a complete
lattice.
Concrete and abstract
semantics are typically
expressed in terms of (systems of) least fixpoint equations, 
with conditions ensuring that the approximation
obtained is sound, i.e., that properties derived from
the abstract semantics are also valid at the concrete level. In an
ideal situation also the converse holds and the abstract interpretation is called complete (see
e.g.~\cite{GRS:MAIC}).
Abstract interpretation has been applied also for the model checking of various kinds of mu-calculi and temporal logics (see, e.g.,~\cite{glls:abstraction-full-mu-calculus,lgsbb:property-preserving-abstractions,CC:TLA,s:relations-abstraction-refinement,dgg:ai-reactive,lt:modal-process-logic}).

We generalise this idea to systems of fixpoint equations, where least
and greatest fixpoints can coexist (\S\ref{se:abstractions}).
A system over some concrete domain $C$ is abstracted by a system over
some abstract domain $A$. Suitable conditions are
identified that ensure the soundness and completeness of the
approximation.
This enables the use of the approximation theory on a number of
verification tasks.
We show how to recover some results on property
preserving abstractions for the
$\mu$-calculus~\cite{lgsbb:property-preserving-abstractions}. We also discuss a fixpoint extension of {\L}ukasiewicz logic,
considered in~\cite{MS:MS} as a precursor to model-checking PCTL or
probabilistic $\mu$-calculi.

When dealing with greatest fixpoints, a key proof technique relies on
the coinduction principle, which uses the fact that a
monotone function $f$ over a complete lattice has a greatest fixpoint
$\nu f$, which is the join of all post-fixpoints, i.e., the elements
$l$ such that $l \sqsubseteq f(l)$. As a consequence proving
$l \sqsubseteq f(l)$ suffices to conclude that $l \sqsubseteq \nu
f$.

Up-to techniques have been proposed for
``simplifying''
proofs~\cite{Mil:CC,SM:PBUT,ps:enhancements-coinductive,p:complete-lattices-up-to}
and for reducing the search space in
verification (e.g., in~\cite{BP:NFA}, up-to techniques applied to language equivalence of NFAs
are shown to provide in many cases an exponential speed-up).
A sound up-to function is a function $u$ on the lattice such
that $\nu (f \circ u) \sqsubseteq \nu f$ and hence
$l \sqsubseteq f(u(l))$ implies
$l \sqsubseteq \nu (f \circ u) \sqsubseteq \nu f$. The
characteristics of $u$ (typically, extensiveness) make it
easier to show that an element is a post-fixpoint of $f \circ u$
rather than a post-fixpoint of $f$.

We show that
up-to techniques admit a natural interpretation as abstractions in our
framework
(\S\ref{se:up-to-from-abstraction}). This allows us to generalise
the theory of up-to techniques to systems of fixpoint equations and contributes to the understanding of the relation
between abstract interpretation and up-to techniques, a theme that
received some recent
attention~\cite{bggp:sound-up-to-complete-abstract}.

We have recently shown in~\cite{BKMP:FPCL} that the solution of
systems of fixpoint equations can be characterised in terms of a
parity game when working in a suitable subclass of complete lattices,
the so-called continuous lattices~\cite{Scott:CL}.
Here, relying on our approximation theory, we get rid of continuity and design a game
that works for
general complete lattices
(\S\ref{se:game-characterization}).

The above results open the way to the development of game-theoretical
algorithms, possibly integrating abstraction and up-to techniques, for
solving systems of equations over complete lattices.  While global
algorithms deciding the game at all positions, based on progress
measures~\cite{j:progress-measures-parity}, have already been studied
in~\cite{hsc:lattice-progress-measures,BKMP:FPCL}, here we focus on
local algorithms, confining the attention to specific positions.
For instance, in the case of the $\mu$-calculus, rather than computing
the set of states satisfying some formula $\varphi$, one could be
interested in checking whether a specific state satisfies or does not
satisfy $\varphi$. For probabilistic logics, rather than determining
the full evaluation of $\varphi$, we could be interested in
determining the value for a specific state or only in establishing a
bound for such a value. Similarly, in the case of behavioural
equivalences, rather than computing the full behavioural relation, one
could be interested in determining whether two specific states are
equivalent.
Taking inspiration from backtracking methods for
bisimilarity~\cite{h:proving-up-to} and for the
$\mu$-calculus~\cite{s:local-modcheck-games,ss:practical-modcheck-games},
we first design a local  (also called on-the-fly) algorithm for the
case of a single equation (\S\ref{ss:algorithmic-view}). The algorithm is then extended to general
systems~\S\ref{se:on-the-fly}.
We also show how these algorithms can be
enhanced with up-to techniques.

This also establishes a link with some recent work relating abstract
interpretation and up-to
techniques~\cite{bggp:sound-up-to-complete-abstract} and exploiting
up-to techniques for language equivalence on NFAs~\cite{BP:NFA}.

This paper is the full version of
\cite{bkp:abstraction-up-to-games-fixpoint}, extended with additional
examples, containing all proofs and a description of the local
algorithm for general systems.

\section{Preliminaries and Notation}
\label{sec:preliminaries}

A preordered or partially ordered set $\langle P, \sqsubseteq \rangle$
is often denoted simply as $P$, omitting the (pre)order
relation. Given $X \subseteq P$, we denote by
$\cone{X} = \{p\in P\mid \exists x\in X.\ p \sqsubseteq x\}$ the
\emph{downward-closure} and by
$\filter{X} = \{p\in P\mid \exists x\in X.\ x \sqsubseteq p\}$ the
\emph{upward-closure} of $X$.
The \emph{join} and the \emph{meet} of a
subset $X \subseteq P$ (if they exist) are denoted $\lub X$
and $\glb X$, respectively.
 
\begin{definition}[complete lattice, basis]
  A \emph{complete lattice} is a partially ordered set
  $(L, \sqsubseteq)$ such that each subset $X \subseteq L$ admits a
  join $\lub X$ and a meet $\glb X$. A complete lattice
  $(L, \sqsubseteq)$ always has a least element
  $\bot = \lub \emptyset$ and a greatest element
  $\top = \glb \emptyset$.
  A \emph{basis} for a complete lattice is a subset $B_L \subseteq L$ such that
  for each $l \in L$ it holds that
  $l = \lub (\cone{l}\, \cap\, B_L)$.
\end{definition}

For instance, the powerset of any set $X$, ordered by subset inclusion
$(\Pow{X}, \subseteq)$ is a complete lattice. Join is union, meet is
intersection, top ($\top$) is
$X$ and bottom ($\bot$) is $\emptyset$. A basis is the set of singletons
$B_{\Pow{X}} = \{ \{x\} \mid x \in X\}$.
Another complete lattice used in the paper is the real interval
$[0,1]$ with the usual order $\leq$. Join and meet are the sup and inf
over the reals, $0$ is bottom and $1$ is top. Any dense subset,
e.g., the set of rationals $\mathbb{Q} \cap (0,1]$, is a basis.
  
A function $f\colon L \to L$ is \emph{monotone} if for all
$l, l' \in L$, if $l \sqsubseteq l'$ then $f(l) \sqsubseteq f(l')$. By
Knaster-Tarski's theorem~\cite[Theorem 1]{t:lattice-fixed-point}, any
monotone function $f$ on a complete lattice has a least fixpoint
arising as the meet of all pre-fixpoints
$\mu f = \glb \{ l \mid f(l) \sqsubseteq l \}$ and a greatest fixpoint
arising as the join of all post-fixpoints
$\nu f = \lub \{ l \mid l \sqsubseteq f(l) \}$.

The least and greatest fixpoint can also be obtained by iterating the
function on the bottom and top elements of the lattice. This is often
referred to as Kleene's theorem (at least for continuous functions)
and it is one of the pillars of abstract
interpretation~\cite{CC:CCVTFP}. Given a complete lattice $L$, define its
\emph{height} $\asc{L}$ as the supremum of the length of any strictly
ascending, possibly transfinite, chain. Then we have the following result.

\begin{theorem}[Kleene's iteration~\cite{CC:CCVTFP}]
  Let $L$ be a complete lattice and let $f \colon L \to L$ be a monotone
  function.  Consider the (transfinite) ascending chain
  $(f^{\beta}(\bot))_{\beta}$ where $\beta$ ranges over the ordinals,
  defined by $f^0(\bot) = \bot$,
  $f^{\alpha+1}(\bot) = f(f^\alpha(\bot))$ for any ordinal $\alpha$
  and $f^{\alpha}(\bot) = \lub_{\beta < \alpha} f^{\beta}(\bot)$ for
  any limit ordinal $\alpha$. Then $\mu f = f^{\gamma}(\bot)$ for some
  ordinal $\gamma \leq \asc{L}$. The greatest fixpoint $\nu f$ can be
  characterised dually, via the (transfinite) descending chain
  $(f^{\alpha}(\top))_{\alpha}$.
\end{theorem}

Given a complete lattice $L$, a subset
$X \subseteq L$ is \emph{directed} if $X \neq \emptyset$ and every
pair of elements in $X$ has an upper bound in $X$. If $L, L'$ are
complete lattices, a function $f : L \to L'$ is
\emph{(directed-)continuous} if for any directed set $X \subseteq L$
it holds that $f(\lub X) = \lub f(X)$. The function $f$ is called \emph{strict} if $f(\bot) = \bot$. \emph{Co-continuity} and \emph{co-strictness} are defined dually.

\begin{definition}[Galois connection]
  Let $(C, \sqsubseteq)$, $(A, \leq)$ be complete lattices. A
  \emph{Galois connection} (or \emph{adjunction}) is a pair of
  monotone functions $\langle \alpha, \gamma\rangle$ such that
  $\alpha : C \to A$, $\gamma : A \to C$ and for all $a \in A$ and
  $c \in C$ it holds that $\alpha(c) \leq a$ iff
  $c \sqsubseteq \gamma (a)$.

  Equivalently, for all
  $a \in A$ and $c \in C$, (i)
  $c \sqsubseteq \gamma (\alpha(c))$ and
  (ii)~$\alpha(\gamma(a)) \leq a$.
  In this case we will write
  $\langle \alpha, \gamma\rangle : C \to A$.
  The Galois connection is called an \emph{insertion} when
  $\alpha \circ \gamma = id_A$.
\end{definition}

For a Galois connection
$\langle \alpha, \gamma \rangle: C \to A$, the function
$\alpha$ is called the left (or lower) adjoint and $\gamma$ the right
(or upper) adjoint.
The left adjoint $\alpha$ preserves all joins and the right adjoint
$\gamma$ preserves all meets. Hence, in particular, the left adjoint is
strict and continuous, while the right adjoint is co-strict and co-continuous.

A function $f : L \to L$ is \emph{idempotent} if $f \circ f = f$ and
\emph{extensive} if $l \sqsubseteq f(l)$ for all $l \in L$. When $f$
is monotone, extensive and idempotent it is called an \emph{(upper)
  closure}. In this case, $\langle f, i\rangle : L \to f(L)$, where
$i$ is the inclusion, is an insertion
and
$f(L) = \{ f(l) \mid l \in L \}$ is a complete lattice.

We will often consider tuples of elements.
Given a set $A$, an $n$-tuple in $A^n$ is denoted by a
boldface letter $\vec{a}$ and its components are
denoted as $\vec{a} = (a_1, \ldots, a_n)$. For an index
$n \in \mathbb{N}$ we write $\interval{n}$ for the
integer interval $\{1, \ldots, n\}$. Given $\vec{a} \in A^n$ and
$i, j \in \interval{n}$ we write $\subvec{a}{i}{j}$ for the subtuple
$(a_i, a_{i+1}, \ldots, a_j)$.
The empty tuple is denoted by $()$.
Given two tuples $\vec{a} \in A^m$ and $\vec{a}' \in A^n$ we denote by $(\vec{a}, \vec{a'})$ or simply by $\vec{a} \vec{a'}$ their concatenation in $A^{m+n}$.

Given a complete lattice $(L,\sqsubseteq)$ we will denote by
$(L^n,\sqsubseteq)$ the set of $n$-tuples endowed with the
\emph{pointwise order} defined, for $\vec{l}, \vec{l'} \in L^n$, by
$\vec{l} \sqsubseteq \vec{l'}$ if $l_i \sqsubseteq l_i'$ for all
$i \in \interval{n}$.
The structure $(L^n, \sqsubseteq)$ is a complete lattice.
More generally, for any
set $X$, the set of functions $L^X = \{ f \mid f\colon X \to L \}$, endowed with pointwise order, is a complete lattice.

A tuple of functions $\vec{f} = (f_1, \ldots, f_m)$ with
$f_i : X \to Y$, will be seen itself as a function
$\vec{f} : X \to Y^m$, defined by
$\vec{f}(x) = ( f_1(x), \ldots, f_m(x) )$. We will also need to
consider the \emph{product function} $\prd{\vec{f}} : X^m \to Y^m$, defined
by $\prd{\vec{f}}(x_1, \ldots, x_m) = ( f_1(x_1), \ldots, f_m(x_m) )$.

\section{Systems of Fixpoint Equations over Complete Lattices}
\label{sec:system}

We deal with systems of (fixpoint) equations over
some complete lattice, where, for each equation one can be interested
either in the least or in the greatest solution. We define systems, their solutions and we provide some examples that will be used as running examples.

\begin{definition}[system of equations]
  \label{de:system}
  Let $L$ be a complete lattice. A system of equations $E$ over $L$ is
  an ordered list of $m$ equations of the form
  $x_i =_{\eta_i} f_i(x_1,\ldots,x_m)$, where $f_i\colon L^m\to L$ are
  monotone functions (with respect to the pointwise order on $L^m$)
  and $\eta_i\in\{\mu,\nu\}$.
  The system will often be denoted as
  $\vec{x} =_{\vec{\eta}} \vec{f} (\vec{x})$, where $\vec{x}$,
  $\vec{\eta}$ and $\vec{f}$ are the obvious tuples.
  We denote by
  $\emptyset$ the system with no equations.
\end{definition}

Systems
of this kind have been often considered
in connection to verification problems 
(see e.g.,~\cite{cks:faster-modcheck-mu,s:fast-simple-nested-fixpoints,hsc:lattice-progress-measures,BKMP:FPCL}).
In particular,~\cite{hsc:lattice-progress-measures,BKMP:FPCL} work on general classes of complete
lattices.

Note that $\vec{f}$ can be seen as a function $\vec{f}\colon L^m \to L^m$. The solution of the system is a selected fixpoint of such function.
We first need some auxiliary notation.

\begin{definition}[substitution]
  \label{de:substitution}
  Given a system $E$ of $m$ equations over a complete lattice $L$ of the kind
  $\vec{x} =_{\vec{\eta}} \vec{f} (\vec{x})$, an index
  $i \in \interval{m}$ and $l \in L$ we write
  $\subst{E}{x_i}{l}$ for the system of $m-1$ equations obtained from
  $E$ by removing the $i$-th equation and replacing $x_i$ by $l$ in the other equations,
  i.e., if $\vec{x} = \vec{x}' x_i \vec{x}''$,
  $\vec{\eta} = \vec{\eta}' \eta_i \vec{\eta}''$ and
  $\vec{f} = \vec{f}' f_i \vec{f}''$ then $\subst{E}{x_i}{l}$ is
  $\vec{x}' \vec{x}'' =_{\vec{\eta}' \vec{\eta}''} \vec{f}'
  \vec{f}''(\vec{x}', l, \vec{x}'')$.
\end{definition}

For solving a system of $m$ equations
$\vec{x} =_{\vec{\eta}} \vec{f}(\vec{x})$, the last variable $x_m$ is
considered as a fixed parameter $x$ and the system of $m-1$ equations
$\subst{E}{x_m}{x}$ that arises from dropping the last equation is recursively
solved. This produces an $(m-1)$-tuple parametric on $x$, i.e., we get
$\subvec{s}{1}{m-1}(x) = \sol{\subst{E}{x_m}{x}}$. Inserting this
parametric solution into the last equation, we get an equation in a
single variable $x =_{\eta_m} f_m(\subvec{s}{1}{m-1}(x), x)$ that can
be solved by taking for the
function
$\lambda x.\, f_m(\subvec{s}{1}{m-1}(x), x)$, the least or greatest
fixpoint, depending on whether the last equation is a $\mu$- or
$\nu$-equation.
This provides the $m$-th component of the solution
$s_m = \eta_m (\lambda x.\, f_m(\subvec{s}{1}{m-1}(x), x))$. The
remaining components
are obtained inserting $s_m$ in
the parametric solution $\subvec{s}{1}{m-1}(x)$ previously computed, i.e.,
$\subvec{s}{1}{m-1} = \subvec{s}{1}{m-1}(s_m)$.

\begin{definition}[solution]
  \label{de:solution}
  Let $L$ be a complete lattice and let $E$ be a system of $m$ equations over
  $L$ of the kind $\vec{x} =_{\vec{\eta}} \vec{f}(\vec{x})$.
  The \emph{solution} of $E$, denoted $\sol{E} \in L^m$, is defined
  inductively:
  \begin{center}
    $
      \sol{\emptyset} = () \qquad\qquad
      \sol{E} = (\sol{\subst{E}{x_m}{s_m}},\, s_m) 
    $
  \end{center}
  where $s_m = \eta_m (\lambda x.\, f_m(\sol{\subst{E}{x_m}{x}},x))$.
\end{definition}

In words, for solving a system of $m$ equations, the last
variable is considered as a fixed parameter $x$ and the system of
$m-1$ equations that arises from dropping the last equation is
recursively solved. This produces an $(m-1)$-tuple parametric on $x$,
i.e., we get $\subvec{s}{1}{m-1}(x) = \sol{\subst{E}{x_m}{x}}$. Inserting
this parametric solution into the last equation, 
we get an equation in a single variable
\begin{center}
  $x  =_{\eta_m} f_m(\subvec{s}{1}{m-1}(x), x)$
\end{center}
that can be solved by taking for the 
function
$\lambda x.\, f_m(\subvec{s}{1}{m-1}(x), x)$, the least or greatest
fixpoint, depending on whether the last equation is a $\mu$- or
$\nu$-equation.
This provides the $m$-th component of the solution
$s_m = \eta_m (\lambda x.\, f_m(\subvec{s}{1}{m-1}(x), x))$. The
remaining components of the solution are obtained inserting $s_m$ in
the parametric solution $\subvec{s}{1}{m-1}(x)$ previously computed, i.e.,
$\subvec{s}{1}{m-1} = \subvec{s}{1}{m-1}(s_m)$.

The order of equations matters: changing the order
typically leads to a different solution.

\begin{example}[solving a simple system of equations]
  \label{ex:order-of-eqns}
  Consider the powerset lattice $\Pow{S}$ of any non-empty set $S$ and the system of equations $E$ consisting of the following two equations
   \begin{center}
     $x =_\mu x \cup y$\\
     $y =_\nu x \cap y$
  \end{center}
  In order to solve the system $E$, initially we need to compute the
  solution of the first equation $x =_\mu x \cup y$ parametric in $y$,
  that is, $s_x(y) = \mu(\lambda x. (x \cup y)) = y$.
  Now we can solve the second equation $y =_\nu x \cap y$ replacing
  $x$ with the parametric solution, obtaining an equation in
  a single variable whose solution is
  $\nu (\lambda y. (s_x(y) \cap y)) = \nu (\lambda y. y) = S$. Finally, the
  solution of the first equation is obtained by inserting $y = S$ in
  the parametric solution $x = s_x(S) =S$.

  Observe that even in this simple example the order of the equations
  matters. Indeed, if we consider the system where the two
  equations above are swapped the solution is $x = y = \emptyset$.
\end{example}

\begin{example}[$\mu$-calculus formulae as fixpoint equations]
  \label{ex:mu}
  We adopt a standard $\mu$-calculus syntax. For fixed disjoint sets
  $\PVar$ of propositional variables, ranged over by $x, y, z, \ldots$
  and $\Prop$ of propositional symbols, ranged over by
  $p,q,r, \ldots$, formulae are defined by
  \begin{center}
    $\varphi \ ::=\ \true\ \mid\ \false\ \mid\ p \mid\ x\ \mid\ \varphi
    \land \varphi\ \mid\ \varphi \lor \varphi\ \mid\ \Box \varphi\ \mid\
    \Diamond \varphi\ \mid\ \eta x.\, \varphi$
  \end{center}
  where $p \in \Prop$, $x \in \PVar$ and $\eta \in \{ \mu, \nu
  \}$.

  The semantics of a formula is given with respect to an unlabelled
  transition system (or Kripke structure) $T = (\mathbb{S}_T, \to_T)$ where
  $\mathbb{S}_T$ is the set of states and
  $\to_T\ \subseteq \mathbb{S}_T \times \mathbb{S}_T$ is the transition
  relation. Given a formula $\varphi$ and an environment
  $\rho\colon \Prop \cup \PVar \to \Pow{\mathbb{S}_T}$ mapping each
  proposition or propositional variable to the set of states where it
  holds, we denote by $\sem[T]{\varphi}_\rho$ the semantics of $\varphi$
  defined as usual (see, e.g.,~\cite{bw:mu-calculus-modcheck}).

  \begin{figure}
    \captionsetup[subfigure]{justification=centering}
    \begin{subfigure}[b]{.22\textwidth}
      \centering
      \begin{tikzpicture}[->, >=stealth, nodes={draw, circle, minimum size=1.8em}, node distance=0.5cm]
        \node (a) {$a$};
        \node (c) [below=of a]{$c$};
        \node (b) [right=0.7 of a,fill=black!10] {$b$};
        \node (d) [below left=0.7 and 0 of b,fill=black!10]{$d$};
        \node (e) [right=of d,fill=black!10]{$e$};
        \path
        (a) edge (b)
        (a) edge [loop left] ()
        (a) edge (c) 
        (b) edge (d)
        (b) edge (e) 
        (c) edge [loop left] ()
        (d) edge [in=150, out=120, loop] ()
        (e) edge [in=30, out=60, loop] ();
      \end{tikzpicture}
      \caption{}
      \label{fi:running-ts}
    \end{subfigure}
    \begin{subfigure}[b]{.22\textwidth}
      \centering
      \begin{tikzpicture}[->, >=stealth, nodes={draw, circle, minimum size=1.6em, inner sep=0.3em}, node distance=0.5cm]
        \node (a) {$a$};
        \node (c) [below=of a]{$c$};
        \node (bd) [right=of a,fill=black!10] {{\footnotesize $bde$}};
        \path
        (a) edge (bd)
        (a) edge [loop left] ()
        (a) edge (c) 
        (c) edge [loop left] ()        
        (bd) edge [loop below] ();
      \end{tikzpicture}
      \caption{}
      \label{fi:running-ts-bis}
    \end{subfigure}
    \begin{subfigure}[b]{.22\textwidth}
      \centering
      $
        \begin{array}{lcl}
          x_1 & =_{\nu} & p \land \Box x_1\\
          x_2 & =_{\mu} & x_1 \lor \Diamond x_2
        \end{array}
      $
      \caption{}
      \label{fi:running-eqf}
    \end{subfigure}
    \hspace{2mm}
    \begin{subfigure}[b]{.22\textwidth}
      \centering
      $
        \begin{array}{lcl}
          x_1 & =_{\nu} & \{ b, d, e\} \cap \sembox_T x_1\\
          x_2 & =_{\mu} & x_1 \cup \semdia_T x_2
        \end{array}
      $
      \caption{}
      \label{fi:running-eq}
    \end{subfigure}
    \caption{}
  \end{figure}

  As observed by several authors
  (see, e.g.,~\cite{cks:faster-modcheck-mu,s:fast-simple-nested-fixpoints}),
  a $\mu$-calculus formula can be seen as a system of equations, with
  an equation for each fixpoint subformula.
  For instance, consider
  $\phi = \mu x_2.((\nu x_1.(p \land\Box x_1))\lor\Diamond x_2)$ that
  requires that a state is eventually reached from which $p$ always
  holds. The equational form is reported in Fig.~\ref{fi:running-eqf}.
  Consider a transition system $T = (\mathbb{S}_T, \to_T)$ where
  $\mathbb{S}_T = \{a,b,c,d, e\}$ and $\to_T$ is as depicted in
  Fig.~\ref{fi:running-ts}, with $p$ that holds in the grey states
  $b$, $d$ and $e$.
  Define the semantic counterpart of the modal operators as follows:
  given a
  relation $R \subseteq X \times X$ let 
  $\semdia_R, \sembox_R : \Pow{X} \to \Pow{X}$ be the functions
  defined,  for $Y \subseteq X$,  by
  $\semdia_R(Y) = \{x \in X \mid \exists y \in Y.\ (x,y) \in R\}$,
  $\sembox_R(Y) = \{x \in X \mid \forall y \in X. (x, y) \in R
  \Rightarrow y \in Y \}$.
  Then the formula $\varphi$ interpreted over the transition system
  $T$ leads to the system of equations over the lattice
  $\Pow{\mathbb{S}_T}$ in Fig.~\ref{fi:running-eq}, where we write
  $\semdia_T$ and $\sembox_T$ for $\semdia_{\to_T}$ and
  $\sembox_{\to_T}$.

  The solution is $x_1 = \{ b, d, e\}$ (states where $p$
  always holds) and $x_2 = \{ a, b, d, e \}$ (states where
  the formula $\varphi$ holds).
\end{example}

\begin{example}[{\L}ukasiewicz $\mu$-terms]
  \label{ex:reals}
  Systems of equations over the real interval $[0,1]$ have been
  considered in~\cite{MS:MS} as a precursor to model-checking PCTL or
  probabilistic $\mu$-calculi. More precisely, the authors study a
  fixpoint extension of {\L}ukasiewicz logic, referred to as
  {\L}ukasiewicz $\mu$-terms, whose syntax is as follows:
  \begin{center}
    $t \ ::=\ \mathbf{1}\ \mid\
    \mathbf{0}\ \mid\ x\ \mid\
    r \cdot t\ \mid\
    t \sqcup t\ \mid\
    t \sqcap t\ \mid\
    t \oplus t\ \mid\
    t \odot t\ \mid\
    \eta x. t$
  \end{center}
  where $x \in \PVar$ is a variable (ranging over $[0,1]$),
  $r \in [0,1]$ and $\eta \in \{\mu, \nu \}$. The various syntactic
  operators have a semantic counterpart, given in
  Fig.~\ref{fi:semantics-mu}.

  Then, each {\L}ukasiewicz $\mu$-term, in an environment
  $\rho : \PVar \to [0,1]$, can be assigned a semantics which is a real
  number in $[0,1]$, denoted as $\sem{t}_\rho$.
  Exactly as for the $\mu$-calculus, a {\L}ukasiewicz $\mu$-term can
  be naturally seen as a system of fixpoint equations over the lattice $[0,1]$.
  For instance, the term
  $\nu x_2.\ (\mu x_1.\ (\frac{5}{8} \oplus \frac{3}{8} x_2) \odot (\frac{1}{2}
  \sqcup (\frac{3}{8} \oplus \frac{1}{2} x_1)))$ from an example
  in~\cite{MS:MS}, can be written as the system:
  \begin{align*}
    x_1 & =_\mu (\frac{5}{8} \oplus \frac{3}{8} x_2) \odot (\frac{1}{2}
  \sqcup (\frac{3}{8} \oplus \frac{1}{2} x_1))\\
    x_2 & =_\nu x_1
  \end{align*}
\end{example}

\begin{example}[{\L}ukasiewicz $\mu$-calculus]
  \label{ex:lukasievicz-modal}
  The {\L}ukasiewicz $\mu$-calculus, as defined in~\cite{MS:MS},
  extends the {\L}ukasiewicz $\mu$-terms with propositions and modal
  operators. The syntax is as follows:
  \begin{center}
    $\varphi \ ::=\
    p\                      \mid\
    \bar{p}\                \mid\
    r \cdot \varphi\        \mid\
    \varphi \sqcup \varphi\ \mid\
    \varphi \sqcap \varphi\ \mid\
    \varphi \oplus \varphi\ \mid\
    \varphi \odot \varphi\  \mid\
    \Diamond \varphi\       \mid\
    \Box \varphi\           \mid\
    \eta x. t$
  \end{center}
  where $x$ ranges in a set $\PVar$ of propositional variables, $p$
  ranges in a set $\Prop$ of propositional symbols, each paired with
  an associated complement $\bar{p}$, and $\eta \in \{ \mu, \nu \}$.

  The {\L}ukasiewicz $\mu$-calculus can be seen as a logic for probabilistic
  transition systems.   It extends the quantitative modal $\mu$-calculus of~\cite{mm:quantitative-mu,hk:quantitative-analysis-mc} and it allows to encode PCTL~\cite{BdA:MCPNS}.
  For a finite set $\mathbb{S}$, the set of (discrete) probability
  distributions over $\mathbb{S}$ is defined as
  $\distr{\mathbb{S}} = \{ d : \mathbb{S} \to [0,1] \mid \sum_{s \in
    \mathbb{S}} d(s) =1 \}$. A formula is interpreted over a
  \emph{probabilistic non-deterministic transition system (PNDT)}
  $N = (\mathbb{S}, \to)$ where
  $\to\ \subseteq \mathbb{S} \times \distr{\mathbb{S}}$ is the
  transition relation.
  An example of PNDT can be found in Fig.~\ref{fi:pndt}. Imagine that
  the aim is to reach state $b$.  State $a$ has two transitions. A
  ``lucky'' one where the probability to get to $b$ is $\frac{1}{3}$
  and an ``unlucky'' one where $b$ is reached with probability
  $\frac{1}{6}$. For both transitions, with probability $\frac{1}{3}$
  one gets back to $a$ and then, with the residual probability, one
  moves to $c$.  Once in
  states $b$ or $c$, the system remains in the same state with
  probability $1$.

  \begin{figure}
    \begin{subfigure}[b]{.49\textwidth}
      \centering
      \begin{tabular}{ll}
        $\mathbf{0}(x) = 0$, $\mathbf{1}(x)=1$ & {\small (constant)}\\
        $r \cdot x = r x$ & {\small (scalar mult.)} \\
        $x \sqcup y = \max(x,y)$ & {\small (weak disj.)}\\
        $x \sqcap y = \min(x, y)$ & {\small (weak conj.)}\\
        $x \oplus y = \min(x+y, 1)$ &  {\small (strong disj.)}\\
        $x \odot y = \max(x+y -1, 0)$ &  {\small (strong conj.)}
      \end{tabular}
      \caption{Semantics of $\mu$-terms ($x, y \in [0,1]$)}
      \label{fi:semantics-mu}
    \end{subfigure}
    \begin{subfigure}[b]{.22\textwidth}
      \centering
      \begin{tikzpicture}[->, >=stealth, state/.style={draw, circle, minimum size=1.8em}, lab/.style={font=\scriptsize}, node distance=0.5cm]
        \node (a)  [state] {$a$};
        \node (t1) [above right=of a] {$\cdot$};
        \node (t2) [below right=of a] {$\cdot$};
        \node (b)  [state, fill=black!10, right=of t1]{$b$};
        \node (c)  [state,right=of t2]{$c$};
        \draw [thick, -]  (a)  to (t1) {};
        \draw [bend right=45, ->]  (t1) to node[lab, above] {$\frac{1}{3}$} (a);
        \draw [->] (t1) to node[lab, above] {$\frac{1}{3}$} (b);
        \draw [->] (t1) to node[lab, near start, left] {$\frac{1}{3}$} (c) ;
        \draw [thick, -] (a) to (t2) {};
        \draw [bend left=45, ->] (t2) to node[lab, below] {$\frac{1}{3}$}  (a);
        \draw [->] (t2) to node[lab, near start, left] {$\frac{1}{6}$} (b);
        \draw [->] (t2) to node[lab, below] {$\frac{1}{2}$} (c);
        \draw [->, loop right] (b) to node[lab,right] {$1$} (c);        
        \draw [->, loop right] (c) to node[lab,right] {$1$} (c);
      \end{tikzpicture}
      \caption{A PNDT}
      \label{fi:pndt}
    \end{subfigure}
    \begin{subfigure}[b]{.23\textwidth}
      \centering
      $
      \begin{array}{ll}
        \phi &
               \left\{               
               \begin{array}{l}
                 x_1 =_\nu p \odot \Box x_1\\
                 x_2 =_\mu x_1 \oplus \Diamond x_2 
               \end{array}
        \right.
        \\[5mm]
        \phi' &
                \left\{
                \begin{array}{l}
                  x_1 =_\nu p \odot \Box x_1\\
                  x_2 =_\mu x_1 \oplus \Box x_2 
                \end{array}
        \right.
      \end{array}      
      $
      \caption{Formulas as systems}
      \label{fi:pndt-formulas}
    \end{subfigure}
    \caption{}
  \end{figure}

  Given a formula $\varphi$ and an environment
  $\rho : \Prop \cup \PVar \to (\mathbb{S}\to [0,1])$ mapping each
  proposition or propositional variable to a real-valued function over
  the states, the semantics of $\varphi$ is a function
  $\sem[N]{\varphi}_\rho\colon \mathbb{S}\to [0,1]$ defined as expected
  using the semantic operators.
  In addition to those already discussed, we have the semantic
  operators for the complement and the modalities: for
  $v\colon \mathbb{S}\to [0,1]$
  \begin{align*}
    \bar{v}(x)
    = 1 - v(x)
    \qquad
    \semdia_N(v)(x)
    = \max_{x \to d} \sum_{y \in \mathbb{S}} d(y)\cdot v(y)
    \qquad
    \sembox_N(v)(x)
    = \min_{x \to d} \sum_{y \in \mathbb{S}} d(y)\cdot v(y)
  \end{align*}
  As it happens for the propositional $\mu$-calculus, also formulas of
  the {\L}ukasiewicz $\mu$-calculus can be seen as systems of equations,
  but on a different complete lattice, i.e., $[0,1]^{\mathbb{S}}$. For
  instance, consider the formulas
  $\varphi = \mu x_2. ( \nu x_1. (p \odot \Box x_1) \oplus \Diamond
  x_2)$ and
  $\varphi' = \mu x_2. ( \nu x_1. (p \odot \Box x_1) \oplus \Box x_2)$,
  rendered as (syntactic) equations in
  Fig.~\ref{fi:pndt-formulas}. Roughly speaking, they capture the
  probability of eventually satisfying forever $p$, with an angelic
  scheduler and a daemonic one, choosing at each step the best
  or worst transition, respectively. Assuming that $p$ holds with
  probability $1$ on $b$ and $0$ on $a$ and $c$, we have
  $\sem[]{\varphi}_\rho(a) = \frac{1}{2}$ and
  $\sem[]{\varphi'}_\rho(a)=\frac{1}{4}$.
\end{example}

\begin{example}[(bi)similarity over transition systems]
  \label{ex:bisimilarity}
  For defining (bi)similarity uniformly with the example on
  $\mu$-calculus, we work on unlabelled transition systems with atoms
  $T = (\mathbb{S}, \to, A)$ where $A \subseteq \Pow{\mathbb{S}}$ is a
  fixed set of atomic properties over the states.
  Everything can be easily adapted to labelled transition systems.

  Given $T = (\mathbb{S}, \to, A)$,
  consider the lattice of relations on $\mathbb{S}$, namely
  $\Rel{\mathbb{S}} = (\Pow{\mathbb{S}\times\mathbb{S}},
  \subseteq)$. We take as basis the set
  of singletons
  $B_{\Rel{\mathbb{S}}} = \{\{(x,y)\}\,\mid\,x,y \in \mathbb{S}\}$. The
  \emph{similarity relation} on $T$, denoted $\precsim_T$, is
  the greatest
  fixpoint of the function
  $\mathit{sim}_T : \Rel{\mathbb{S}} \to \Rel{\mathbb{S}}$, defined by
  \[
    \mathit{sim}_T(R) =
    \left\{ 
    \begin{array}{ll}
      \!\!(x,y) \in R\ \mid\!\!\!
      & \forall a \in A.\, (x \in a \Rightarrow y \in a) \land
      \forall x \to x'.\ \exists y \to y'.\ (x',y') \in R
    \end{array}
    \right\}
  \]
  In other words it can be seen as the solution of a system consisting
  of a single greatest fixpoint equation $x =_\nu \mathit{sim}_T(x)$.

  For instance, consider the transition system $T$ in
  Fig.~\ref{fi:running-ts} and take $p = \{ b, d , e \}$ as the only
  atom.
  Then similarity $\precsim_T$ is the transitive reflexive closure
  of $\{(c,a), (a,b), (b,d), (d,e), (e,b)\}$.

  Bisimilarity $\sim_T$ can be obtained analogously as the greatest
  fixpoint of
  $\mathit{bis}_T(R) = \mathit{sim}_T(R) \cap \mathit{sim}_T(R^{-1})$.
  In the transition system $T$ above, bisimilarity $\sim_T$ is the
  least equivalence such that $b \sim_T d \sim_T e$.
\end{example}

\section{Approximation for Systems of Fixpoint Equations}
\label{se:abstractions}

In this section we design a theory of approximation for systems of
fixpoint equations over complete lattices. The general setup is
borrowed from abstract
interpretation~\cite{cc:ai-unified-lattice-model,CC:SDPAF}, where a
concrete domain $C$ and an abstract domain $A$ are fixed. Semantic
operators on the concrete domain $C$ have a counterpart in the
abstract domain $A$, and suitable conditions can be imposed on such
operators to ensure that the least fixpoints of the abstract
operators
are sound
and/or complete approximations of the fixpoints of their concrete
counterparts.

Similarly, here we will have a system of equations
$\vec{x} =_{\vec{\eta}} \vec{f}^C(\vec{x})$ over a concrete
domain $C$ and its abstract counterpart
$\vec{x} =_{\vec{\eta}} \vec{f}^A(\vec{x})$ over an abstract
domain $A$, and we want that the solution of the latter
provides an approximation of the solution of the former.

Let us first focus on the case of a single equation. Let
$(C, \sqsubseteq)$ and $(A, \leq)$ be complete lattices and let
$f^C : C \to C$ and $f^A : A \to A$ be monotone functions. The fact
that $f^A$ is a sound (over)approximation of $f^C$ can be formulated
in terms of a concretisation function $\gamma : A \to C$, that maps
each abstract element $a \in A$ to a concrete element
$\gamma(a) \in C$, for which, intuitively, $a$ is an
overapproximation.  In the setting of abstract interpretation, where
the interest is for program semantics, typically expressed in terms of
least fixpoints, the desired \emph{soundness} property is
$\mu f^C \sqsubseteq \gamma(\mu f^A)$. A standard sufficient condition
for soundness
(see~\cite{cc:ai-unified-lattice-model,CC:SDPAF,Min:AIT}) is
$f^C \circ \gamma \sqsubseteq \gamma \circ f^A$.  The same condition
ensures soundness also for greatest fixpoints, i.e.,
$\nu f^C \sqsubseteq \gamma(\nu f^A)$, provided that $\gamma$ is
co-continuous and co-strict (see, e.g.,~\cite[Proposition 15]{CC:TLA},
which states the dual result). For clarity we state this result
explicitly in the appendix (see
Lemma~\ref{le:gamma}(\ref{le:gamma:1})).

\begin{toappendix}

\begin{lemma}[concretisation for single fixpoints]
  \label{le:gamma}
  Let $\gamma : A \to C$ be a monotone function.
  \begin{enumerate}
    
  \item \label{le:gamma:1}
    If 
    \begin{equation}
      \label{eq:gamma-sound}
      f_C \circ \gamma \sqsubseteq \gamma \circ f_A
    \end{equation}
    then $\mu f_C \sqsubseteq \gamma(\mu  f_A)$; if, in addition, $\gamma$ is
    co-continuous and co-strict $\nu f_C \sqsubseteq \gamma(\nu  f_A)$.
    
  \item \label{le:gamma:2}
    If 
    \begin{equation}
      \label{eq:gamma-complete}
      \gamma \circ f_A \sqsubseteq f_C \circ \gamma 
    \end{equation}
    then $\gamma(\nu f_A) \sqsubseteq \nu f_C$;  if, in addition, $\gamma$ is
    continuous and strict then
    $\gamma(\mu f_A) \sqsubseteq \mu f_C$.
  \end{enumerate}
\end{lemma}

\begin{proof}
  We focus on the soundness results since the completeness results
  follow by duality.
  
  For least fixpoint, we prove that for all ordinals $\beta$ we have
  $f_C^{\beta}(\bot_C) \leq \gamma(f_A^{\beta}(\bot_A))$, whence the
  thesis, since $\mu f_C = f_C^{\beta}(\bot_C)$ and
  $\mu f_A = f_A^{\beta}(\bot_A)$ for some ordinal $\beta$ (just take
  the largest of the ordinals needed to reach the two fixpoints).

  We proceed by transfinite induction:
  
  \begin{itemize}
  \item $(\beta=0)$ We have
    $f_C^{0}(\bot_C) = \bot_C \sqsubseteq \gamma(f_A^{0}(\bot_A))$, as
    desired.
            
  \item $(\beta \to \beta+1)$ Observe that
    \begin{align*}
      f_C^{\beta+1}(\bot_C)
      & = f_C(f_C^{\beta}(\bot_C))\\
      & \sqsubseteq f_C(\gamma(f_C^{\beta}(\bot_C)))
      & \mbox{[by ind. hyp. and monotonicity of $f_C$]}\\
      & \leq \gamma(f_A(f_A^{\beta}(\bot_A)))
      & \mbox{[by (\ref{eq:gamma-sound})]}\\
      & = \gamma(f_A^{\beta+1}(\bot_A))\\
    \end{align*}
      
  \item ($\beta$ limit ordinal) In this case
    \begin{align*}
      f_C^{\beta}(\bot_C)
      & = \lub_{\beta' < \beta} f_C^{\beta'}(\bot_C)\\
      & \sqsubseteq \lub_{\beta' < \beta} \gamma(f_A^{\beta'}(\bot_A))
      & \mbox{[by ind. hyp.]}\\
      & \sqsubseteq \gamma(\lub_{\beta' < \beta} f_A^{\beta'}(\bot_A))
      & \mbox{[by properties of joins]}\\
      & = \gamma(f_A^{\beta}(\bot_A))
    \end{align*}
      
  \end{itemize}

  For greatest fixpoints, we prove that for all ordinals $\beta$ we
  have $f_C^{\beta}(\top_C) \leq \gamma(f_A^{\beta}(\top_A))$, again
  by transfinite induction.

  \begin{itemize}
  \item $(\beta=0)$ We have
    $f_C^{0}(\top_C) = \top_C = \gamma(\top_A) =
    \gamma(f_A^{0}(\top_A))$, since $\gamma$ is assumed to be
    co-strict, hence we have the desired inequality.
   
  \item $(\beta \to \beta+1)$ Observe that
    \begin{align*}
      f_C^{\beta+1}(\top_C)
      & = f_C(f_C^{\beta}(\top_C))\\
      & \sqsubseteq f_C(\gamma(f_A^{\beta}(\top_A)))
      & \mbox{[by ind. hyp. and monotonicity of $f_C$]}\\
      & \sqsubseteq \gamma(f_A(f_A^{\beta}(\top_A)))
      & \mbox{[by (\ref{eq:gamma-sound})]}\\
      & = \gamma(f_A^{\beta+1}(\top_A))
    \end{align*}
   
 \item ($\beta$ limit ordinal) In this case
   \begin{align*}
     f_C^{\beta}(\top_C)
     & = \glb_{\beta' < \beta} f_C^{\beta'}(\top_C))\\
     & \sqsubseteq \glb_{\beta' < \beta} \gamma(f_A^{\beta'}(\top_A))
     & \mbox{[by ind. hyp.]}\\
     & = \gamma(\glb_{\beta' < \beta} f_A^{\beta'}(\top_A))
     & \mbox{[since $\gamma$ is co-continuous]}\\
     & = \gamma(f_A^{\beta}(\top_A))
   \end{align*}
   
 \end{itemize}
\end{proof}

We can get analogous results for abstractions, by duality.

\begin{lemma}[abstraction for single fixpoints]
  \label{le:alpha}
  Let $\alpha : C \to A$ be an abstraction function.
  \begin{enumerate}
    
  \item 
    If 
    \begin{equation}
      \label{eq:alpha-sound}
      \alpha \circ f_C \leq f_A \circ \alpha
    \end{equation}
    then $\alpha(\nu f_C) \leq \nu f_A$; if, in addition, $\alpha$ is
    continuous and strict $\alpha(\mu f_C) \leq \mu f_A$.
    
  \item 
    If 
    \begin{equation}
      \label{eq:alpha-complete}
      f_A \circ \alpha \leq \alpha \circ f_C
    \end{equation}
    then $\mu f_A \leq \alpha(\mu f_C)$;  if, in addition, $\alpha$ is
    co-continuous and co-strict then
    $\nu f_A \leq \alpha(\nu f_C)$.
  \end{enumerate}
\end{lemma}

\begin{lemma}[Galois insertions]
  \label{le:insertion}
  Let $f_C : C \to C$ and $f_A : A \to A$ be monotone functions and
  let $\langle \alpha, \gamma \rangle : C \to A$ be a Galois
  insertion.
  \begin{enumerate}
  \item
    \label{le:insertion:1}
    Assume soundness for $\alpha$ i.e.,
    (\ref{eq:alpha-sound}) (equivalent to soundness for $\gamma$, i.e.,
    (\ref{eq:gamma-sound})), and completeness for both $\alpha$ and
    $\beta$, i.e., (\ref{eq:alpha-complete}),
    (\ref{eq:gamma-complete}). Then
    \begin{center}
      $\alpha(\eta f_C) = \eta f_A$  for $\eta \in \{\mu,\nu\}$
      \qquad
      $\nu f_C = \gamma(\nu f_A)$
      \qquad
      $\mu f_C \sqsubseteq \gamma(\mu f_A)$
    \end{center}

  \item
    \label{le:insertion:2}
    Assume
    \begin{equation}
      \label{eq:strong-sound-complete}
      f_C = \gamma \circ f_A \circ \alpha
    \end{equation}
    then $\alpha(\eta f_C) = \eta f_A$ and
    $\eta f_C = \gamma(\eta f_A)$ for $\eta \in \{\mu,\nu\}$.
  \end{enumerate}
\end{lemma}

\begin{proof}
  \begin{enumerate}
  \item Just using Lemma~\ref{le:gamma} and Lemma~\ref{le:alpha}, we
    obtain
    \begin{center}
      (a) $\alpha(\mu f_C) = \mu f_A$ \quad
      (b) $\nu f_C = \gamma(\nu f_A)$ \quad
      (c) $\alpha(\mu f_C) \leq \mu f_A$ \quad
      (d) $\mu f_C \sqsubseteq \gamma(\mu f_A)$
    \end{center}
    From (b), applying $\alpha$, we obtain
    $\alpha(\nu f_C) = \alpha(\gamma(\nu f_A) = \nu f_A$, and we are
    done.

  \item In this case, from the assumption
    $f_C = \gamma \circ f_A \circ \alpha$ one can easily deduce the
    soundness and completeness conditions for $\alpha$ and $\gamma$,
    i.e., (\ref{eq:alpha-sound}), (\ref{eq:alpha-complete}),
    (\ref{eq:gamma-sound}), (\ref{eq:gamma-complete}). Therefore, by
    the previous point we get all desired inequalities but
    $\gamma(\mu f_A) \sqsubseteq \mu f_C$. For this observe that
    \begin{align*}
      \gamma(\mu f_A)
      & =  \gamma(\alpha(\mu f_C))
      & \mbox{[since $\mu f_A=\alpha(\mu f_C)$]}\\
      & = \gamma(\alpha(f_C(\mu f_C)))
      & \mbox{[since $\mu f_C$ is a fixpoint of $f_C$]}\\
      & = \gamma(\alpha(\gamma(f_A(\alpha(\mu f_C)))))
      & \mbox{[since $f_C = \gamma \circ f_A \circ \alpha$]}\\
      & = \gamma(f_A(\alpha(\mu f_C)))
      & \mbox{[since $\alpha \circ \gamma=id_A$]}\\
      & = f_C(\mu f_C)
      & \mbox{[since $f_C = \gamma \circ f_A \circ \alpha$]}\\
      & = \mu f_C
      & \mbox{[since $\mu f_C$ is a fixpoint of $f_C$]}      
    \end{align*}
  \end{enumerate}
\end{proof}
\end{toappendix}

Then we can suitably combine the conditions for least and greatest
fixpoints.  We will allow a different concretisation function for each
equation.

\begin{theoremrep}[sound concretisation for systems]
  \label{th:gamma-sys}
  Let $(C, \sqsubseteq)$ and $(A, \leq)$ be complete lattices, let
  $E_C$ of the kind $\vec{x} =_{\vec{\eta}} \vec{f}^C(\vec{x})$ and
  $E_A$ of the kind $\vec{x} =_{\vec{\eta}} \vec{f}^A(\vec{x})$ be
  systems of $m$ equations over $C$ and $A$, with solutions
  $\vec{s}^C \in C^m$ and $\vec{s}^A \in A^m$, respectively. Let
  $\vec{\gamma}$ be an $m$-tuple of monotone functions, with 
  $\gamma_i : A \to C$ for $i \in \interval{m}$.
  If $\vec{\gamma}$ satisfies
  $\vec{f}^C \circ \prd{\vec{\gamma}} \sqsubseteq \prd{\vec{\gamma}}
  \circ \vec{f}^A$ with $\gamma_i$ co-continuous and co-strict for
  each $i \in \interval{m}$ such that $\eta_i = \nu$, then
  $\vec{s}^C \sqsubseteq \prd{\vec{\gamma}}(\vec{s}^A)$.
\end{theoremrep}

\begin{proof}
  We proceed by induction on $m$. The case $m=0$
  is trivial.

  For the inductive case, consider systems with $m+1$ equations.
  Recall that, in order to solve the system, the last variable
  $x_{m+1}$ is considered as a fixed parameter $x$ and the system of
  $m$ equations that arises from dropping the last equation is
  recursively solved. This produces an $m$-tuple
  $\subvec{t^z}{1}{m}(x) = \sol{\subst{E_z}{x_{m+1}}{x}}$ parametric
  on $x$, for $z \in \{A,C\}$.
  For all $a \in A$, by inductive hypothesis applied to the systems
  $\subst{E_A}{x_{m+1}}{a}$ and
  $\subst{E_C}{x_{m+1}}{\gamma_{m+1}(a)}$ we obtain
  \begin{equation}
    \label{eq:sys-ind-hyp}
    \subvec{t^C}{1}{m}(\gamma_{m+1}(a))
    \sqsubseteq
    \prd{\subvec{\gamma}{1}{m}} (\subvec{t^A}{1}{m}(a))
  \end{equation}

  Inserting the parametric solution into the last equation, we get an
  equation in a single variable
  \begin{center}
    $a =_{\eta_m} f_{m+1}^A(\subvec{t^A}{1}{m}(a), a)$.
  \end{center}
  This equation can be solved by taking the corresponding fixpoint,
  i.e., if we define
  $f_A(a) = f_{m+1}^A(\subvec{t^A}{1}{m}(a), a)$, then
  $s^A_{m+1} = \eta_{m+1}f_A$.
  In the same way, $s^C_{m+1} = \eta_{m+1}f_C$ where
  $f_C(c) = f_{m+1}^C(\subvec{t^C}{1}{m}(c), c)$.

  Observe that
  $f_C \circ \gamma_{m+1} \sqsubseteq \gamma_{m+1} \circ f_A$. In fact
  \begin{align*}
    & f_C(\gamma_{m+1}(a))=\\
    & = f_{m+1}^C( \subvec{t^C}{1}{m}(\gamma_{m+1}(a)), \gamma_{m+1}(a)))
    & \mbox{[definition of $f_C$]}\\     
    & \sqsubseteq  f_{m+1}^C(\prd{\subvec{\gamma}{1}{m}}(\subvec{t^A}{1}{m}(a)), \gamma_{m+1}(a)))
    & \mbox{[by (\ref{eq:sys-ind-hyp})]}\\
    & \sqsubseteq  f_{m+1}^C(\prd{\vec{\gamma}}(\subvec{t^A}{1}{m}(a), a))
    & \mbox{[application of $\vec{\gamma}$]}\\
    & \sqsubseteq \gamma_{m+1} (f^A_{m+1}(\subvec{t^A}{1}{m}(a), a)) &
    \mbox{[hypothesis
      $\vec{f}^C \circ \prd{\vec{\gamma}} \sqsubseteq
      \prd{\vec{\gamma}}
      \circ \vec{f}^A$]}\\
    & = \gamma_{m+1}(f_A(a))
    & \mbox{[definition of $f_A$]}
  \end{align*}

  Therefore, recalling that when $\eta_{m+1}=\mu$ we are assuming
  co-continuity and co-strictness for $\gamma_{m+1}$, we can apply
  Lemma~\ref{le:gamma}(\ref{le:gamma:1}) and deduce that
  \begin{equation}
    \label{eq:sys-last}
    s^C_{m+1} = \eta_{m+1} f_C \sqsubseteq 
    \gamma_{m+1}(\eta_{m+1} f_A) = \gamma_{m+1}(s^A_{m+1})
  \end{equation}
  
  Finally, recall that the first $m$ components of the solutions are
  $\subvec{s^z}{1}{m} = \subvec{t^z}{1}{m}(s^z_{m+1})$ for
  $z \in \{C,A\}$. Therefore, exploiting (\ref{eq:sys-ind-hyp}), we
  have
  \begin{align*}
    & \subvec{s^C}{1}{m} =\\
    & = \subvec{t^C}{1}{m}(s^C_{m+1})\\
    & \sqsubseteq \subvec{t^C}{1}{m}(\gamma_{m+1}(s^A_{m+1})) & \mbox{[by (\ref{eq:sys-last})]}\\
    & \sqsubseteq \prd{\subvec{\gamma}{1}{m}}(\subvec{t^A}{1}{m}(s^A_{m+1}))
    & \mbox{[by (\ref{eq:sys-ind-hyp})]}\\
    & = \prd{\subvec{\gamma}{1}{m}}(\subvec{s^A}{1}{m})
  \end{align*}
  This concludes the inductive step.
\end{proof}

\begin{toappendix}
Everything can be dually formulated in terms of
abstraction functions.

\begin{theorem}[sound abstraction for systems]
  \label{th:alpha-sys}
  Let $(C, \sqsubseteq)$ and $(A, \leq)$ be complete lattices and let
  $E_C$ of the kind $\vec{x} =_{\vec{\eta}} \vec{f}^C(\vec{x})$ and
  $E_A$ of the kind $\vec{x} =_{\vec{\eta}} \vec{f}^A(\vec{x})$ be
  systems of $m$ equations over $C$ and $A$, with solutions
  $\vec{s}^C \in C^m$ and $\vec{s}^A \in A^m$, respectively. Let
  $\vec{\alpha}$ be an $m$-tuple of monotone functions, with
  $\alpha_i : C \to A$ for $i \in \interval{m}$.
  If $\vec{\alpha}$ satisfies
    \begin{equation*}
      \prd{\vec{\alpha}} \circ \vec{f}^C \leq \vec{f}^A \circ \prd{\vec{\alpha}}
    \end{equation*}
    with $\alpha_i$ continuous and strict for each $i \in \interval{m}$ such
    that $\eta_i = \mu$, then
    $\prd{\vec{\alpha}} (\vec{s}^C) \leq \vec{s}^A$.
\end{theorem}

\begin{proof}
  This follows from Lemma~\ref{th:gamma-sys} by duality.
\end{proof}
\end{toappendix}

The standard abstract interpretation framework of~\cite{CC:CCVTFP}
relies on Galois connections: concretisation functions $\gamma$ are
right adjoints, whose left adjoint, the abstraction function $\alpha$,
intuitively maps each concrete element in $C$ to its ``best''
overapproximation in $A$.
When $\langle \alpha, \gamma\rangle$ is a Galois connection,
$\alpha$ is automatically continuous and strict, while $\gamma$ is
co-continuous and co-strict. This leads to the following result,
where, besides the soundness conditions, we also make explicit the
completeness conditions.

\begin{theoremrep}[abstraction via Galois connections]
  \label{th:galois}
  Let $(C,\sqsubseteq)$ and $(A,\leq)$ be complete
  lattices, let $E_C$ of the kind
  $\vec{x} =_{\vec{\eta}} \vec{f}^C(\vec{x})$ and $E_A$ of the kind
  $\vec{x} =_{\vec{\eta}} \vec{f}^A(\vec{x})$ be systems of $m$
  equations over $C$ and $A$, with solutions $\vec{s}^C \in C^m$ and
  $\vec{s}^A \in A^m$, respectively. Let $\vec{\alpha}$ and
  $\vec{\gamma}$ be $m$-tuples of monotone functions, with
  $\langle \alpha_i , \gamma_i \rangle : C \to A$
  a Galois
  connection for each $i \in \interval{m}$.

  \begin{enumerate}
  \item
    \label{th:galois:1}
    \emph{Soundness}: If $\vec{\gamma}$ satisfies
    $\vec{f}^C \circ \prd{\vec{\gamma}} \sqsubseteq \prd{\vec{\gamma}}
    \circ \vec{f}^A$ or equivalently $\vec{\alpha}$ satisfies
    $\prd{\vec{\alpha}} \circ \vec{f}^C \leq \vec{f}^A \circ
    \prd{\vec{\alpha}}$,
    then $\prd{\vec{\alpha}} (\vec{s}^C) \leq \vec{s}^A$ (equivalent to
      $\vec{s}^C \sqsubseteq \prd{\vec{\gamma}}(\vec{s}^A)$).

  \item
    \label{th:galois:2}
    \emph{Completeness (for abstraction)}: If $\vec{\alpha}$ satisfies
    $\vec{f}^A \circ \prd{\vec{\alpha}} \leq \prd{\vec{\alpha}} \circ
    \vec{f}^C$ with $\alpha_i$ co-continuous and co-strict for each
    $i \in \interval{m}$ such that $\eta_i = \nu$, then
    $\vec{s}^A \leq \prd{\vec{\alpha}} (\vec{s}^C)$.

  \item
    \label{th:galois:3}
    \emph{Completeness (for concretisation)}: If $\vec{\gamma}$
    satisfies
    $\prd{\vec{\gamma}} \circ \vec{f}^A \sqsubseteq \vec{f}^C \circ
    \prd{\vec{\gamma}}$ with $\gamma_i$ continuous and strict for each
    $i \in \interval{m}$ such that $\eta_i = \mu$, then
    $\prd{\vec{\gamma}}(\vec{s}^A) \sqsubseteq \vec{s}^C$.
  \end{enumerate}
\end{theoremrep}

\begin{proof}
  Due to Theorems~\ref{th:gamma-sys} and~\ref{th:alpha-sys} (and the
  fact that we can apply the theorems to lattices with reversed
  order), the only thing to prove is that the conditions
  $\prd{\vec{\alpha}} \circ \vec{f}^C \leq \vec{f}^A \circ
  \prd{\vec{\alpha}}$ and
  $\vec{f}^C \circ \prd{\vec{\gamma}} \sqsubseteq \prd{\vec{\gamma}}
  \circ \vec{f}^A$ are equivalent.
  If we assume
  $\prd{\vec{\alpha}} \circ \vec{f}^C \leq \vec{f}^A \circ
  \prd{\vec{\alpha}}$, by definition of Galois connection, we get
  $ \vec{f}^C \sqsubseteq \prd{\vec{\gamma}} \circ \vec{f}^A \circ
  \prd{\vec{\alpha}}$. Now, post-composing with $\prd{\vec{\gamma}}$
  and exploiting the fact that
  $\prd{\vec{\alpha}} \circ \prd{\vec{\gamma}} \sqsubseteq
  \prd{\vec{id}}$ we obtain
  \begin{center}
    $\vec{f}^C \circ \prd{\vec{\gamma}} \sqsubseteq \prd{\vec{\gamma}} \circ
    \vec{f}^A \circ \prd{\vec{\alpha}} \circ \prd{\vec{\gamma}} \sqsubseteq
    \vec{\gamma} \circ \vec{f}^A $
  \end{center}
  as desired.

  The converse implication is analogous.
\end{proof}

Completeness for the abstraction, i.e.,
$\vec{s}^A \leq \prd{\vec{\alpha}}(\vec{s}^C)$, together with
soundness, leads to $\prd{\vec{\alpha}} (\vec{s}^C) = \vec{s}^A$. This
is a rare but very pleasant situation in which the abstraction does
not lose any information as far as the abstract properties are
concerned.
We remark that here the notion of ``completeness'' slightly deviates from
the standard abstract interpretation terminology where soundness is
normally indispensable, and thus complete abstractions
(see, e.g.,~\cite{GRS:MAIC}) are, by default, also
sound.

Moreover, completeness for the concretisation is normally of
limited interest in abstract interpretation.
Alone, it states that the abstract
solution is an underapproximation of the concrete one, while
typically the interest is for overapproximations. Together with
soundness, it leads to $\vec{s^C} = \prd{\vec{\gamma}}(\vec{s^A})$, a
very strong property which is not meaningful in program analysis.
In our case, keeping the concepts of soundness and completeness
separated and considering also completeness for the concretisation is
helpful in some cases, especially when dealing with up-to functions,
which are designed to provide underapproximations of fixpoints.

\begin{toappendix}

For Galois insertions, we make explicit a very special case where we
get rid of all the (co-)continuity and (co-)strictness requirements,
and get soundness and completeness both for the abstraction and the
concretisation.

\begin{lemma}[Galois insertions for systems]
  \label{le:insertion-sys}
  Let $(C, \sqsubseteq)$ and $(A, \leq)$ be complete lattices, let
  $E_C$ of the kind $\vec{x} =_{\vec{\eta}} \vec{f}^C(\vec{x})$ and
  $E_A$ of the kind $\vec{x} =_{\vec{\eta}} \vec{f}^A(\vec{x})$ be
  systems of $m$ equations over $C$ and $A$, with solutions
  $\vec{s}^C \in C^m$ and $\vec{s}^A \in A^m$, respectively. Let
  $\vec{\alpha}$ and $\vec{\gamma}$ be $m$-tuples of abstraction and
  concretisation functions, with
  $\langle \alpha_i , \gamma_i \rangle : C \to A$ a Galois insertion
  for each $i \in \interval{m}$.
  If
  \begin{equation}
    \label{eq:strong-sound-complete-sys}
    \vec{f}_C = \prd{\vec{\gamma}} \circ \vec{f}^A \circ \vec{\alpha}
  \end{equation}
  then $\prd{\vec{\alpha}} (\vec{s}^C) = \vec{s}^A$ and
    $\vec{s}^C = \prd{\vec{\gamma}}(\vec{s}^A)$.
\end{lemma}

\end{toappendix}

As in the standard abstract interpretation framework, dealing with
Galois connections, we can consider the best (smallest) sound
abstraction of the concrete system in the abstract domain.

\begin{definition}[best abstraction]
  \label{de:best-abs}
  Let $(C, \sqsubseteq)$ and $(A, \leq)$ be complete lattices, let
  $E_C$ be a system of $m$ equations over $C$ of the kind
  $\vec{x} =_{\vec{\eta}} \vec{f}(\vec{x})$. Let $\vec{\alpha}$ and
  $\vec{\gamma}$ be $m$-tuples of monotone functions, with
  $\langle \alpha_i , \gamma_i \rangle : C \to A$ a Galois
  connection for each $i \in \interval{m}$.  The \emph{best abstraction} of
  $E_C$ is the system over $A$ defined by
  $\vec{x} =_{\vec{\eta}} \vec{f}^\# (\vec{x})$, where
  $\vec{f}^\# = \prd{\vec{\alpha}} \circ \vec{f} \circ \prd{\vec{\gamma}}$.
\end{definition}

Standard arguments shows that $\vec{f}^\#$ is a
sound abstraction of $\vec{f}$ over $A$, and it is the smallest one.

Moreover, sound abstract operators can be obtained compositionally out
of basic ones, preserving soundness.

\begin{example}[abstraction for the $\mu$-calculus]
  \label{ex:abstraction-mu}
  The paper~\cite{lgsbb:property-preserving-abstractions} observes
  that (bi)simulations over transition systems can be seen as Galois
  connections and interpreted as abstractions. Then it characterises
  fragments of the $\mu$-calculus which are preserved and strongly
  preserved by the abstraction.
  We next discuss how this can be derived as an instance of our framework.

  Let $T_C = (\mathbb{S}_C, \to_C)$ and $T_A = (\mathbb{S}_A,\to_A)$
  be transition systems and let
  $\langle \alpha, \gamma \rangle : \Pow{\mathbb{S}_C} \to
  \Pow{\mathbb{S}_A}$ be a Galois connection. It is a
  \emph{simulation}, according
  to~\cite{lgsbb:property-preserving-abstractions}, if it satisfies
  the following condition:
  $\alpha \circ \semdia_{T_C} \circ \gamma \subseteq \semdia_{T_A}$.
  In this case $T_A$ is called a
  $\langle \alpha, \gamma \rangle$-\emph{abstraction} of $T_C$,
  written $T_C \sqsubseteq_{\langle \alpha, \gamma \rangle} T_A$. This
  can be shown to be equivalent to the ordinary notion of simulation
  between transition systems~\cite[Propositions~9 and
  10]{lgsbb:property-preserving-abstractions}. In particular, if
  $R \subseteq \mathbb{S}_C \times \mathbb{S}_A$ is a simulation in
  the ordinary sense then one can consider
  $\langle \semdia_{R^{-1}}, \sembox_R \rangle : \Pow{\mathbb{S}_C}
  \to \Pow{\mathbb{S}_A}$, where $\semdia_{R^{-1}}$ is the function
  $\semdia_{R^{-1}}(X) = \{ y \in \mathbb{S}_A \mid \exists x \in X.\
  (x,y) \in R \}$. This is a Galois connection (in the abstract
  interpretation setting $\semdia_{R^{-1}}$ and $\sembox_R$ are often
  denoted $\widetilde{\mathit{pre}}_R$ and $\mathit{post}_R$,
  respectively~\cite{Cous:PCAFC})
  inducing a simulation in the above
  sense, i.e.,
  $\semdia_{R^{-1}} \circ \semdia_{T_C} \circ \sembox_R \subseteq
  \semdia_{T_A}$.

  When $T_C \sqsubseteq_{\langle \alpha, \gamma \rangle} T_A$,
  by~\cite[Theorem~2]{lgsbb:property-preserving-abstractions},
  one has that
  $\alpha$ ``preserves'' the
  $\mu\Diamond$-calculus, i.e., the
  fragment of the $\mu$-calculus without $\Box$ operators.
  More precisely,
  for any formula $\varphi$ of the $\mu\Diamond$-calculus,
  we have
  $\alpha(\sem[T_C]{\varphi}_\rho) \subseteq
  \sem[T_A]{\varphi}_{\alpha \circ \rho}$. This means that for each
  $s_C \in \mathbb{S}_C$, if $s_C$ satisfies $\varphi$ in the
  concrete system, then all the states in $\alpha(\{s_C\})$ satisfy
  $\varphi$ in the abstract system, provided that each proposition $p$
  is interpreted in $A$ with $\alpha(\rho(p))$, the abstraction of its interpretation in
  $C$.

  This can be obtained as an easy consequence of
  Theorem~\ref{th:galois}, where we use the same function $\alpha$ as
  an abstraction for all equations. The condition $\alpha \circ
  \semdia_{T_C} \circ \gamma \subseteq \semdia_{T_A}$ 
  above can be rewritten as
  $\alpha \circ \semdia_{T_C} \subseteq \semdia_{T_A} \circ \alpha$
  which is the soundness
  condition~($\prd{\vec{\alpha}} \circ \vec{f}^C \leq \vec{f}^A \circ
  \prd{\vec{\alpha}}$) in Theorem~\ref{th:galois} for the semantics
  of the diamond operator. For the other operators the soundness
  condition is trivially shown to hold. In fact,
  \begin{itemize}
  \item for $\true$ and $\false$ we have
    $\alpha(\emptyset) = \emptyset$ and
    $\alpha(\mathbb{S}_C) \subseteq \mathbb{S}_A$;

  \item for $\land$ and $\lor$ we have
    $\alpha(X \cup Y) = \alpha(X) \cup \alpha (Y)$ and
    $\alpha(X \cap Y) \subseteq \alpha(X) \cap \alpha(Y)$;

  \item a proposition $p$ represents the constant function $\rho(p)$ in
    $T_C$ and $\alpha(\rho(p))$ in $T_A$.
  \end{itemize}

  In order to extend the logic by including negation on propositions,
  in~\cite{lgsbb:property-preserving-abstractions}, an additional
  condition is required, called \emph{consistency} of the abstraction
  with respect to the interpretation: 
  $\alpha(\rho(p)) \cap \alpha(\overline{\rho(p)}) = \emptyset$, for all $p$.
  This is easily seen to be equivalent to
  $\alpha(\overline{\rho(p)}) \subseteq \overline{\alpha(\rho(p))}$
  which is
  the soundness
  condition~($\prd{\vec{\alpha}} \circ \vec{f}^C \leq \vec{f}^A \circ
  \prd{\vec{\alpha}}$)  in Theorem~\ref{th:galois} for negated propositions.

  Our theory naturally suggests generalisations of~\cite{lgsbb:property-preserving-abstractions}.
  E.g.,
  by (the dual of) Theorem~\ref{th:gamma-sys}, continuity and strictness of the abstraction $\alpha$ are sufficient to retain the results, hence one could deal with an abstraction
  not being an adjoint, thus going beyond ordinary simulations.
\end{example}

\begin{example}[abstraction for {\L}ukasiewicz $\mu$-terms]
  \label{ex:reals-reprise}
  For {\L}ukasiewicz $\mu$-terms, as introduced in
  Example~\ref{ex:reals}, leading to systems of fixpoint equations
  over the reals, we can consider as an abstraction a form of
  discretisation: for some fixed $n$ define the abstract domain
  $[0,1]_{/n} = \{0\} \cup \{ k/n \mid k \in \interval{n}\}$ and the
  insertion
  $\langle \alpha_n, \gamma_n \rangle : [0,1] \to [0,1]_{/n}$ with
  $\alpha_n$ defined by $\alpha_n(x) = \ceil{n\cdot x}/n$ and $\gamma_n$
  the inclusion. We can consider for all operators $op$, their best
  abstraction $op^\# = \alpha_n \circ op \circ \prd{\vec{\gamma_n}}$, thus
  getting a sound abstraction.

  Note that for all semantic operators, $op^\#$ is
  the restriction
  of $op$ to the abstract domain, with the exception of
  $r \cdot^\# x = \alpha_n(r \cdot x)$ for $x \in [0,1]_{/n}$.
  Moreover,
  for $x, y \in [0,1]$ we have
  \begin{itemize}    
  \item
    $\alpha_n(\mathbf{0}(x)) = \mathbf{0}^\#(\alpha_n(x))$,
    $\alpha_n(\mathbf{1}(x)) = \mathbf{1}^\#(\alpha_n(x))$;

  \item
    $\alpha_n( r \cdot x) \leq r \cdot^\# \alpha_n(x)$;

  \item 
    $\alpha_n(x \sqcup y) = \alpha_n(x) \sqcup^\# \alpha_n(y)$,
    $\alpha_n(x \sqcap y) = \alpha_n(x) \sqcap^\# \alpha_n(y)$;

  \item
    $\alpha_n(x \oplus y) \leq \alpha_n(x) \oplus^\# \alpha_n(y)$,
    $\alpha_n(x \odot y) \leq \alpha_n(x) \odot^\# \alpha_n(y)$
    since $\alpha_n(x+y) \leq \alpha_n(x) + \alpha_n(y)$
  \end{itemize}
  i.e., the abstraction is complete for $\mathbf{0}$, $\mathbf{1}$,
  $\sqcup$, $\sqcap$, while it is just sound for the remaining
  operators.

  For instance, the system in Example~\ref{ex:reals} can be shown to
  have solution $x_1 = x_2 = 0.2$. With abstraction $\alpha_{10}$ we
  get $x_1 = x_2 = 0.8$, with a more precise abstraction
  $\alpha_{100}$ we get $x_1 = x_2 = 0.22$ and with
  $\alpha_{1000}$ we get $x_1 = x_2 = 0.201$.
\end{example}

\begin{example}[abstraction for {\L}ukasiewicz $\mu$-calculus]
  Although space limitations prevent a detailed discussion, observe
  that when dealing with {\L}ukasiewicz $\mu$-calculus over some
  probabilistic transition system $N = (\mathbb{S}, \to)$, we can lift
  the Galois insertion above to $[0,1]^{\mathbb{S}}$. Define
  $\alpha_n^\to : [0,1]^{\mathbb{S}} \to [0,1]_{/n}^{\mathbb{S}}$ by
  letting, $\alpha^\to_n(v) = \alpha_n \circ v$ for
  $v \in [0,1]^{\mathbb{S}}$. Then
  $\langle \alpha^\to_n, \gamma^\to_n\rangle : [0,1]^{\mathbb{S}} \to
  [0,1]_{/n}^{\mathbb{S}}$, where $\gamma^\to_n$ is the inclusion, is
  a Galois insertion and, as in the previous case, we can consider the
  best abstraction for the operators of the {\L}ukasiewicz
  $\mu$-calculus.

  For instance, consider the system for $\phi'$ in
  Example~\ref{ex:lukasievicz-modal}. Recall that the exact solution is $x_2(a)=0.25$. With abstraction $\alpha_{10}$
  we get $x_2(a) = 0.3$, with $\alpha_{15}$ we
  get $x_2(a) = 0.2\bar{6}$.
\end{example}

\section{Up-To Techniques}
\label{se:up-to-from-abstraction}

Up-to techniques have been shown effective in easing the proof of
properties of greatest fixpoints. Originally proposed for coinductive
behavioural equivalences~\cite{Mil:CC,SM:PBUT}, they have been later
studied in the setting of complete
lattices~\cite{p:complete-lattices-up-to,Pou:CAWU}.
Some recent work~\cite{bggp:sound-up-to-complete-abstract} started the
exploration of the relation between up-to techniques and abstract
interpretation. Roughly, they work in a setting where the semantic
function of interest $f^* : L \to L$ admits a left adjoint
$f_* : L \to L$, the intuition being that $f^*$ and $f_*$ are
predicate transformers mapping a condition into, respectively, its
strongest postcondition and weakest precondition. Then complete
abstractions for $f^*$ and sound up-to functions for $f_*$ are shown
to coincide.
This
has a natural interpretation in our game theoretic framework,
as discussed in~\S\ref{ss:algorithmic-view}.

Here we take another view. We work with general semantic functions
and, in \S\ref{ss:up-to-abstraction}, we first argue that up-to
techniques can be naturally interpreted as abstractions where the
concretisation is complete (and sound, if the up-to function is a
closure).
Then, in \S\ref{ss:up-to-systems} we can smoothly
extend up-to techniques from a single fixpoint to systems of fixpoint
equations.

\subsection{Up-To Techniques as Abstractions}
\label{ss:up-to-abstraction}

The general idea of up-to techniques is as follows. Given a monotone
function $f : L \to L$ one is interested in the greatest fixpoint
$\nu f$. In general, the aim is to establish whether some given
element of the lattice $l \in L$ is under the fixpoint, i.e., if
$l \sqsubseteq \nu f$.
In turn, since by Tarski's Theorem,
$\nu f = \lub \{ x \mid x \sqsubseteq f(x) \}$, this amounts to
proving that $l$ is under some post-fixpoint $l'$, i.e.,
$l \sqsubseteq l'\sqsubseteq f(l')$. For instance, consider the
function $\mathit{bis}_T : \Rel{\mathbb{S}} \to \Rel{\mathbb{S}}$ for
bisimilarity on a transition system $T$ in
Example~\ref{ex:bisimilarity}. Given two states
$s_1, s_2 \in \mathbb{S}$, proving
$\{ (s_1, s_2) \} \subseteq \nu \mathit{bis}_T$, i.e., showing the two
states bisimilar, amounts to finding a post-fixpoint, i.e., a relation
$R$ such that $R \subseteq \mathit{bis}_T(R)$ (namely, a bisimulation)
such that $\{ (s_1, s_2) \} \subseteq R$.

\begin{definition}[up-to function]
  Let $L$ be a complete lattice and let $f : L \to L$ be a monotone
  function. A \emph{sound up-to function} for $f$ is any monotone
  function $u : L \to L$ such that
  $\nu (f \circ u) \sqsubseteq \nu f$. It is called \emph{complete} if also
  the converse inequality $\nu f \sqsubseteq \nu (f \circ u)$ holds.
\end{definition}

When $u$ is sound, if $l$ is a post-fixpoint of $f \circ u$, i.e., $l \sqsubseteq f(u(l))$ we have
$l \sqsubseteq \nu (f \circ u) \sqsubseteq \nu f$.
The idea is that the characteristics of $u$ should make it easier to
prove that $l$ is a postfix-point of $f \circ u$ than proving that it
is for $f$. This is clearly the case
when $u$ is
extensive. In fact by extensiveness of $u$ and monotonicity of $f$ we
get $f(l) \sqsubseteq f (u(l))$ and thus obtaining
$l \sqsubseteq f(u(l))$ is ``easier'' than obtaining
$l \sqsubseteq f(l)$.
Note that extensiveness also implies
``completeness'' of the up-to function: since $f \sqsubseteq f \circ u$
clearly $\nu f \sqsubseteq \nu (f \circ u)$.
We remark that for up-to functions, since the interest is for underapproximating fixpoints, the terms soundness and completeness are somehow reversed with respect to their meaning in abstract interpretation.

A common sufficient condition
ensuring
soundness of up-to
functions is compatibility~\cite{p:complete-lattices-up-to}.

\begin{definition}[compatibility]
  \label{def:compatibility}
  Let $L$ be a complete lattice and let $f : L \to L$ be a monotone
  function. A monotone function $u : L \to L$ is \emph{$f$-compatible} if
  $u \circ f \sqsubseteq f \circ u$.
\end{definition}

The soundness of
an $f$-compatible up-to function $u$ can be proved by viewing it as an
abstraction. When $u$ is a closure (i.e., extensive and idempotent),
$u(L)$ is a complete lattice that can be seen as an abstract domain in
a way that $\langle u, i \rangle : L \to u(L)$, with $i$ being the
inclusion, is a Galois insertion. Moreover $f_{|u(L)}$ can be
shown to provide an abstraction of both $f$ and
$f \circ u$ over $L$, sound and complete with respect to the inclusion
$i$, seen as the concretisation.
The formal details are given below. Since we later aim
to apply up-to techniques to
systems of equations, we
deal with not only greatest but also least fixpoints.  

\begin{lemmarep}[compatible up-to functions as sound and complete abstractions]
  \label{le:up-to-closure}
  Let $f : L \to L$ be a monotone function and let $u : L \to L$ be an
  $f$-compatible closure. Consider the Galois insertion
  $\langle u, i \rangle : L \to u(L)$ where $i : u(L) \to L$ is the inclusion. Then
  \begin{enumerate}
  \item $f$ restricts to $u(L)$, i.e., $f_{|u(L)}\colon u(L)\to u(L)$;
  \item $\nu f = i(\nu f_{|u(L)}) = \nu (f\circ u)$.
    If $u$ is continuous and strict then
    $\mu f = i(\mu f_{|u(L)}) = \mu (f\circ u)$.
  \end{enumerate}
  \begin{center}
    \begin{tikzcd}[ampersand replacement=\&]  
      L \arrow[loop above, "f" very near start]
        \arrow[loop left, "f \circ u"]
        \arrow[r, "u" swap, shift right=1ex]
      \&
      u(L)
      \arrow[l, "i" swap, shift right=1ex]
      \arrow[loop right, "f_{|u(L)}"]
    \end{tikzcd}
  \end{center}
\end{lemmarep}

\begin{proof}
  
  \begin{enumerate}
  \item We have that for all $l \in u(L)$, the $f$-image $f(l) \in
    u(L)$. Let $l \in u(L)$, i.e., $l = u(l')$ for some $l'\in
    L$. Observe that
    \begin{align*}
      f(l) 
        & \sqsubseteq u(f(l)) & \mbox{[by extensiveness]}\\
        & \sqsubseteq f(u(l))  & \mbox{[by compatibility]}\\
        & = f(u(u(l')))\\
        & = f(u(l'))
        & \mbox{[by idempotency]}\\
        & = f(l)
    \end{align*}
    Hence $f(l) = u(f(l))$, which means that $f(l) \in u(L)$.

  \item We first prove that $\nu f = \nu f_{|u(L)}$.  Consider
    \begin{center}
      \begin{tikzcd}
        L \arrow[loop below,"f" below] \arrow[r, "\alpha =
        u" above, shift
        left=1ex] & u(L) \arrow[l, "\gamma = \mathit{i}" below, shift left=1ex]
        \arrow[loop below,"f|_{u(L)}" below]
      \end{tikzcd}
    \end{center}
    Note that for all $l \in u(L)$, we have
    $f (\gamma(l)) = f(l) = \gamma(f_{|u(L)}(l))$, i.e., $\gamma$
    satisfies soundness (\ref{eq:gamma-sound}) and completeness
    (\ref{eq:gamma-complete}) in Lemma~\ref{le:gamma}.
    Therefore,
    $\nu f = \gamma (\nu f_{|u(L)}) = \eta f_{|u(L)}$, as
    desired. 

    \medskip

    Next we prove that $\nu (f \circ u) = \nu f_{|u(L)}$
    Consider
    \begin{center}
      \begin{tikzcd}
        L \arrow[loop below,"f\circ u" below] \arrow[r, "\alpha =
        u" above, shift
        left=1ex] & u(L) \arrow[l, "\gamma = \mathit{i}" below, shift left=1ex]
        \arrow[loop below,"f_{|u(L)}" below]
      \end{tikzcd}
    \end{center}

    Again, for all $l \in u(L)$, we have
    $f \circ u (\gamma(l)) = f(u(l)) = f(l) = \gamma(f_{|u(L)}(l))$,
    i.e., $\gamma$ satisfies soundness (\ref{eq:gamma-sound}) and
    completeness (\ref{eq:gamma-complete}) in Lemma~\ref{le:gamma}.
    Therefore,
    $\nu (f \circ u) = \gamma (\nu f_{|u(L)}) = \nu f_{|u(L)}$, as
    desired. 

    \medskip

    Finally, if $u$ is continuous and strict then also
    $\gamma=i$ is so: First, since $\bot=u(\bot)\in u(L)$ and hence
    the inclusion $i$ maps $\bot$ to $\bot$. Second, since $u$ is
    continuous, directed suprema in both lattices coincide: let
    $D\subseteq u(L)$, then $\bigsqcup D = \bigsqcup \{u(d)\mid d\in
    D\} = u(\bigsqcup D)\in u(L)$. Hence $i$ preserves directed
    suprema.

    Hence we get the previous results also for least fixpoints.
  \end{enumerate}
\end{proof}

When the up-to function is just $f$-compatible (hence sound), but
possibly not a closure, we canonically turn $u$ into an $f$-compatible
closure (hence sound and complete) by taking the least closure
$\bar{u}$ above $u$.

\begin{definition}[least upper closure]
  \label{de:lcu}
  Let $L$ be a complete lattice and let $u : L \to L$ be a monotone
  function. We let $\bar{u} : L \to L$ be the function defined by
  $\bar{u}(x) = \mu (\hat{u}_x)$ where  $\hat{u}_x(y) = u(y) \sqcup x$. 
\end{definition}

\begin{lemmarep}[properties of $\bar{u}$]
  \label{le:extension}
  Let $u : L \to L$ be a monotone function. Then
  \begin{enumerate}
  \item \label{le:extension:1}\label{le:extension:2} \label{le:extension:3}
    $\bar{u}$ is the least closure larger than $u$;
    
  \item \label{le:extension:4}
    if $u$ is $f$-compatible then $\bar{u}$ is;
    
  \item \label{le:extension:5}
    if $u$ is continuous and strict then $\bar{u}$ is.
  \end{enumerate}
\end{lemmarep}

\begin{proof}
  \begin{enumerate}
  \item We first observe that $\bar{u}$ is a closure. For
    extensiveness, just observe that
    $\hat{u}_x(y) = u(y) \sqcup x \sqsupseteq x$ for all $y \in L$ and
    thus obviously $\bar{u}(x) = \mu (\hat{u}_x) \sqsupseteq x$.
    
    In order to show that $\bar{u}$ is idempotent, note that, by
    extensiveness, $\bar{u} \sqsubseteq \bar{u} \circ \bar{u}$. Hence
    to conclude, we just need to prove the converse inequality
    $\bar{u} \circ \bar{u} \sqsubseteq \bar{u}$. For all $x \in L$, we
    have
    $\bar{u}(\bar{u}(x)) = \mu (\hat{u}_{\bar{u}(x)}) =
    \hat{u}_{\bar{u}(x)}^\gamma$ for some ordinal $\gamma$.
    We prove, by transfinite induction that for all $\alpha$, that
    $\hat{u}_{\bar{u}(x)}^\alpha \sqsubseteq \bar{u}(x)$.

    \medskip

    ($\alpha=0$) We have that
    $\hat{u}_{\bar{u}(x)}^0 = \bot \sqsubseteq \bar{u}(x)$.

    \medskip

    ($\alpha \to \alpha+1$) We have that
    \begin{align*}
      \hat{u}_{\bar{u}(x)}^{\alpha+1} 
      & = \hat{u}_{\bar{u}(x)}(\hat{u}_{\bar{u}(x)}^{\alpha})\\
      & = u(\hat{u}_{\bar{u}(x)}^{\alpha}) \sqcup \bar{u}(x)  
      & \mbox{[by def. $\hat{u}_{\bar{u}(x)}$ ]}\\
      & \sqsubseteq u(\bar{u}(x)) \sqcup \bar{u}(x) 
      & \mbox{[by ind. hyp.]}\\
      & \sqsubseteq \hat{u}_x(\bar{u}(x)) \sqcup \bar{u}(x)  
      & \mbox{[since $u \sqsubseteq \hat{u}_x$]}\\
      & = \bar{u}(x) \sqcup \bar{u}(x)  
      & \mbox{[since $\hat{u}_x(\bar{u}(x)) = \bar{u}(x)$]}\\
      & = \bar{u}(x) 
    \end{align*}
    
    \medskip

    ($\alpha$ limit) We have that
    \begin{align*}
      \hat{u}_{\bar{u}(x)}^{\alpha}
      & = \lub_{\beta<\alpha} \hat{u}_{\bar{u}(x)}^{\beta}\\
      & \sqsubseteq \lub_{\beta<\alpha} \bar{u}(x)  & \mbox{[by ind. hyp.]}\\
      & = \bar{u}(x)
    \end{align*}
    
    Moreover, $\bar{u}$ is larger than $u$, i.e.,
    $u \sqsubseteq \bar{u}$. In fact,
    \begin{align*}
      \bar{u}(x) 
      & = \hat{u}_x(\bar{u}(x)) & \mbox{[since $\bar{u}(x)$ is a fixpoint of $\hat{u}_x$]}\\
      & = u(\bar{u}(x)) \sqcup x & \mbox{[by def. of  $\hat{u}_x$]}\\
      & \sqsupseteq u(x) \sqcup x & \mbox{[since $\bar{u}$ is extensive]}\\
      & \sqsupseteq u(x)
    \end{align*}

    \medskip

    Finally, let $v$ any closure such that $u \sqsubseteq v$. We show
    that for all $x \in L$, $\hat{u}_x^\alpha \sqsubseteq v(x)$,
    whence $\bar{u}(x) \sqsubseteq v(x)$, as desired.

        ($\alpha=0$) We have that
    $\hat{u}_{\bar{u}(x)}^0 = \bot \sqsubseteq v(x)$.

    \medskip

    ($\alpha \to \alpha+1$) We have that
    \begin{align*}
      \hat{u}_{\bar{u}(x)}^{\alpha+1} 
      & = \hat{u}_{x}(\hat{u}_{x}^{\alpha})\\
      & = u(\hat{u}_{x}^{\alpha}) \sqcup x
      & \mbox{[by def. $\hat{u}_{x}$ ]}\\
      & \sqsubseteq u(v(x)) \sqcup x
      & \mbox{[by ind. hyp.]}\\
      & \sqsubseteq v(v(x)) \sqcup x
      & \mbox{[since $u \sqsubseteq v$]}\\
      & = v(x) \sqcup x
      & \mbox{[by idempotency of $v$]}\\
      & = v(x)
      & \mbox{[by extensiveness of $v$]}
    \end{align*}
    
    \medskip

    ($\alpha$ limit) We have that
    \begin{align*}
      \hat{u}_{x}^{\alpha}
      & = \lub_{\beta<\alpha} \hat{u}_{x}^{\beta}\\
      & \sqsubseteq \lub_{\beta<\alpha} v(x) & \mbox{[by ind. hyp.]}\\
      & = v(x)
    \end{align*}

  \item Observe that for all $x \in L$, we have
    $\bar{u}(f(x)) = \hat{u}_{f(x)}^\gamma$ for some ordinal
    $\gamma$. Hence also here we proceed by transfinite induction,
    showing that for all $\alpha$
    \begin{center}
      $\hat{u}_{f(x)}^\alpha \sqsubseteq f (\bar{u}(x))$
    \end{center}

    ($\alpha=0$) We have that
    $\hat{u}_{f(x)}^0 = \bot \sqsubseteq f(\bar{u}(x))$.

    \medskip

    ($\alpha \to \alpha+1$) We have that
    \begin{align*}
      \hat{u}_{f(x)}^{\alpha+1} 
      & = \hat{u}_{f(x)} (\hat{u}_{f(x)}^{\alpha})\\
      & \sqsubseteq \hat{u}_{f(x)}(f(\bar{u}(x))  & \mbox{[by ind. hyp.]}\\
      & = u(f(\bar{u}(x))) \sqcup f(x) & \mbox{[by def. of $\hat{u}_{f(x)}$]}\\
      & \sqsubseteq f(u(\bar{u}(x))) \sqcup f(x) & \mbox{[by compatibility of $f$]}\\
      & \sqsubseteq f(u(\bar{u}(x)) \sqcup x) & \mbox{[by general properties of $\sqcup$]}\\
      & = f(\hat{u}_x(\bar{u}(x)))) & \mbox{[by def. of $\hat{u}_x$]}\\
      & = f(\bar{u}(x)) & \mbox{[since $\hat{u}(x)$ is a fixpoint]}
    \end{align*}
    
    \medskip

    ($\alpha$ limit) We have that
    \begin{align*}
      \hat{u}_{f(x)}^{\alpha} 
      & = \lub_{\beta<\alpha} \hat{u}_{f(x)}^{\beta} \\
       & \sqsubseteq \lub_{\beta<\alpha} f (\bar{u}(x))  & \mbox{[by ind. hyp.]}\\
      & = f (\bar{u}(x))
    \end{align*}

  \item Assume that $u$ is continuous and strict. Then $\hat{u}_x$ is
    continuous for all $x \in L$. In fact, for each directed set $D \subseteq L$
    we have
    \begin{align*}
      \hat{u}_x(\lub D)
      & = u(\lub D) \sqcup x\\
      & = \lub \{ u(d) \mid d \in D\}) \sqcup x\\
      & = \lub \{ u(d) \sqcup x \mid d \in D\})\\
      & = \lub \{ \hat{u}_x(d) \mid d \in D\})\\
    \end{align*}

    Now, we can show that $\bar{u}$ is continuous. Let $D \subseteq L$
    be a directed set. We have to prove that
    $\bar{u}(\lub D) = \lub_{d \in D} \bar{u}(d)$. It is sufficient to
    prove that
    $\bar{u}(\lub D) \sqsubseteq \lub_{d \in D} \bar{u}(d)$, as the
    other inequality follows by monotonicity and general properties of
    $\lub$.
    As usual, we recall that
    $\bar{u}(\lub D) = \hat{u}_{\lub D}^\gamma$ for some $\gamma$ and
    thus show, by transfinite induction on $\alpha$ that
    \begin{center}
      $\hat{u}_{\lub D}^\alpha \sqsubseteq \lub_{d \in D} \bar{u}(d)$.
    \end{center}

    ($\alpha=0$) We have that
    $\hat{u}_{\lub D}^0 = \bot \sqsubseteq \lub_{d \in D} \bar{u}(d)$.

    \medskip

    ($\alpha \to \alpha+1$) We have that
    \begin{align*}
      \hat{u}_{\lub D}^{\alpha+1} 
      & = \hat{u}_{\lub D} (\hat{u}_{\lub D}^{\alpha})\\
      & \sqsubseteq \hat{u}_{\lub D}(\lub_{d \in D} \bar{u}(d))  & \mbox{[by ind. hyp.]}\\
      & = \lub_{d \in D}  \hat{u}_{\lub D}(\bar{u}(d))  & \mbox{[by continuity of $\hat{u}_{\lub D}$]}\\
      & = \lub_{d \in D} (u(\bar{u}(d)) \sqcup \lub D) & \mbox{[by def. of $\hat{u}_{\lub D}$]}\\
      & \sqsubseteq \lub_{d \in D} (\hat{u}_d(\bar{u}(d)) \sqcup \lub D) & \mbox{[since $u \sqsubseteq \hat{u}_d$]}\\
      & = \lub_{d \in D} (\bar{u}(d) \sqcup \lub D) & \mbox{[since $\hat{u}(d)$ is a fixpoint]}\\
      & = \lub_{d \in D} (\bar{u}(d) \sqcup d)\\
      & = \lub_{d \in D} \bar{u}(d) & \mbox{[by extensiveness of $\bar{u}$]}
    \end{align*} 
    
    \medskip

    ($\alpha$ limit) We have that
    \begin{align*}
      \hat{u}_{\lub D}^{\alpha} 
      & = \lub_{\beta<\alpha} \hat{u}_{\lub D}^{\beta} \\
       & \sqsubseteq \lub_{\beta<\alpha} \lub_{d \in D} \bar{u}(d)  & \mbox{[by ind. hyp.]}\\
      & = \lub_{d \in D} \bar{u}(d) 
    \end{align*}

    Furthermore, $\bar{u}$ is strict since
    $\hat{u}_\bot(\bot) = u(\bot) \sqcup \bot = \bot \sqcup \bot =
    \bot$, and thus $\bar{u}(\bot) = \mu (\hat{u}_\bot) = \bot$.
   
  \end{enumerate}
\end{proof}

The least upper closure above a given function has
been considered already in~\cite{CC:CCLR}, with a slightly different
construction.

Using Lemmas~\ref{le:up-to-closure} and~\ref{le:extension},
whenever $u$ is a compatible up-to function for $f$, we have that
$\bar{u}$ is a sound and complete up-to function for $f$.
The soundness of $u$ then immediately follows.

\begin{corollary}[soundness of compatible up-to functions]
  \label{co:up-to}
  Let $f : L \to L$ be a monotone function, let $u : L \to L$ be an
  $f$-compatible up-to function
  and let $\bar{u}$ be the least closure above $u$. Then
  $\nu (f \circ u) \sqsubseteq \nu (f \circ \bar{u}) = \nu
  f$. If $u$ is continuous and strict, then
  $\mu (f \circ u) \sqsubseteq \mu (f \circ \bar{u}) = \mu f$.
\end{corollary}

In~\cite{p:complete-lattices-up-to} the proof of soundness of a compatible up-to technique $u$ relies on the definition of a function $u^\omega$ defined as $u^\omega(x) = \bigsqcup \{ u^n(x) \mid n \in \nat\}$, where $u^n(x)$ is defined inductively as $u^0(x) = x$ and $u^{n+1}(x) = u(u^n(x))$. The function $u^\omega$ is extensive but not idempotent in general, and it can be easily seen that $u^\omega \sqsubseteq \bar{u}$.
The paper~\cite{Pou:CAWU} shows that for any monotone
function one can consider the largest compatible up-to function, the
so-called companion, which is extensive and idempotent. The companion
could be used in place of $\bar{u}$ for part of the theory. However,
we find it convenient to work with $\bar{u}$ since, despite not
discussed in the present paper, it plays a key role for the
integration of up-to techniques into the verification
algorithms. Furthermore the companion is usually hard to determine.

\subsection{Up-To Techniques for Systems of Equations}
\label{ss:up-to-systems}

Exploiting the view of up-to functions as
abstractions, moving to systems of equations is easy.
As in the case of abstractions, a different up-to
function is allowed for each equation.

\begin{definition}[compatible up-to for systems of equations]
  \label{de:compatible-systems}
  Let $(L, \sqsubseteq)$  be a complete lattice and
  let $E$ be $\vec{x} =_{\vec{\eta}} \vec{f}(\vec{x})$, a
  system of $m$ equations
  over $L$. A \emph{compatible tuple of up-to functions} for $E$ is an
  $m$-tuple of monotone functions $\vec{u}$, with $u_i : L \to L$,
  satisfying compatibility
  ($\prd{\vec{u}} \circ \vec{f} \sqsubseteq \vec{f} \circ \prd{\vec{u}}$)
  with $u_i$ continuous and strict for each $i \in \interval{m}$ such
  that $\eta_i = \mu$.
\end{definition}

We can then generalise Corollary~\ref{co:up-to} to systems of
equations.

\begin{theoremrep}[up-to for systems]
  \label{th:up-to-sys}
  Let $(L, \sqsubseteq)$ be a complete lattice and let $E$ be
  $\vec{x} =_{\vec{\eta}} \vec{f}(\vec{x})$, a system of $m$ equations
  over $L$, with solution $\vec{s} \in L^m$. Let
  $\vec{u}$ be a compatible tuple of up-to functions for $E$
  and let $\bar{\vec{u}} = (\bar{u}_1, \ldots, \bar{u}_m)$ be the
  corresponding tuple of least
  closures.
  Let $\vec{s}'$ and
  $\bar{\vec{s}}$ be the solutions of the systems
  $\vec{x} =_{\vec{\eta}} \vec{f}(\prd{\vec{u}}(\vec{x}))$
  and
  $\vec{x} =_{\vec{\eta}} \vec{f}(\prd{\bar{\vec{u}}}(\vec{x}))$,
  respectively.
  Then $\vec{s}' \sqsubseteq \bar{\vec{s}} = \vec{s}$. Moreover, if
  $\vec{u}$ is extensive then $\vec{s}' = \vec{s}$.
\end{theoremrep}

\begin{proof}
  Immediate extension to systems of the proofs of the
  Lemma~\ref{le:up-to-closure} and
  Corollary~\ref{co:up-to}, exploiting Theorem~\ref{th:gamma-sys}.
\end{proof}

\begin{example}[$\mu$-calculus up-to (bi)similarity]
  \label{ex:up-to-bisim}
  Consider the problem of model-checking the $\mu$-calculus over some
  transition system with atoms $T = (\mathbb{S}, \to, A)$.

  Assuming that we have an a priori knowledge about the similarity
  relation $\precsim$ over some of the states in $T$, then,
  restricting to a suitable fragment of the $\mu$-calculus we can
  avoid checking the same formula on similar states. This intuition
  can be captured in the form of an up-to technique, that we refer to
  as up-to similarity. It is based on an up-to function
  $\vec{u}_\precsim : \Pow{\mathbb{S}} \to \Pow{\mathbb{S}}$ defined,
  for $X \in \Pow{\mathbb{S}}$, by
  $u_\precsim(X) = \{s \in \mathbb{S}\,\mid\, \exists s' \in X.\ s' \precsim s \}.$

  Function $u_\precsim$ is monotone, extensive,
  and idempotent. It is also continuous and strict.
  
  Moreover, $u_\precsim$ is a compatible (and thus sound) up-to
  function for the $\mu\Diamond$-calculus where propositional
  variables are interpreted as atoms.
  In fact, $\precsim$ is a simulation (the largest one) and
  the function $u_\precsim$ is the associated abstraction as defined
  in Example~\ref{ex:abstraction-mu}, namely
  $u_\precsim = \semdia_{\succsim}$.
  Therefore, compatibility
  $u_\precsim \circ f \sqsubseteq f \circ u_\precsim$ corresponds to
  condition
  $\alpha \circ \semdia_{T_C} \circ \gamma \subseteq \semdia_{T_A}$ in
  Example~\ref{ex:abstraction-mu} which has been already observed to
  coincide with soundness in the sense of Theorem~\ref{th:galois} for
  the operators of the $\mu\Diamond$-calculus.
  Concerning propositional variables, in
  Example~\ref{ex:abstraction-mu}, they were interpreted, in the
  target transition system, by the abstraction of their interpretation
  in the source transition system. Since here we have a single
  transition system and a single interpretation
  $\rho : \Prop \to \Pow{\mathbb{S}}$, we must have
  $\rho(p) = u_\precsim (\rho(p))$, i.e., $\rho(p)$
  upward-closed with respect to $\precsim$. This automatically holds
  by the fact that $\precsim$ is a simulation.

  \medskip

  Similarly, we can define up-to bisimilarity via the up-to
  function $u_\sim(X) = \{s \in \mathbb{S}\,\mid\, \exists s' \in X.\ s \sim s' \}$.
  As above, one can see that
  compatibility $u_\sim \circ f \sqsubseteq f \circ u_\sim$ holds for
  the full $\mu$-calculus with propositional variables interpreted
  as atoms.
  For instance, consider the formula $\varphi$ in Example~\ref{ex:mu}
  and the transition system in Fig.~\ref{fi:running-ts}. Using the
  up-to function $u_{\sim}$ corresponds to working in the bisimilarity
  quotient in Fig.~\ref{fi:running-ts-bis}. Note, however, that when
  using a local algorithm (see
  \S\ref{ss:algorithmic-view})
  the quotient does not need to be actually computed. Rather, only the
  bisimilarity over the states explored by the searching procedure is
  possibly exploited.
\end{example}

\begin{example}[bisimilarity up-to transitivity]
  \label{ex:up-to-trans}
  Consider the problem of checking bisimilarity on a transition system
  $T = \langle \mathbb{S}, \to \rangle$. A number of well-known sound
  up-to techniques have been introduced in the
  literature \cite{ps:enhancements-coinductive}. As an example,
  we consider the up-to function
  $u_{\mathit{tr}} : \Rel{\mathbb{S}} \to \Rel{\mathbb{S}}$ performing
  a single step of transitive closure. It is defined as:
  \[u_{\mathit{tr}}(R) = R \circ R =
  \{(x,y)\,\mid\, \exists\, z \in \mathbb{S}.\ (x,z) \in R\ \land\ (z,y) \in R\}.\]

  It is easy to see that $u_{\mathit{tr}}$ is monotone and compatible
  with respect to the function
  $\mathit{bis}_T : \Rel{\mathbb{S}} \to \Rel{\mathbb{S}}$ of
  which bisimilarity is the greatest fixpoint (see
  Example~\ref{ex:bisimilarity}). Since $A$ is deterministic,
  bisimilarity coincides with language equivalence.

  Note that $u_{\mathit{tr}}$ is neither idempotent nor extensive. The
  corresponding closure $\bar{u}_{\mathit{tr}}$ maps a relation to its (full) transitive closure (this is known to be
  itself a sound up-to technique, a fact that we can also derive from
  the compatibility of $u_{\mathit{tr}}$ and
  Corollary~\ref{co:up-to}).
\end{example}

\section{Solving Systems of Equations via Games}
\label{se:powerset-game}

In this section, we first provide a characterisation of the
solution of a system of fixpoint equations over a complete lattice in terms of a parity game.
This generalises a result in~\cite{BKMP:FPCL}. While the
original result was limited to continuous lattices, here, exploiting
the results on abstraction in \S\ref{se:abstractions}, we devise a
game working for any complete lattice.

The game characterisation opens the way to the development of
algorithms for solving the game and thus the associated verification
problem.


\subsection{Game Characterization}
\label{se:game-characterization}

We show that the solution of a system of equations over a complete
lattice can be characterised using
a parity game.

\begin{table}
  \small
  \begin{center}
    \begin{tabular}{l|c|l}
      Position & Player & Moves \\ \hline 
      $(b,i)$ & $\exists$ & $\vec{X}$
                            s.t.
                            $b \sqsubseteq f_i(\lub \vec{X})$ \\[1mm]
      $\vec{X}$ & $\forall$ & $(b',j)$ s.t. $b' \in X_j$
    \end{tabular}
  \end{center}
  \caption{The game on the powerset of the basis}
  \label{tab:powerset-game}
\end{table}

\begin{definition}[powerset game]
  \label{de:powerset-game}
  Let $L$ be a complete lattice with a basis $B_L$. Given
  a system $E$ of $m$ equations over $L$ of the kind
  $\vec{x} =_{\vec{\eta}} \vec{f} (\vec{x})$, the corresponding
  \emph{powerset game} is a parity game, with an existential player $\exists$
  and a universal player $\forall$, defined as follows:

  \begin{itemize}
  \item The positions of $\exists$ are pairs $(b, i)$ where
    $b \in B_L$, $i \in \interval{m}$. Those of $\forall$ are
    tuples of subsets of the basis $\vec{X} = (X_1, \ldots, X_m) \in (\Pow{B_L})^m$.

  \item From position $(b, i)$ the
    moves of $\exists$ are 
    $\Emoves{b,i} = \{ \vec{X} \mid \vec{X} \in (\Pow{B_L})^m\ \land\ 
    b \sqsubseteq f_i(\lub \vec{X})\}$.

  \item From position $\vec{X} \in (\Pow{B_L})^m$ the
    moves of $\forall$ are
    $\Amoves{\vec{X}} = \{ (b, i) \mid i \in \interval{m}\ \land\ b
    \in X_i \}$.
  \end{itemize}
  The game is schematised in Table~\ref{tab:powerset-game}.
  For a finite play, the winner is the player who moved last.
  For an infinite play, let $h$ be the highest index that
  occurs infinitely often in a pair $(b, i)$. If $\eta_h = \nu$ then
  $\exists$ wins, else $\forall$ wins.
\end{definition}

If we instantiate the game to the setting of standard
$\mu$-calculus model-checking, we obtain an alternative encoding of
$\mu$-calculus into parity games, typically resulting in more compact
games.

\begin{example}
  We provide a simple example illustrating the game. Consider the
  infinite lattice $L=\nat\cup \{\omega,\omega+1\}$ (where
  $n\le \omega\le \omega+1$ for every $n\in\nat$) with basis
  $B_L = L$.  Furthermore let $f\colon L\to L$ be a monotone
  function with $f(n) = n+1$ for $n\in\nat$ and $f(\omega)=\omega$,
  $f(\omega+1)=\omega+1$. Hence $\mu f = \omega$.

  We set $b=\omega$ and attempt to show via the game that
  $b\le \mu f$, by exhibiting a winning strategy for $\exists$. Note
  that since we are dealing with a $\mu$-equation, in order to win
  $\exists$ must ensure that $\forall$ eventually has no moves left.
  Since there is only one fixpoint equation, we omit the indices.
  Starting with $b=\omega$, $\exists$ plays $X = \nat$, which is a
  valid move since $\omega \le f(\lub X) = f(\omega)$. Now $\forall$
  has to pick some $n\in X$. In the next move, $\exists$ can play
  $X = \{n-1\}$, which means that $\forall$ picks $n-1$. Hence we
  obtain a descending chain, leading to $1$, which can be covered by
  $\exists$ by choosing $X = \emptyset$, since
  $1\le f(\lub \emptyset) = f(0)$. Now $\forall$ has no moves left
  and $\exists$ wins.

  Instead for $b = \omega+1\not\le \mu f$, $\exists$ has
  no winning strategy since she has to play a set $X$ that contains
  $\omega+1$. Then player $\forall$ can reply by choosing $\omega+1$
  and the game will continue forever. This is won by $\forall$ since
  we are dealing with a $\mu$-equation.
\end{example}

Interestingly, the correctness and completeness of the game can be proved by exploiting the results in
\S\ref{se:abstractions}. The crucial observation is
that there is a Galois insertion between $L$ and the powerset lattice
of its basis (which is algebraic hence continuous)
$\langle \alpha , \gamma \rangle : \Pow{B_L} \to L$
where abstraction $\alpha$ is the join $\alpha(X) = \lub X$ and
concretisation $\gamma$ takes the lower cone
$\gamma(l) = \cone{l} \cap B_L$.
Then a system of equations over a complete lattice $L$ can be ``transferred''
to a system of equations over the powerset of the basis $\Pow{B_L}$
along such insertion, in a way that
the system in $L$ can be seen as a sound and complete abstraction of
the one in $\Pow{B_L}$.

\begin{theoremrep}[correctness and completeness]
  \label{th:game-corr-comp}
  Let $E$ be a system of $m$ equations over a complete lattice $L$ of
  the kind $\vec{x} =_{\vec{\eta}} \vec{f} (\vec{x})$ with solution
  $\vec{s}$. For all $b \in B_L$ and $i \in \interval{m}$,
  $b \sqsubseteq s_i$ iff $\exists$ has a winning strategy from
  position $(b,i)$.
\end{theoremrep}

\begin{proof}
  Define
  $\langle \alpha , \gamma \rangle : \Pow{B_L} \to L$,
  by letting $\alpha(X) = \lub X$ for $X \in \Pow{B_L}$ and
  $\gamma(l) = \cone{l} \cap B_L$ for $l \in L$.
  It is immediate to see that this is a Galois insertion: for all
  $X \in \Pow{B_L}$ we have
  $X \subseteq \gamma(\alpha(X)) = (\cone{\bigsqcup X}) \cap B_L$
  and, for $l \in L$ we have
  $l = \alpha(\gamma(l)) = \bigsqcup (\cone{l} \cap B_L)$.

  Below we abuse the notation and write $\downarrow$ and $\lub$ for
  the $m$-tuples where each function is $\downarrow$ and $\lub$
  applied componentwise, respectively.
  \begin{center}
    \begin{tikzcd}[column sep=huge]
      ((\Pow{B_L})^m,\subseteq) \arrow[loop below,"{\vec{f}^C}=\cone\, {\vec{f}} \lub" below]
      \arrow[r,  "\vec{\alpha}=\lub\_" below, shift right=1ex]
      &
      L^m \arrow[l,  "\vec{\gamma}=\cone{\,\,\_} \cap B_L" above, shift right=1ex]
      \arrow[loop below,"\vec{f}" below]
    \end{tikzcd}
  \end{center}
    
  Define a ``concrete'' system
  $\vec{x} =_{\vec{\eta}} \vec{f}^C(\vec{x})$ where
  $\vec{f}^C = \prd{\vec{\gamma}} \circ \vec{f} \circ
  \prd{\vec{\alpha}} : (\Pow{B_L})^m \to (\Pow{B_L})^m$. Then we can
  use Lemma~\ref{le:insertion-sys} to deduce that, if we denote by
  $\vec{S}^C$ the solution of the concrete system and by $\vec{s}$ the
  solution of the original system, we have
  $\vec{S}^C = \cone{\vec{s}} \cap B_L^m$.

  Now, $(\Pow{B_L}, \subseteq)$ is an algebraic, hence continuous
  lattice. Therefore, by~\cite[Theorem~4.8]{BKMP:FPCL}, the lattice
  game for the ``concrete'' system on $(\Pow{B_L})^m$ is sound and
  complete.

  It is immediate to realise that, if we fix as basis for
  $\Pow{B_L}$ the set of singletons, this corresponds exactly to what
  we called here the powerset game.
  In fact, the game aims to show that
  $\{ b \} \subseteq S_i^C = \cone{s_i}$, for some $b \in B_L$ and
  $i \in \interval{m}$, and this amounts to $b \sqsubseteq s_i$.
  Positions of $\exists$ are pairs $(\{b\}, i)$ where $b \in B_L$ and
  $i \in \interval{m}$, and she has to play some tuples
  $\vec{X} \in (\Pow{B_L})^m$ such that
  $\{b\} \subseteq f^C_i(\vec{X}) =
  \cone{f_C(\lub \vec{X})}$ which amounts to
  $b \sqsubseteq f_C(\lub \vec{X})$.
  Positions of $\forall$ are tuples $\vec{X} \in (\Pow{B_L})^m$ and he
  chooses some $j \in \interval{m}$ and $b' \in X_j$. This is exactly
  the powerset game, hence we conclude.
\end{proof}

\subsection{An Algorithmic View}
\label{ss:algorithmic-view}

The game theoretical characterisation can be the basis for the
development of algorithms, possibly integrating abstraction and up-to
techniques, for solving systems of equations.
Here we consider local
algorithms for the case of a single equation.
Our main focus is to provide a general procedure which transcends the verification problem at hand, and also takes advantage of heuristics based on abstractions and up-to techniques.
This allows us also to establish a link with
some recent work relating abstract interpretation and up-to
techniques~\cite{bggp:sound-up-to-complete-abstract} and exploiting
up-to techniques for computing language equivalence on
NFAs~\cite{BP:NFA}.
While not improving the complexity bounds, our algorithm is still in line with other local algorithms designed for specific settings, such as \cite{BP:NFA,h:proving-up-to,h:mise-oeuvre-preuves-bisim}, as they arise as proper instantiations.

An algorithm for general systems is considerably more difficult and
the description of such an algorithm will be postponed to
\S\ref{se:on-the-fly}.  We first focus on the special case of a single
(greatest) fixpoint equation $x =_\nu f(x)$.

\subsubsection{Selections}

For a practical use of the game it can be useful to observe
that the set of moves of the existential player can be suitably restricted
without affecting the completeness of the game, by introducing a
notion of selection, similarly to what is done in~\cite{BKMP:FPCL}.

Given a lattice $L$, define a preorder $\sqsubseteq_H$ on $\Pow{B_L}$
by letting, for $X, Y \in \Pow{B_L}$, $X \sqsubseteq_H Y$ if
$\bigsqcup X \sqsubseteq \bigsqcup Y$. (The subscript $H$ comes from the fact that
for completely distributive lattices, if $B_L$ is the set of
irreducible elements, then $\sqsubseteq_H$ is the
``Hoare preorder''~\cite{AJ:DT}, requiring that $\forall x \in X.\, \exists y \in Y.\, x \sqsubseteq y$.)
Observe that $\sqsubseteq_H$ is not antisymmetric. We write $\equiv_H$
for the corresponding equivalence, i.e., $X \equiv_H Y$ when
$X \sqsubseteq_H Y \sqsubseteq_H X$.

The moves of player $\exists$ can be ordered by the pointwise
extension of $\sqsubseteq_H$, thus leading to the following
definition. Since we deal with a single equation, we will omit the
indices from the positions of player $\exists$ and write $b$ instead
of $(b,1)$.

\begin{definition}[selection]
  \label{de:selection}
  Let
  $x =_\nu f(x)$ be an equation over a complete lattice $L$, with basis $B_L$. A \emph{selection} is a
  function $\sigma : B_L \to \Pow{\Pow{B_L}}$ such that for all
  $b \in B_L$ it holds
  $\filtersub{\sigma(b)}{H} = \Emoves{b}$,
  i.e.\ the set of moves of $\exists$ from position $b$,
  where $\filtersub{}{H}$ is the upward-closure with respect to
  $\sqsubseteq_H$.
\end{definition}

This is equivalent to requiring that $\sigma(b)\subseteq \Emoves{b}$
and for each $X\in\Emoves{b}$ there exists $Y \in \sigma(b)$ such that
$\bigsqcup Y\sqsubseteq \bigsqcup X$.

For the case of a single fixpoint equation it is easy to see that
Theorem~\ref{th:game-corr-comp} continues to hold if we restrict the
moves of player $\exists$ to those prescribed by a selection.

\begin{theoremrep}[game with selections]
  \label{th:game-corr-comp-selections}
  Let $x =_\nu f(x)$ be an equation over a complete lattice $L$ with
  solution $s$. For all $b \in B_L$, it holds that
  $b \sqsubseteq s$
  iff
  $\exists$ has a winning
  strategy from position $b$ in the game restricted to selections.
\end{theoremrep}

\begin{proof}
  Assume that $\exists$ has a winning strategy in the original game:
  given $b$ she would play $X$, where all $b'\in \Amoves{X}$ are
  winning positions.

  Instead, in the game restricted by selections, she might only be
  able to play $Y$ where $\bigsqcup Y \sqsubseteq \bigsqcup X$. Now
  $\forall$ picks $b'\in Y$.  By construction
  $b' \sqsubseteq \bigsqcup X$. Since all elements of $X$ are winning
  positions in the original game (and hence below the solution), $b'$
  is also a winning position and we can continue.  Now either
  $\exists$ wins directly or the game continues forever, giving us a
  winning strategy in the restricted game.
\end{proof}

\subsubsection{Local Algorithm for a Special Case}
\label{sec:on-the-fly-special}

In this section we assume that $f: L \to L$ is some fixed function that preserves non-empty
meets, i.e., for $X \neq \emptyset$,
$f(\bigsqcap X) = \bigsqcap f(X)$.
This is equivalent to
asking $f(x) = f^*(x) \sqcap c$ for some $c \in L$ (just take
$c=f(\top)$), with $f^*$ being a right adjoint of a map $f_*$, a
setting that has been studied also
in~\cite{bggp:sound-up-to-complete-abstract}.
We will call a function satisfying this assumption a \emph{deterministic function}.
Note that the adjunction $\langle f_*,f^*\rangle$ is completely
orthogonal to the adjunctions (Galois connections) studied so far.

\begin{example}
\label{ex:deterministic}
For a simple example adopted from \cite{BP:NFA}, consider a deterministic finite automaton
$A = (Q,\Sigma,\delta,F)$, where $Q$ is a finite set of states, $\Sigma$
is a finite alphabet, $\delta : Q\times \Sigma\to Q$ is the
transition function and $F\subseteq Q$ is the set of final states.
Since $A$ is deterministic, language equivalence coincides with
bisimilarity. Consider the
lattice of relations $L = (\Pow{Q \times Q}, s\sqsubseteq)$ with basis
$B_L = \{\{(q_1,q_2)\} \mid q_1,q_2\in Q\}$. The behaviour map, having
bisimilarity as largest fixpoint, is
$f : \Pow{Q \times Q} \to \Pow{Q \times Q}$ defined as
$f(R) = f^*(R)\cap C$ where
$f^*(R) = \{(q_1,q_2) \mid \forall a \in \Sigma.
\, (\delta(q_1,a),\delta(q_2,a))\in R\}$ with
$C = \{(q_1,q_2) \mid q_1\in F \iff q_2 \in F \}$. The left adjoint is $f_*(R) = \{(\delta(q_1,a),\delta(q_2,a)) \mid (q_1,q_2)\in R, a\in\Sigma\}$. 

Given two states $q_1, q_2 \in R$,  we want to decide whether $(q_1, q_2) \in S$, where $S$ is bisimilarity, the solution of the greatest fixpoint equation $R =_\nu f(R)$.
\end{example}

We first observe that for deterministic functions we can take
 a very simple selection.

\begin{lemmarep}[selection]
  \label{le:simple-sel}
  Let $L$ be a complete lattice with basis $B_L$, and let $f : L \to L$ be a deterministic function, i.e., $f(x) = f^*(x) \sqcap c$ for some $c \in
  L$ and $\langle f_*, f^* \rangle: L \to L$ a Galois connection.
  A selection $\sigma : B_L \to \Pow{\Pow{B_L}}$ for $x =_\nu f(x)$ can be defined, for
  $b \in B_L$, as:
  \begin{center}
    $\sigma(b) = \left\{
      \begin{array}{ll}
        \mbox{$\{ X \}$ with $X \subseteq B_L$ s.t. $X \equiv_H \cone{f_*(b)} \cap B_L$} & \mbox{when $b \sqsubseteq c$}\\
        \emptyset & \mbox{otherwise}
      \end{array}
    \right.
    $
  \end{center}
\end{lemmarep}

\begin{proof}
  In order to see that this is a selection, note that if
  $b \sqsubseteq c$ then given $X \subseteq B_L$ it holds that
  $X \in \Emoves{b}$ (i.e.,
  $b \sqsubseteq f(\bigsqcup X) = f^*(\bigsqcup X) \sqcap c$) iff
  $b \sqsubseteq f^*(\bigsqcup X)$ iff
  $f_*(b) \sqsubseteq \bigsqcup X$, where the last step is§ motivated
  by adjointness.
\end{proof}

Observe that there might be several choices for $X\subseteq B_L$: one
that always works is $X = \cone{f_*(b)} \cap B_L$, but subsets
$X\subseteq \cone{f_*(b)} \cap B_L$ are also feasible, as long as 
$\bigsqcup X = f_*(b)$.
In Example~\ref{ex:deterministic}, given $\{(q_1, q_2)\} \in B_L$, we
can define
$\sigma(\{(q_1,q_2)\}) = \{\{ \{ (q_1', q_2') \} \mid (q_1', q_2') \in
f_*(\{(q_1,q_2)\})\}\} = \{ \{ \{ (\delta(q_1,a), \delta(q_2,a)) \}
\mid a \in \Sigma\}\}$.

By Lemma~\ref{le:simple-sel}, in the game for $x =_\nu f(x)$, either
the existential player is stuck or she has a best move.
As a consequence, the game
in
\S\ref{se:game-characterization} can be
simplified.
Let $B_L$ be any basis for $L$ such that $\bot \notin B_L$.
The moves of player $\exists$ are deterministic, governed by
$\sigma$, and only player $\forall$ has choices when exploring the
elements included in such moves.

For checking whether $b \sqsubseteq \nu f$, for some $b \in
B_L$, the game starts from position $b$.
Then, at a generic position $b'$, we do the following:
\begin{enumerate}
\item if $b' \not\sqsubseteq c$ then $\sigma(b') = \emptyset$ and $\exists$ loses;
\item otherwise, $\exists$ has to play the only element in $\sigma(b') = \{ X \}$
  \begin{enumerate}
  \item if $f_*(b') = \bot$ then take $X = \emptyset$; hence $\exists$ wins since $\forall$ has no moves;
  \item if instead $f_*(b') \neq \bot$, we can take $X \equiv_H \cone{f_*(b')} \cap B_L$ and thus player $\forall$ can play any $b'' \in X$ and the game continues.
  \end{enumerate}
\end{enumerate}
Player $\exists$ wins the game iff no losing
position for her ($b' \not\sqsubseteq c$) is
encountered in the exploration. When a losing position for $\exists$ is encountered we immediately know that $\forall$ wins.

The game can be further simplified by observing that, if $W$ denotes
the set of positions already visited during the exploration,
whenever, at a
position $b'$, we have $b' \sqsubseteq \bigsqcup W$ then $\exists$
wins from $b'$ as long as she wins from all the positions in $W$.
This leads to the local algorithm outlined in
List.~\ref{fi:simple-game-general}, whose proof of correctness formalises the
arguments above.
The procedure $\texttt{Explore}$ allows to check if $b \sqsubseteq \nu f = \nu (f^* \sqcap c)$ by invoking \texttt{Explore($b$,$\emptyset$)}, which returns \texttt{true} if and only if player $\exists$ wins in the simplified game.

\begin{lstlisting}[caption={Local algorithm for the simplified game.}, label={fi:simple-game-general}, mathescape]
Explore($b'$,$W$):
  if $b' \not\sqsubseteq c$ then return false;
  else if $b' \sqsubseteq \bigsqcup W$ then return true;
  else take $X \subseteq B_L$ s.t. $X \equiv_H \cone{f_*(b')} \cap B_L$;
       return $\wedge_{b'' \in X}$ Explore($b''$,$W \cup \{b'\}$);
\end{lstlisting}

\begin{theoremrep}[correctness and completeness of the simplified game]
  Let $L$ be a complete lattice with basis $B_L \subseteq L \setminus \{\bot\}$, and let $f : L \to L$ be a deterministic function, i.e., $f(x) = f^*(x) \sqcap c$ for some $c \in
  L$ and $\langle f_*, f^* \rangle: L \to L$ a Galois connection.
  Then, for all $b \in B_L$, $b \sqsubseteq \nu f$ iff the invocation \textnormal{\texttt{Explore($b$,$\emptyset$)}} returns \textnormal{\texttt{true}}. 
\end{theoremrep}

\begin{proof}
  First we prove that if $b \sqsubseteq \nu f$, then the algorithm in
  Fig.~\ref{fi:simple-game-general} determines that $\exists$ wins
  the simplified game. Observe that by monotonicity of $f$,
  we have that $b \sqsubseteq \nu f = f(\nu f) \sqsubseteq f(\top)$,
  so $\exists$ does not immediately lose. Moreover, let $X \subseteq
  B_L$ such that $X \equiv_H \cone{f_*(b)} \cap B_L$, hence $\lub X =
  f_*(b)$. Since $b \sqsubseteq \nu f$, by monotonicity of $f_*$ we
  have $\lub X = f_*(b) \sqsubseteq f_*(\nu f) = f_*(f(\nu f)) =
  f_*(f^*(\nu f) \sqcap f(\top)) \sqsubseteq f_*(f^*(\nu f))
  \sqsubseteq \nu f$ because of the properties of the Galois connection. Since $\lub X \sqsubseteq \nu f$ we must have that $b' \sqsubseteq \nu f$ for all $b' \in X$, therefore the same argument as before holds on all $b' \in X$ as well, and so no position losing for $\exists$ can ever be reached, hence $\exists$ wins.

  Now we prove that if $\exists$ wins starting from $b_0 = b$, then $b
  \sqsubseteq \nu f$. Actually, we show that if $\exists$ wins the
  simplified game according to the local algorithm (Fig.~\ref{fi:simple-game-general}) then she wins also the general fixpoint game for the single fixpoint equation $x =_\nu f(x)$, and so by Theorem~\ref{th:game-corr-comp} we know that $b \sqsubseteq \nu f$. Since $\exists$ wins, for every path $(b_0,b_1,\ldots)$ in the tree of positions explored we have three possible cases:
  \begin{itemize}
  \item the path is infinite, thus for all $i$, $b_i \sqsubseteq f(\top)$ and $b_{i+1} \in X_i$ for some $X_i \subseteq B_L$ such that $X_i \equiv_H \cone{f_*(b_i)} \cap B_L$, hence $\lub X_i = f_*(b_i)$. Then, for all $i$, observe that by the Galois connection we have $b_i = b_i \sqcap f(\top) \sqsubseteq f^*(f_*(b_i)) \sqcap f(\top) = f^*(\lub X_i) \sqcap f(\top) = f(\lub X_i)$. This means that, for all $i$, $X_i \in \Emoves{b_i}$ is a valid move for player $\exists$ from position $b_i$ in the fixpoint game. Furthermore, $b_{i+1} \in \Amoves{X_i}$ is a valid move for player $\forall$ in the fixpoint game. Therefore, the infinite sequence $(b_0,X_0,b_1,X_1,\ldots)$ is an infinite play in the fixpoint game, which is won by $\exists$ since there is a single greatest fixpoint equation.
  \item the path is finite and the exploration has been stopped because at some point $b_i \sqsubseteq f(\top)$ and $f_*(b_i) = \bot$ thus the only possible $X_i \subseteq B_L$ such that $X_i \equiv_H \cone{\bot} \cap B_L$ is $X_i = \emptyset$. Similarly to before, for all $j < i$, we have $b_{j+1} \in X_j$ for some $X_j \subseteq B_L$ such that $X_j \equiv_H \cone{f_*(b_j)} \cap B_L$. Note that since the game reached the position $b_i$, for all $j < i$ we must have $b_j \sqsubseteq f(\top)$. For the same reasons in the previous case, $b_i \sqsubseteq f(\lub X_i) = f(\lub \emptyset)$, thus $\emptyset \in \Emoves{b_i}$, and the sequence $(b_0,X_0,b_1,X_1,\ldots,b_i,\emptyset)$ is a finite play in the fixpoint game leading to the position $\emptyset$ where player $\forall$ cannot move, hence $\exists$ wins.
  \item the path is finite and the exploration has been stopped
    because at some point it holds $b_i \sqsubseteq \lub W$ and
    $b_i \sqsubseteq f(\top)$. Again, for all $j < i$, we have
    $b_{j+1} \in X_j$ for some $X_j \subseteq B_L$ such that
    $X_j \equiv_H \cone{f_*(b_j)} \cap B_L$. Observe that since $W$ is
    the set of positions previously encountered, it contains every
    position previously explored, thus not losing for $\exists$,
    including all $b_j$ for $j < i$. Then, for all $b' \in W$ we must
    have $b' \sqsubseteq f(\top)$. Furthermore, note that positions
    are put in $W$ only when all their successors are going to be
    explored. Therefore, for all $b' \in W$, we have
    $f_*(b') \sqsubseteq f(\top)$, otherwise there would exists a
    successor $b'' \in X' \equiv_H \cone{f_*(b')} \cap B_L$ such that
    $b'' \not\sqsubseteq f(\top)$ contradicting the fact that
    $\exists$ wins the simplified game. Let $X_i \subseteq B_L$ such
    that $X_i \equiv_H \cone{f_*(b_i)} \cap B_L$. Then, we have that
    $\lub X_i = f_*(b_i) \sqsubseteq f_*(\lub W) = \bigsqcup_{b' \in
      W} f_*(b') \sqsubseteq f(\top)$ since $b_i \sqsubseteq \lub W$,
    $f_*$ as a left adjoint preserves non-empty joins and
    $f_*(b') \sqsubseteq f(\top)$ for all $b' \in W$. Then, for all
    $b' \in X_i$, this implies that $b' \sqsubseteq f(\top)$.
    Moreover, for the same reasoning used in the first case we have
    that $b_i \sqsubseteq f(\lub X_i)$, hence $X_i \in \Emoves{b_i}$.
    An inductive argument thus proves that every path continuing the
    exploration from a $b' \in X_i$ we will never go beyond
    $f(\top)$, and so for each of those paths there exists an
    infinite sequence $(b_0,X_0,b_1,X_1,\ldots)$ such that for all
    $j$, $X_j \in \Emoves{b_j}$. Then this is an infinite play of
    the fixpoint game won by $\exists$.
  \end{itemize}
  Since all the possible moves of player $\forall$ in every set $X$ are explored, and all the paths obtained in this way (divided in the three cases above) correspond to plays in the fixpoint game won by $\exists$, we can conclude that, indeed, $\exists$ wins the fixpoint game.
\end{proof}

For instance, for Example~\ref{ex:deterministic}, the local algorithm of List.~\ref{fi:simple-game-general} works as follows: for checking whether $\{(q_1, q_2)\}$ is dominated by the solution, i.e., states $q_1$ and $q_2$ are bisimilar, one starts from $\{(q_1, q_2)\}$. At position $\{(q_1', q_2')\}$, if one state is final and the other is not, $\exists$ loses. If the pair has been already explored, the branch is not considered. Otherwise,  the pairs arising as $a$-successors $\{(q_1', q_2')\}$ are explored. If no losing position is found, the exploration finishes (recall that there are finitely many states) and $\bigcup W$ is a bisimulation including $(q_1, q_2)$.

Observe that when the basis is $B_L = L \setminus \{\bot\}$, the game
becomes deterministic also for player $\forall$: in
List.~\ref{fi:simple-game-general}, when $f_*(b') \neq \bot$ one can
take $X=\{\{ f_*(b') \}\}$, otherwise $X = \emptyset$. Therefore,  since $f_*$ is a left adjoint and thus continuous, if we take the set $S$ of all the positions generated during the exploration (i.e., $W$ with the addition of the last position, for finite games) then $\bigsqcup S = \bigsqcup_i f_*^i(b)$ is the least fixpoint of $f_*$ above $b$, which in turn coincides with the least fixpoint of $f^* \sqcup b$.
This establishes a direct link with \cite{bggp:sound-up-to-complete-abstract} which shows that for $b\in L$ it holds that $\mu (f^*\sqcup b)\sqsubseteq c$ iff $b \sqsubseteq \nu(f^*\sqcap c) = \nu f$.

Furthermore, we can bring up-to techniques into the picture: given an
up-to function $u$ we can modify the procedure in
List.~\ref{fi:simple-game-general} by replacing the winning condition
for $\exists$, that is, $b' \sqsubseteq \bigsqcup W$, by
$b' \sqsubseteq u(\bigsqcup W)$. The procedure remains clearly
complete and it is also correct due to Theorem~\ref{th:up-to-sys}. This
allows us to cover the algorithm in~\cite{BP:NFA} which checks
language equivalence for non-deterministic automata. It performs
on-the-fly determinization and
constructs a bisimulation up-to congruence on the determinized
automaton. More concretely, it tries to construct a bisimulation
relation for the determinized automaton (along the lines of
Example~\ref{ex:deterministic}) and remembers pairs $(X_1,X_2)$ of
sets of states seen so far in a relation $W$ (as explained in the
algorithm in List.~\ref{fi:simple-game-general}). Once a pair
$(Y_1,Y_2)$ is encountered that is contained in the congruence closure
of $W$ (the least equivalence, closed under union, that contains $W$),
one can stop exploring this branch. A more detailed comparison can be
found in Appendix~\ref{sec:bonchi-pous}.

\subsection{Local Algorithm for Solving the Game in the General Case}
\label{se:on-the-fly}

We now extend the local algorithm to the general case of a system of
equations. This gives us a technique for determining whether a
lattice element is below a component of the solution.
As in the simpler case, the idea consists in computing only the
information needed for the local problem of interest, in the line
of other local algorithms developed for
bisimilarity~\cite{h:mise-oeuvre-preuves-bisim} and for $\mu$-calculus
model checking~\cite{ss:practical-modcheck-games}. In particular, our
algorithm arises as a natural generalisation of the one
in~\cite{ss:practical-modcheck-games} to the setting of fixpoint games
(see Definition~\ref{de:powerset-game}).

We fix some notation and conventions which will be useful for describing the algorithm.

\paragraph*{Notation}
For the rest of the section, $L$ denotes a complete lattice, with a
basis $B_L$, and $E$ is a system of $m$ fixpoint equations over
$L$ of the kind $\vec{x} =_{\vec{\eta}} \vec{f}(\vec{x})$, with
solution $\vec{s} \in L^m$.

A generic \emph{player}, that can be either $\exists$ or $\forall$, is
usually represented by the upper case letter $P$. The opponent of
player $P$ is denoted by $\oppo{P}$.
The set of all \emph{positions} of the game is denoted by $\confs = \Econfs \cup \Aconfs$, where $\Econfs = B_L \times \interval{m}$, ranged over by $(b,i)$ is the set of
positions controlled by $\exists$, and
$\Aconfs = (\Pow{B_L})^m$, ranged over by $\vec{X}$ is the set of
positions controlled by $\forall$.
A generic position is usually denoted by
the upper case letter $C$ and we write $\owner{C}$ for the player
controlling the position $C$.

Given a position $C \in \confs$, the possible moves for player
$\owner{C}$ are indicated by $\moves{C} \subseteq \confs$. In
particular, if $C \in \Econfs$ then $\moves{C} \subseteq \Aconfs$,
otherwise $\moves{C} \subseteq \Econfs$. A function
$\mathsf{i} : \confs \to (\interval{m} \cup \{0\})$ maps every
position to a \emph{priority}, which, for positions $(b,i)$ of player
$\exists$ is the index $i$, while it is $0$ for positions of
$\forall$. With this notation, the winning condition can be expressed
as follows:
\begin{itemize}
\item Every
  finite play is won by the player who moved last. 

\item Every infinite play, seen as a sequence of positions
  $(C_1,C_2,\ldots)$, is won by player $\exists$ (resp.\ $\forall$) if
  there exists a priority $h \in \interval{m}$ s.t.\ $\eta_h = \nu$
  (resp.\ $\mu$), the set $\{j\,\mid\,\priori{C_j} = h\}$ is infinite
  and the set $\{j\,\mid\,\priori{C_j} > h\}$ is finite.
\end{itemize}

Note that there cannot be
an infinite sequence of positions with priority $0$ since only
positions of player $\forall$ have priority $0$ and players alternate
during the game.

\subsubsection{The Algorithm}

Given an element of the basis $b \in B_L$ and some index
$i \in \interval{m}$, the algorithm checks whether $b$ is below the
solution of the $i$-th fixpoint equation of the system, i.e.,
$b \sqsubseteq s_i$.
According to Theorem~\ref{th:game-corr-comp}, this corresponds to establish which of the players has a winning strategy in the fixpoint game starting from the position $(b,i)$.
The procedure roughly consists in a depth-first exploration of the
tree of plays arising as unfolding of the
game graph starting from the initial position $(b,i)$.
The algorithm optimises the search by making assumptions on particular
subtrees, which are thus pruned. Assumptions can be later confirmed or
invalidated, and thus withdrawn. The algorithm is split into three
different functions (see Fig.~\ref{fi:alg}).

\begin{itemize}

\item Function \fnname{Explore} explores the tree of plays of the
  game, trying different moves from each node in order to determine
  the player who has a winning strategy from such node.

\item Function \fnname{Backtrack} allows to backtrack from a node
  after the algorithm has established who was the winner from it,
  transmitting the information backwards.
  
\item Sometimes the algorithm makes erroneous assumptions when pruning
  the search in some position, this leads it to incorrectly designate
  a player as the winner from that position. However, the algorithm is
  able to detect this fact and correct its decisions. The correction
  is performed by the function \fnname{Forget}.
\end{itemize}

The algorithm uses the following data structures:

\begin{itemize}
\item The \emph{counter} $\vec{k}$,
  i.e., an $m$-tuple of natural numbers, which associates each
  non-zero priority with the number of times the priority has been
  encountered in the play since an higher priority was last
  encountered (the current positions is not included).
  After any move, the counter is updated taking into account the
  priority of the current position. More precisely, the update of a
  counter $\vec{k}$ when moving from a position with priority $i$,
  denoted $\mathit{next}(\vec{k},i)$, is defined as follows:
  $\mathit{next}(\vec{k},i)_j = 0$ for all $j < i$,
  $\mathit{next}(\vec{k},i)_i = k_i + 1$, and
  $\mathit{next}(\vec{k},i)_j = k_j$ for all $j > i$. Note that, in
  particular, $\mathit{next}(\vec{k},0) = \vec{k}$, i.e., moves from a
  position with priority $0$, which are the moves of $\forall$, do not
  change $\vec{k}$.
  We also define two total orders $<_\exists$ and $<_\forall$ on
  counters, that intuitively measure how good the current advancement
  of the game is for the two players. We let
  $\vec{k} <_\exists \vec{k}'$ when the largest $i$ s.t.\
  $k_i \neq k'_i$ is the index of a greatest fixpoint equation and
  $k_i < k'_i$, or it is the index of a least fixpoint and
  $k_i > k'_i$. The other order $<_\forall$ is the reverse of
  $<_\exists$, that is $\vec{k} <_\forall \vec{k}'$ iff
  $\vec{k}' <_\exists \vec{k}$.
  For each player $P$, we write $\vec{k} \leq_P \vec{k}'$ for
  $\vec{k} <_P \vec{k}'$ or $\vec{k} = \vec{k}'$.
  Notice that the update function $\mathit{next}$ is monotone
  on the counter, that is, given a priority $i$, for every
  player $P$, if $\vec{k} \leq_P \vec{k}'$, then
  $\mathit{next}(\vec{k},i) \leq_P \mathit{next}(\vec{k}',i)$.

\item The \emph{playlist} $\rho$, i.e., a list of the positions
  encountered from the root to the current node (empty if the current
  node is the root), each with the corresponding counter $\vec{k}$ and
  the indication of the alternative moves which have not been explored
  (exploration is performed depth-first). Thus, $\rho$ is a list of
  triples $(C,\vec{k},\pi)$, where $C$ is a position, $\vec{k}$ is a
  counter and $\pi \subseteq \confs$ is the set of the unexplored
  moves from that position.
  
\item The \emph{assumptions} for players $\exists$ and $\forall$,
  i.e., a pair of sets
  $\Gamma = (\Gamma_\exists, \Gamma_\forall)$.
  A position $C$ is assumed to be winning for some player when it is
  encountered for the second time in the current playlist $\rho$. This
  reveals the presence of a loop in the game graph which can be
  unfolded into an infinite play. Position $C$ is assumed to be
  winning for the player who would win such an infinite play. In
  detail, if $\vec{k}$ is the current counter and $\vec{k}'$ is the
  counter of the previous occurrence of $C$, then the winner $P$ is
  the player such that $\vec{k}' <_P \vec{k}$.  In fact, this ensures
  that the highest priority in the loop is the index of a least
  fixpoint if $P = \forall$ and of a greatest fixpoint if $P=\exists$.
  The assumption is stored with the
  corresponding counter, i.e., $\Gamma_P$ contains pairs of the kind
  $(C, \vec{k})$.
  Since other possible paths branching from the loop are possibly
  unexplored, assumptions can still be falsified afterwards.

\item The \emph{decisions} for player $\exists$ and $\forall$, i.e., a
  pair of sets $\Delta = (\Delta_\exists,
  \Delta_\forall)$. Intuitively, a decision for a player $P$ is a
  position $C$ of the game such that we established that $P$ has a
  winning strategy from $C$. The decision is stored with the
  corresponding counter, i.e., $\Delta_P$ contains pairs of the kind
  $(C, \vec{k})$.  When a new decision is added, we also record its
  \emph{justification}, i.e., the assumptions and decisions we relied
  on for deriving the new decision, if any.
\end{itemize}

For checking whether $b \sqsubseteq s_i$
for $b \in B_L$ and $i \in \interval{m}$, we call the
function \fnname{Explore}($(b,i)$, $\vec{0}$, $[]$,
$(\emptyset,\emptyset)$, $(\emptyset,\emptyset)$), where $\vec{0}$ is
the everywhere-zero counter.
This returns the (only) player $P$ having a winning strategy from position
$(b,i)$, and, by Theorem~\ref{th:game-corr-comp}, $P = \exists$ if and only if $b \sqsubseteq s_i$.

\begin{figure*}[ht]
  \footnotesize
\begin{tcbraster}[raster columns=2,raster equal height, raster force size=false, raster column skip=0mm]
  \begin{tcolorbox}[width=.47\textwidth, left=-1mm, right=0mm,colframe=white]
  \begin{algorithmic}
  \Function{Explore}{$C$, $\vec{k}$, $\rho$, $\Gamma$, $\Delta$}
    \If{$\moves{C} = \emptyset$}
    \State $\Delta_{\oppo{\owner{C}}} \coloneqq \Delta_{\oppo{\owner{C}}} \cup \{(C,\vec{k})\}$;
    \State \fnname{Backtrack}($\oppo{\owner{C}}$, $C$, $\rho$, $\Gamma$, $\Delta$);
    \ElsIf{there is $(C,\vec{k}') \in \Delta_P$ s.t.\ $\vec{k}' \leq_P \vec{k}$}
      \State \fnname{Backtrack}($P$, $C$, $\rho$, $\Gamma$, $\Delta$);
    \ElsIf{there is $(C,\vec{k}',\pi) \in \rho$}
      \State let $P$ s.t.\ $\vec{k}' <_P \vec{k}$;
      \State $\Gamma_P \coloneqq \Gamma_P \cup \{(C, \vec{k}')\}$;
      \State \fnname{Backtrack}($P$, $C$, $\rho$, $\Gamma$, $\Delta$);
    \Else
      \State pick $C' \in \moves{C}$;
      \State $\vec{k}' \coloneqq \mathit{next}(\vec{k},\priori{C})$;
      \State $\pi \coloneqq (\moves{C} \smallsetminus \{C'\}) \times \{\vec{k}'\}$;
      \State \fnname{Explore}($C'$, $\vec{k}'$, $((C,\vec{k},\pi)::\rho)$, $\Gamma$, $\Delta$);
    \EndIf
  \EndFunction
\end{algorithmic}
\end{tcolorbox}
\begin{tcolorbox}[width=.56\textwidth, left=-1mm, right=0mm, colframe=white]
\begin{algorithmic}
  \Function{Backtrack}{$P$, $C$, $\rho$, $\Gamma$, $\Delta$}
	  \If{$\rho = []$}
	    \State $P$;
	  \ElsIf{$\rho = ((C',\vec{k}',\pi)::t)$}
	    \If{$\owner{C'} \neq P$ and $\pi \neq \emptyset$}
	      \State pick $(C'',\vec{k}'') \in \pi$;
	      \State $\pi' \coloneqq \pi \smallsetminus \{(C'',\vec{k}'')\}$;
	      \State \fnname{Explore}($C''$, $\vec{k}''$, $((C',\vec{k}',\pi')::t)$, $\Gamma$, $\Delta$);
	    \Else
	      \If{$\owner{C'} = P$}
	        \State $\Delta_P \coloneqq \Delta_P \cup \{(C',\vec{k}')\}$ justified by $C$;
	      \Else
	        \State $\Delta_P \coloneqq \Delta_P \cup \{(C',\vec{k}')\}$ justified by $\moves{C'}$;
	      \EndIf
	      \State $\Gamma_P \coloneqq \Gamma_P \smallsetminus \{(C',\vec{k}')\}$;
	      \If{there is $(C',\vec{k}') \in \Gamma_{\oppo{P}}$}
	        \State $\Delta_{\oppo{P}} \coloneqq$ \fnname{Forget}($\Delta_{\oppo{P}}$, $\Gamma_{\oppo{P}}$, $(C',\vec{k}')$);
	        \State $\Gamma_{\oppo{P}} \coloneqq \Gamma_{\oppo{P}} \smallsetminus \{(C',\vec{k}')\}$;
	      \EndIf
	      \State \fnname{Backtrack}($P$, $C'$, $t$, $\Gamma$, $\Delta$);
	    \EndIf
	  \EndIf
  \EndFunction
  \end{algorithmic}
\end{tcolorbox}
\end{tcbraster}

\caption{The general local algorithm.}
  \label{fi:alg}
\end{figure*}

Given the current position $C$, the corresponding counter $\vec{k}$,
the playlist $\rho$ describing the path that led to $C$, and the sets
of assumptions $\Gamma$ and decisions $\Delta$, function
\fnname{Explore}($C$, $\vec{k}$, $\rho$, $\Gamma$, $\Delta$) checks if
one of the following three conditions holds, each one corresponding to
a different \cmd{if} branch.
\begin{itemize}
\item If $\moves{C} = \emptyset$, then the controller 
  $\owner{C}$ of position $C$ cannot move and its opponent $\oppo{\owner{C}}$
  wins. Therefore, a new decision for the current position is added
  for the opponent, and we backtrack.
  A decision of this kind, with empty justification is called a \emph{truth}.
  
\item If there is already a decision for a player $P$ for the current
  position $C$, that is, $(C,\vec{k}') \in \Delta_P$ and
  $\vec{k}' \leq_P \vec{k}$, then we can reuse that information to
  assert that $P$ would win from the current position as well. The
  requirement $\vec{k}' \leq_P \vec{k}$ intuitively ensures that we
  arrived to the current position $C$ with a play that is at least as
  good for $P$ as the play which lead to the previous decision
  $(C,\vec{k}')$.
  
\item If the current position $C$ was already encountered in the play,
  i.e., $(C,\vec{k}',\pi) \in \rho$ for some $\vec{k}'$ and $\pi$,
  then $C$ becomes an assumption for the the player $P$ for which the
  counter got strictly better, that is, $\vec{k}' <_P \vec{k}$. Then we
  backtrack.

\item If none of the conditions above holds, the exploration continues
  from $C$. A move $C' \in \moves{C}$ is chosen to be explored. The
  playlist is thus extended by adding $(C, \vec{k}, \pi)$ where $\pi$
  records the remaining moves to be explored. The counter $\vec{k}$ is
  updated according to the priority of the now past position $C$.
\end{itemize}

Function \fnname{Backtrack}($P$, $C$, $\rho$, $\Gamma$, $\Delta$) is
used to backtrack from a position $C$, reached via the
playlist $\rho$, after assuming or deciding that player $P$ would win
from such position.
\begin{itemize}
\item
  If $\rho = []$ we are back at the root, the position from where the
  computation started, and the exploration is concluded. The algorithm
  decides that player $P$ is the winner from such a position.

\item Otherwise, the head $(C', \vec{k}, \pi)$ of the playlist $\rho$
  is popped and the status of position $C'$ is investigated.

  \begin{itemize}
  \item If $C'$
  is controlled by the opponent of $P$ ($\owner{C'} \neq P$) and
  there are still unexplored moves ($\pi \neq \emptyset$), we must
  explore such moves before deciding the winner from $C'$. Then, a new
  move is extracted from $\pi$ and explored.
  
\item If instead the controller of $C'$ is $P$ ($\owner{C'} = P$) then
  $P$ wins also from $C'$. Hence $C'$ is inserted in $\Delta_P$,
  justified by the move $C$ from where we backtracked.
  Similarly, if the controller of $C'$ is the opponent of $P$
  ($\owner{C'} \neq P$), we already explored all possible moves from
  $C'$ ($\pi=\emptyset$) and all turn out to be winning for $P$, again
  we decide that $P$ wins from $C'$, which is inserted in $\Delta_P$,
  justified by all possible moves from $C'$. Since we decided that $P$
  would win from $C'$ we can now continue to backtrack. However,
  before backtracking we must discard all assumptions for the opponent
  of $P$ in conflict with the newly taken decision, and this must be
  propagated to the decisions depending on such assumptions. This is
  done by the invocation \fnname{Forget}($\Delta_{\oppo{P}}$,
  $\Gamma_{\oppo{P}}$, $(C',\vec{k}')$).
\end{itemize}
\end{itemize}

In general the choice of moves to explore, performed by the action
``pick'' in the pseudocode, is random. However,
we observed in \S\ref{se:game-characterization}, that
for  player $\exists$
it can be shown that it is sufficient to explore the minimal moves.
Furthermore, it is usually convenient to give priority to moves which
are immediately reducible to valid decisions or assumptions for the
player who is moving.
A practical way to do this is to check if there is a decision for a
position $C'$, with a valid counter wrt.\ the current one, such that
either the current position $C = (b,i)$, $C' = (b',i)$ and
$b \sqsubseteq b'$, or $C = \vec{X}$, $C' = \vec{X}'$ and
$\vec{X}' \subseteq \vec{X}$. Then, the move to pick is the one
justifying such decision, which by those features is guaranteed to be
a move also from the current position $C$.

The function \fnname{Forget} is not given explicitly. The precise
definition of the property that function \fnname{Forget} must satisfy
in order to ensure the correctness of the algorithm is quite technical
(it can be found in the appendix provided as extra material).
Intuitively, when an assumption in $\Gamma_P$ fails and is withdrawn,
then it must remove from $\Delta_P$ at least all the decisions
depending on such assumption. It is possible that decisions
taken on the base of the deleted assumption remain valid because
they can be justified by other decisions or assumptions, possibly
introduced later.
Different sound realisations of \fnname{Forget} are then possible
(see~\cite{ss:practical-modcheck-games}) and, experimentally, it can
been seen that those removing only the least possible set of decisions
can be practically inefficient. A simple sound implementation, which,
at least in the setting of the $\mu$-calculus, resulted to be the most
efficient is based on a temporal criterion: when an assumption fails,
all decisions which have been taken after that assumption are
deleted. This can be implemented by associating timestamps with
decisions and assumptions, and avoiding the complex management of
justifications.

\begin{toappendix}
\begin{definition}[sound forget]
  \label{de:sound-forget}
  Whenever function \textup{\fnname{Forget}}($\Delta_P$, $\Gamma_P$,
  $(C,\vec{k})$) is invoked, returning $\Delta'_P$, for every decision
  $(C',\vec{k}') \in \Delta'_P$, for every position $C''$ justifying
  that decision, there exists $(C'',\vec{k}'') \in \Delta'_P$ such
  that $\vec{k}'' \leq_P \mathit{next}(\vec{k}',\priori{C'})$ or there
  exists $(C'',\vec{k}'') \in \Gamma_P \smallsetminus \{(C,\vec{k})\}$
  such that $\vec{k}'' <_P \mathit{next}(\vec{k}',\priori{C'})$.
\end{definition}
\end{toappendix}

\begin{example}[model-checking $\mu$-calculus]
  \label{ex:mc-mu}
  Consider the transition system $T = (\mathbb{S}, \to)$ in
  Fig.~\ref{fi:running-ts} and the $\mu$-calculus formula
  $\phi = \mu x_2.((\nu x_1.(p \land\Box x_1))\lor\Diamond x_2)$
  discussed in Example~\ref{ex:mu}. As already discussed, the formula
  $\phi$ interpreted over $T$ leads to the system $E$ in
  Fig.~\ref{fi:running-eq} over the lattice
  $\Pow{\mathbb{S}}$.

  Suppose that we want to verify whether the state
  $a \in \mathbb{S}$ satisfies the formula $\phi$. This requires to determine the winner of the fixpoint game from position $(a,2)$, which can be done by invoking \fnname{Explore}($(a,2)$, $\vec{0}$, $[]$,
  $(\emptyset, \emptyset)$, $(\emptyset, \emptyset)$).
  A computation performed by the algorithm is schematised in
  Fig.~\ref{fi:alg-execution}, where we only consider minimal
  moves. Since the choice of moves is non-deterministic, other search
  sequences are possible. In the diagram, positions of player
  $\exists$ are represented as diamonds, while those of $\forall$ are
  represented as boxes, the counters associated with the positions is
  on their lefthand side

  \begin{figure}[ht]
    \begin{center}
    \begin{tikzpicture}[node distance=9mm,>=stealth',y=6mm, x=8mm]
      \node (a2) at (4,16) [anode]{$(a,2)$};
      \node (k0) at (12,16) {$(0,0)$};
      \node (a-0) at (-0.5,14) [enode]{$(\{a\},\emptyset)$};
      \node (0-a) at (2.5,14) [enode]{$(\emptyset,\{a\})$};
      \node (0-b) at (5.5,14) [enode]{$(\emptyset,\{b\})$};
      \node (0-c) at (8.5,14) [enode]{$(\emptyset,\{c\})$};
      \node (k1) at (12,14) {$(0,1)$};
      \node (b2) at (4,12) [anode]{$(b,2)$};
      \node (k2) at (12,12) {$(0,1)$};
      \node (0-d) at (0,10) [enode]{$(\emptyset,\{d\})$};
      \node (b-0) at (4,10) [enode]{$(\{b\},\emptyset)$};
      \node (0-e) at (8,10) [enode]{$(\emptyset,\{e\})$};
      \node (k3) at (12,10) {$(0,2)$};
      \node (b1) at (4,8) [anode]{$(b,1)$};
      \node (k4) at (12,8) {$(0,2)$};
      \node (de-0) at (4,6) [enode]{$(\{d,e\},\emptyset)$};
      \node (k5) at (12,6) {$(1,2)$};
      \node (d1) at (1,4) [anode]{$(d,1)$};
      \node (e1) at (7,4) [anode]{$(e,1)$};
      \node (k6) at (12,4) {$(1,2)$};
      \node (d-0) at (1,2) [enode]{$(\{d\},\emptyset)$};
      \node (e-0) at (7,2) [enode]{$(\{e\},\emptyset)$};
      \node (k7) at (12,2) {$(2,2)$};
      \node (d1') at (1,0) [anode]{$(d,1)$};
      \node (e1') at (7,0) [anode]{$(e,1)$};
      \node (k8) at (12,0) {$(2,2)$};
      \draw  [->] (a2) to node {} (a-0);
      \draw  [->] (a2) to node {} (0-a);
      \draw  [->,very thick] (a2) to node {} (0-b);
      \draw  [->] (a2) to node {} (0-c);
      \draw  [->,very thick] (0-b) to node {} (b2);
      \draw  [->] (b2) to node {} (0-d);
      \draw  [->,very thick] (b2) to node {} (b-0);
      \draw  [->] (b2) to node {} (0-e);
      \draw  [->,very thick] (b-0) to node {} (b1);
      \draw  [->,very thick] (b1) to node {} (de-0);
      \draw  [->,very thick] (de-0) to node {} (d1);
      \draw  [->,very thick] (de-0) to node {} (e1);
      \draw  [->,very thick] (d1) to node {} (d-0);
      \draw  [->,very thick] (e1) to node {} (e-0);
      \draw  [->,very thick] (d-0) to node {} (d1');
      \draw  [->,very thick] (e-0) to node {} (e1');
    \end{tikzpicture}
    \end{center}
    \caption{An execution of the local algorithm.}
    \label{fi:alg-execution}
  \end{figure}
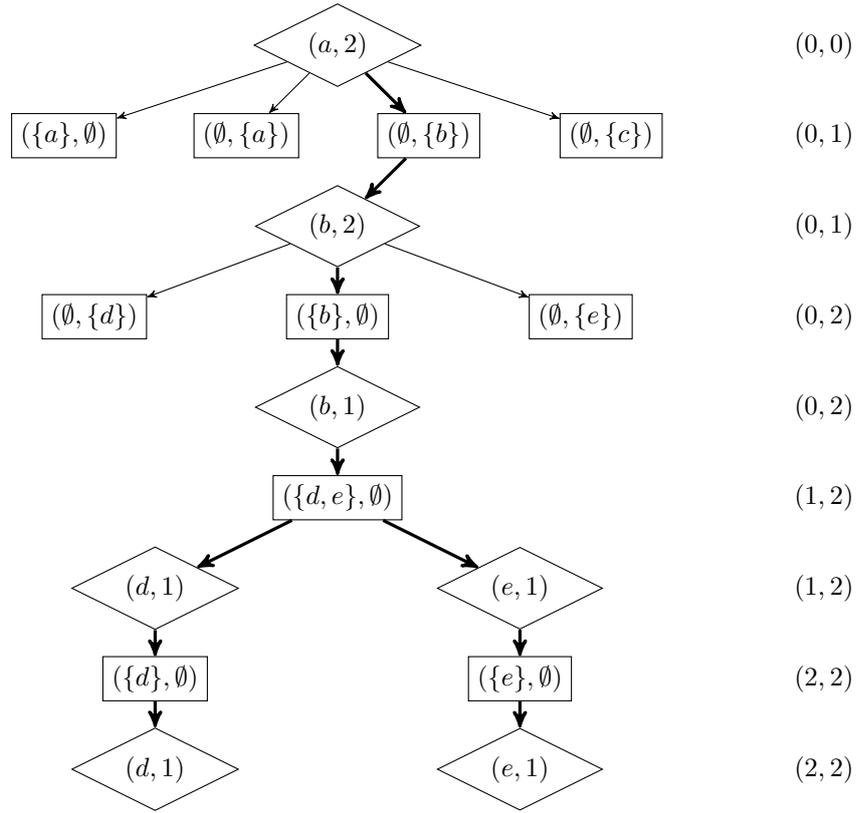

  Recall that the second equation is
  $x_2 =_{\mu} x_1 \cup \semdia_T x_2$.  Then, from the initial
  position $(a,2)$, with counter $(0,0)$, there are four available
  minimal moves, i.e., $(\{a\},\emptyset)$, $(\emptyset,\{a\})$,
  $(\emptyset,\{b\})$ and $(\emptyset,\{c\})$, represented by the four
  outgoing edges from position $(a,2)$ in the diagram, all four will
  have counter $(0,1) = \mathit{next}((0,0),2)$. Indeed, it is easy to
  see that
  $a \in \{a\} \cup \semdia_T \emptyset = \emptyset \cup \semdia_T
  \{a\} = \emptyset \cup \semdia_T \{b\} = \{a\} \subseteq \emptyset
  \cup \semdia_T \{c\} = \{a,c\}$. Suppose that the algorithm chooses
  to explore the move $(\emptyset,\{b\})$, as highlighted by the bold
  arrow. Even though not shown in the diagram, the other moves are
  stored in the set of unexplored moves $\pi$ associated with the
  position $(a,2)$ in the playlist $\rho$. The search proceeds in this
  way along the moves
  \begin{align*}
    (a,2) \stackrel{\exists}{\leadsto} (\emptyset,\{b\})
      \stackrel{\forall}{\leadsto} (b,2) \stackrel{\exists}{\leadsto}
      (\{b\},\emptyset) \stackrel{\forall}{\leadsto} (b,1) 
      \stackrel{\exists}{\leadsto} (\{d,e\},\emptyset)
      \stackrel{\forall}{\leadsto} (d,1) \stackrel{\exists}{\leadsto}
      (\{d\},\emptyset) \stackrel{\forall}{\leadsto}
  \end{align*}
  until position $(d,1)$ occurs again, with counter $(2,2)$. Since the counter associated with the first occurrence of $(d,1)$ was $(1,2)$ and $(1,2) <_\exists (2,2)$, then the pair position and counter $((d,1),(1,2))$ is added as an assumption for player $\exists$ and the algorithm starts backtracking. While backtracking it generates a decision for $\exists$, which is $((\{d\},\emptyset),(2,2))$ justified by the only possible move $(d,1)$ of player $\forall$. When it comes back to the first occurrence of $(d,1)$, since it is a position controlled by $\exists$, the procedure transforms the assumption $((d,1),(1,2))$ into a decision for $\exists$ justified by the move $(\{d\},\emptyset)$. Then, it backtracks to position $(\{d,e\},\emptyset)$, which is controlled by player $\forall$ and there is still an unexplored move $(e,1)$. Therefore, the algorithm starts exploring again from $(e,1)$, and does so similarly to the previous branch of $(d,1)$. After making decisions for those positions as well, the algorithm resumes backtracking from $(\{d,e\},\emptyset)$, since all possible moves have been explored, making decisions for player $\exists$ along the way back. This goes on up until the root is reached again. The last invocation \fnname{Backtrack}($\exists$, (a,2), [], $\Gamma$, $\Delta$) terminates since $\rho = []$, and returns player $\exists$. Indeed, $\exists$ wins starting from position $(a,2)$ since the state $a$ satisfies the formula $\phi$.
\end{example}

\subsubsection{Correctness}

We show that, when the lattice is finite, the algorithm
terminates. Moreover, when it terminates (which could happen also on
infinite lattices), it provides a correct answer.

Termination on finite lattices can be proved by observing that the set of positions (which are either elements of the basis or tuples of sets of elements of the basis) is finite. The length of playlists is bounded by the number of positions, since, whenever a position repeats in a playlist, it necessarily becomes an assumption and backtracking starts. Finally, one can observe that it is not possible to cycle indefinitely between two positions, so that termination immediately follows.

\begin{toappendix}

\begin{lemma}[assumptions and plays]
  \label{le:alg-assumpt}
  Given a fixpoint game, whenever functions
  \textup{\fnname{Explore}}($\cdot$, $\cdot$, $\rho$, $\Gamma$,
  $\Delta$) and \textup{\fnname{Backtrack}}($\cdot$, $\cdot$, $\rho$,
  $\Gamma$, $\Delta$) are invoked, for every player $P$, for all
  $(C,\vec{k}) \in \Gamma_P$ it holds $(C,\vec{k},\pi) \in \rho$ for
  some $\pi$. 
\end{lemma}

\begin{proof}
  Easily proved by an inspection of the code. Initially, on the call
  \fnname{Explore}($C_0$, $\vec{0}$, $[]$, $(\emptyset,\emptyset)$,
  $(\emptyset,\emptyset)$), the property vacuously holds since both
  $\Gamma_\exists$ and $\Gamma_\forall$ are empty. Now, the only way
  that could make the property fail are by adding new assumptions or
  backtracking, hence shortening the playlist $\rho$. The only
  position in the code where new assumptions are added is in the
  function \fnname{Explore}. A new assumption $(C,\vec{k}')$ is added
  only if $(C,\vec{k}',\pi) \in \rho$, for some $\pi$, thus the
  property still holds. On the other hand, the only place where the
  backtracking really happens, that is, $\rho$ is effectively shorten,
  is at the end of the backtracking function, when
  \fnname{Backtrack}($P$, $C'$, $t$, $\Gamma$, $\Delta$) is
  invoked. More precisely, the head $(C',\vec{k}',\pi)$ is removed
  from the playlist $\rho$. However, before the aforementioned
  invocation, $(C',\vec{k}')$ was already removed from $\Gamma_P$ and
  from $\Gamma_{\oppo{P}}$, if it were in $\Gamma_{\oppo{P}}$. And so
  again the property still holds.
\end{proof}
\end{toappendix}

\begin{lemmarep}[termination]
  \label{le:alg-termination}
  Given a fixpoint game on a finite lattice, any call \textup{\fnname{Explore}}($C_0$, $\vec{0}$, $[]$, $(\emptyset,\emptyset)$, $(\emptyset,\emptyset)$) terminates, hence at some point \textup{\fnname{Backtrack}}($P$, $C_0$, $[]$, $(\emptyset, \emptyset)$, $\Delta$) is invoked, for some player $P$ and pairs of sets $\Gamma$ and $\Delta$.
\end{lemmarep}

\begin{proof}
  Consider the sequence $\sigma$ of invocations to functions \fnname{Explore} and \fnname{Backtrack} in the order they happen, originating from a call \fnname{Explore}($C_0$, $\vec{0}$, $[]$, $(\emptyset,\emptyset)$, $(\emptyset,\emptyset)$). Let $\tau$ be the subsequence of $\sigma$ obtained removing all calls to \fnname{Backtrack}. We show that such sequence is finite. First, since the lattice is finite, hence $\confs$ is finite, the set of playlists $\rho$ in the invocations in $\tau$ is also finite. Actually, this is not true in general for any set of playlists, but it holds for the set of lists we obtain during any computation. Indeed, this can be seen inductively, showing that every playlist $\rho$ has length bounded by $|\confs|$. At the beginning we have the empty list $[]$ which is clearly bounded by $|\confs|$. Then, by inspecting the code it can be seen that the only function which increases the size of $\rho$ is \fnname{Explore}, and it happens only if the current position $C$, with counter $\vec{k}$, is not already contained in $\rho$ with a counter $\vec{k}'$ s.t.\ $\vec{k}' <_P \vec{k}$ for some player $P$. But whenever a position $C$ already in $\rho$ is encountered again it must be with a counter strictly larger for one of the players. The only case where this could possibly fail is when the subsequence of $\rho$ between the two occurrences of $C$ contains only positions with priority $0$. But, as already mentioned, this cannot happen because players alternate during the game and only $\forall$ has positions with priority $0$. Thus, every time a position recurs, the playlist is not extended any more. So, the size of the playlist is necessarily bounded by the size of $\confs$. Furthermore, the set of playlists of length bounded by $|\confs|$ is finite because every $\pi$ in them is bounded as well, since $\pi \subseteq \confs$, and the same happens for the counters $\vec{k}$ since they are computed starting from $\vec{0}$ and increased at most by $1$ in some component only when the list is extended. Therefore, $\tau$ must contain only a finite number of different playlists $\rho$, possibly with repetitions. Now, in order to show that $\tau$ is finite, we define a partial order $\leq$ over the playlists in $\tau$ as follows, $\forall \rho,\rho',\rho'',C,\vec{k},\pi,\pi'$:
  \begin{itemize}
    \item $\rho'\rho \leq \rho$
    \item if $\pi \subsetneq \pi'$, then $\rho''((C,\vec{k},\pi)::\rho) \leq \rho'((C,\vec{k},\pi')::\rho)$.
  \end{itemize}
  It is easy to see that such order is reflexive, antisymmetric, and transitive. Since the set of playlists in $\tau$ is finite, so is the corresponding poset with the given partial order. By an inspection of the code it can be seen that for every two playlists $\rho,\rho'$ in consecutive invocations of \fnname{Explore} in $\tau$, we have that $\rho' < \rho$, since:
  \begin{itemize}
    \item function \fnname{Explore} extends the playlist $\rho$ until function \fnname{Backtrack} is invoked
    \item function \fnname{Backtrack} shortens the playlist $\rho$ until it is empty or function \fnname{Explore} is invoked, after shortening the set of unexplored moves $\pi$ in $\rho$.
  \end{itemize}
  So the playlists in $\tau$ form a strictly descending chain in a finite poset, thus $\tau$ must be finite. And this immediately proves that $\sigma$ is finite as well, because otherwise from a certain point on we would have infinitely many calls to \fnname{Backtrack} only, which would shorten the playlist infinitely many times. And so we can conclude that any computation originating from a call \fnname{Explore}($C_0$, $\vec{0}$, $[]$, $(\emptyset,\emptyset)$, $(\emptyset,\emptyset)$) must terminate. Finally, since the only instruction returning a value (hence terminating the execution) is in the function \fnname{Backtrack} and it is reached only when $\rho = []$, then \fnname{Backtrack}($P$, $C$, $[]$, $\Gamma$, $\Delta$) must have been invoked on some $P$, $C$, $\Gamma$, $\Delta$. Furthermore, $C = C_0$ because $\rho = []$ is the list of positions from the root $C_0$ to the current node $C$.

  We immediately conclude that $\Gamma=(\emptyset,\emptyset)$ by exploting Lemma~\ref{le:alg-assumpt}.
\end{proof}

The proof of correctness is long and technical. The underlying idea is
to prove that, at any invocation of \fnname{Explore}($\cdot$, $\cdot$,
$\rho$, $\Gamma$, $\Delta$) and \fnname{Backtrack}($\cdot$, $\cdot$,
$\rho$, $\Gamma$, $\Delta$), the justifications for the decisions
$\Delta_P$, can be interpreted as a winning strategy for player $P$
from the positions $C \in \Delta_P$, in a modified game where $P$
immediately wins on the assumptions $\Gamma_P$. Since at termination,
the set of assumptions is empty, the modified game coincides with the
original one and thus we conclude.

\begin{toappendix}
\begin{lemma}[backtracking position]
  \label{le:alg-backtrack-pos}
  Given a fixpoint game, whenever function \textup{\fnname{Backtrack}}($P$, $C$, $\rho$, $\Gamma$, $\Delta$) is invoked, it holds $(C,\vec{k}) \in \Delta_P \cup \Gamma_P$ for some $\vec{k}$.
\end{lemma}

\begin{proof}
  Immediate by inspecting the invocations of \fnname{Backtrack} in the code.
\end{proof}

\begin{lemma}[uncontrolled decisions]
  \label{le:unown-dec}
  Given a fixpoint game, whenever functions \textup{\fnname{Explore}}($\cdot$, $\cdot$, $\cdot$, $\Gamma$, $\Delta$) and \textup{\fnname{Backtrack}}($\cdot$, $\cdot$, $\cdot$, $\Gamma$, $\Delta$) are invoked, for every player $P$, for all $(C,\vec{k}) \in \Delta_P$, if $\owner{C} \neq P$, then for all $C' \in \moves{C}$ it holds $(C',\vec{k}') \in \Delta_P \cup \Gamma_P$ for some $\vec{k}'$.
\end{lemma}

\begin{proof}
  By inspecting the code it is easy to see that every time we add a new decision $(C,\vec{k})$ for a player $P$ that is not the owner of $C$, either:
  \begin{itemize}
    \item $\moves{C} = \emptyset$, thus the property vacuously holds, or
    \item the procedure already explored all possible moves $\moves{C}$ and they all became decisions or assumptions for $P$, since we are in the case where $\owner{C} \neq P$ and $\pi = \emptyset$.
  \end{itemize}
  Furthermore, such a decision $(C,\vec{k})$ is justified by $\moves{C}$. Therefore, if one of those moves were to be deleted from the assumptions or decisions of $P$ at some point, the function \fnname{Forget} would delete $(C,\vec{k})$ as well.
\end{proof}

For the next results we make use of fixpoint games suitably modified for a set of assumptions for a player. For a set $S$ of decisions or assumptions we denote by $\mathsf{C}(S)$ its first projection, that is, the set of positions appearing as first component in the elements of $S$.

\begin{definition}[game with assumptions]
  Given a fixpoint game $G$ and a player $P$, the corresponding \emph{game with assumptions $\Gamma_P$} is a parity game $G(\Gamma_P)$ obtained from $G$ where for all $C \in \confs$, if $C \in \mathsf{C}(\Gamma_P)$, then $\owner{C} = \oppo{P}$ and $\moves{C} = \emptyset$, otherwise they are the same as in $G$.
\end{definition}

Notice that when the set of assumptions is empty $\Gamma_P = \emptyset$, the modified game is the same of the original one.

Then, we define a kind of strategies based on decisions and assumptions for a player, which fit the modified games above. Such strategies are history-free partial strategies. Indeed they only prescribe moves from decisions.

\begin{definition}[strategy with assumptions]
  \label{de:strat-w-assumpt}
  Let $G$ be a fixpoint game. Given a player $P$, a \emph{strategy with assumptions $\Gamma_P$ from decisions $\Delta_P$} for $P$ is a function $s_P : \mathsf{C}(\Delta_P \cup \Gamma_P) \to \Pow{\mathsf{C}(\Delta_P \cup \Gamma_P)}$ where for all $C \in \mathsf{C}(\Gamma_P)$, $s_P(C) = \emptyset$, and for all $C \in \mathsf{C}(\Delta_P) \smallsetminus \mathsf{C}(\Gamma_P)$, $s_P(C)$ is the set of positions, possibly empty, justifying the decision $(C,\min_{\leq_P} \{\vec{k}\ \mid\ (C,\vec{k}) \in \Delta_P\})$. Given a position $C \in \mathsf{C}(\Delta_P)$, we denote by $d_P(C) = \min_{\leq_P} \{\vec{k}\ \mid\ (C,\vec{k}) \in \Delta_P\}$ the counter that was associated with $C$.

  We say that the strategy $s_P$ is \emph{winning} when it is winning in the modified game $G(\Gamma_P)$, that is, every play in $G(\Gamma_P)$ following $s_P$ starting from a position in $\mathsf{C}(\Delta_P)$ is won by player $P$.
\end{definition}

The definition above is well given since by Lemmata~\ref{le:alg-backtrack-pos}~and~\ref{le:unown-dec} we know that when we add a new decision justified by some other, those are already included in the decisions or assumptions for the same player. Moreover, notice that the minimum of $\{\vec{k}\ \mid\ (C,\vec{k}) \in \Delta_P\}$ is guaranteed to be in the set itself because $\leq_P$ is a total order and the set is never empty since $C \in \mathsf{C}(\Delta_P)$.

In the modified game $G(\Gamma_P)$, given the strategy $s_P$ with assumptions $\Gamma_P$ from decisions $\Delta_P$, for each position $C \in \mathsf{C}(\Delta_P)$ we can build a tree including all the plays starting from $C$ where player $P$ follows the strategy $s_P$.

\begin{definition}[tree of plays]
  Let $G$ be a fixpoint game. Given a player $P$ and the strategy $s_P$ with assumptions $\Gamma_P$ from decisions $\Delta_P$, for each position $C \in \mathsf{C}(\Delta_P)$, the \emph{tree of the plays following $s_P$ starting from $C$} is the tree $\tau^{C}_{s_P}$ rooted in $C$, where every node $C'$ in it has successors $s_P(C')$.
\end{definition}

Such trees can contain both finite and infinite paths. Finite complete
paths terminate in assumptions or truths, infinite ones contain only decisions. By construction and definition of strategy with assumptions every node is either a decision or an assumption for $P$. More precisely, every inner node is a position in $\mathsf{C}(\Delta_P)$, and every leaf corresponds to either a truth in $\Delta_P$ or an assumption in $\Gamma_P$. It is easy to see that a tree $\tau^{C}_{s_P}$ includes all the possible plays from $C$ following $s_P$ since the successors of inner nodes owned by the opponent are all the possible moves from those positions (decisions controlled by the opponent are justified by all the possible opponent's moves, Lemma~\ref{le:unown-dec}).

The trees defined above are all we need to show that a strategy with assumptions is winning. Indeed, it is enough to show that every complete path in each of those trees corresponds to a play won by the player. To this end, first we observe some key properties of the paths in the trees.

\begin{lemmarep}[priorities in strategy paths]
  \label{le:high-prior-path}
  Given a fixpoint game, whenever functions \textup{\fnname{Explore}}($\cdot$, $\cdot$, $\cdot$, $\Gamma$, $\Delta$) and \textup{\fnname{Backtrack}}($\cdot$, $\cdot$, $\cdot$, $\Gamma$, $\Delta$) are invoked, for every player $P$, given the strategy $s_P$ with assumptions $\Gamma_P$ from decisions $\Delta_P$, for all $\hat{C} \in \mathsf{C}(\Delta_P)$, the tree of plays $\tau^{\hat{C}}_{s_P}$ satisfies the following properties
  \begin{enumerate}
    \item for every pair of inner nodes $C,C'$ in $\tau^{\hat{C}}_{s_P}$ s.t.\ $C'$ is a successor of $C$, it holds $d_P(C') \leq_P \mathit{next}(d_P(C),\priori{C})$
    \item for every non-empty inner path $C_1,\ldots,C_n$ in $\tau^{\hat{C}}_{s_P}$, if $d_P(C_1) <_P \mathit{next}(d_P(C_n),\priori{C_n})$, then $P = \exists$ iff $\eta_h = \nu$, where $h$ is the highest priority occurring along the path.
  \end{enumerate}
\end{lemmarep}

\begin{proof}
  We prove the two properties separately.
  \begin{enumerate}
    \item Observe that we must have $C' \in s_P(C)$ by definition of $\tau^{\hat{C}}_{s_P}$. This means that there exists a decision $(C,d_P(C)) \in \Delta_P$ justified by the position $C'$. Then $(C,d_P(C))$ must have been added by a call to \fnname{Backtrack}. By inspecting the code it is easy to see that we were backtracking either after adding a new decision $(C',\mathit{next}(d_P(C),\priori{C}))$ or because there was already a decision $(C',\vec{k}')$ s.t.\ $\vec{k}' \leq_P \mathit{next}(d_P(C),\priori{C})$. Since $d_P(C') = \min_{\leq_P} \{\vec{k}\ \mid\ (C',\vec{k}) \in \Delta_P\}$, in both cases we can immediately conclude that $d_P(C') \leq_P \mathit{next}(d_P(C),\priori{C})$.

    \item We assume that $d_P(C_1) <_P \mathit{next}(d_P(C_{n}),\mathsf{i}(C_{n}))$ and $P = \exists$, and we prove that $\eta_h = \nu$, where $h$ is the highest priority occurring along the path. A dual reasoning holds for $P = \forall$. Let $\mathit{next}^{j}$ be a function that computes the counter after a subsequence of positions $C_1,\ldots,C_j$ in the path $C_1,\ldots,C_n$, for $j \in \interval{n}$. The function is inductively defined by $\mathit{next}^{j}(\vec{k}) = \mathit{next}(\mathit{next}^{j-1}(\vec{k}),\priori{C_j})$ for all $j \in \interval{n}$, and $\mathit{next}^{0}(\vec{k}) = \vec{k}$. The inductive computation just repeatedly applies the function $\mathit{next}$ for each position encountered along the sequence starting from a given counter $\vec{k}$. We observe that the function satisfies the property $d_\exists(C_j) \leq_\exists \mathit{next}^{j-1}(d_\exists(C_1))$ for all $j \in \interval{n}$. We show this by induction on $j$. Clearly it holds for $j = 1$, since by definition $\mathit{next}^{0}(d_\exists(C_1)) = d_\exists(C_1)$. Then, assuming it holds for $j$, we prove it for $j+1$. Since we know that $\mathit{next}$ is monotone wrt.\ the input counter, by inductive hypothesis we obtain that $\mathit{next}(d_\exists(C_j),\mathsf{i}(C_{j})) \leq_P \mathit{next}(\mathit{next}^{j-1}(d_\exists(C_1)),\mathsf{i}(C_{j})) = \mathit{next}^{j}(d_\exists(C_1))$, where the last equality holds by definition of $\mathit{next}^{j}$. Furthermore, we know that $d_\exists(C_{j+1}) \leq_\exists \mathit{next}(d_\exists(C_{j}),\mathsf{i}(C_{j}))$ by (a) above, since $C_{j+1}$ is a successor of $C_j$. And so we can immediately deduce that indeed $d_\exists(C_{j+1}) \leq_\exists \mathit{next}^{j}(d_\exists(C_1))$. From this and the initial assumptions we have that $d_\exists(C_1) <_\exists \mathit{next}(d_\exists(C_{n}),\mathsf{i}(C_{n})) \leq_\exists \mathit{next}^{n}(d_\exists(C_1))$, where the last inequality holds by definition of $\mathit{next}^{n}$ and monotonicity of $\mathit{next}$. Observe that since $\mathit{next}^{n}$ just recursively applies the function $\mathit{next}$ on the positions $C_1,\ldots,C_n$, the final result and the initial counter $d_\exists(C_1)$ can only differ on priorities among those of the positions $C_1,\ldots,C_n$ and lower ones (which could have been zeroed). Therefore, the highest priority on which $d_\exists(C_1)$ and $\mathit{next}^{n}(d_\exists(C_1))$ do not coincide must be the highest priority $h$ appearing along the path. Furthermore, we must have $d_\exists(C_1)_h < \mathit{next}^{n}(d_\exists(C_1))_h$, because values can only increase or become zero, when a higher priority is encountered (and its value increased), but this would contradict the fact that $h$ is the highest. Now we can easily conclude since by hypothesis $d_\exists(C_1) <_\exists \mathit{next}^{n}(d_\exists(C_1))$, and so by definition of the order $<_\exists$ we must have that $\eta_h = \nu$.
  \end{enumerate}
\end{proof}

We observe that winning strategies with assumptions are preserved by a sound function \fnname{Forget} after removing an assumption and the related decisions.

\begin{lemmarep}[strategies and forget]
  \label{le:forget-strat}
  Given a fixpoint game, whenever \textup{\fnname{Forget}}($\Delta_P$, $\Gamma_P$, $(C,\vec{k})$) is invoked, returning $\Delta'_P$, if the strategy with assumptions $\Gamma_P$ from decisions $\Delta_P$ is winning in the modified game with assumptions $\Gamma_P$, then the strategy with assumptions $\Gamma_P \smallsetminus \{(C,\vec{k})\}$ from decisions $\Delta'_P$ is winning in the modified game with assumptions $\Gamma_P \smallsetminus \{(C,\vec{k})\}$.
\end{lemmarep}

\begin{proof}
  It follows immediately from Definitions~\ref{de:sound-forget}~and~\ref{de:strat-w-assumpt}.
\end{proof}

\begin{lemma}[winning strategy from decisions]
  \label{le:alg-inv}
  Given a fixpoint game, whenever functions \textup{\fnname{Explore}}($\cdot$, $\cdot$, $\cdot$, $\Gamma$, $\Delta$) and \textup{\fnname{Backtrack}}($\cdot$, $\cdot$, $\cdot$, $\Gamma$, $\Delta$) are invoked, for every player $P$, the strategy with assumptions $\Gamma_P$ from decisions $\Delta_P$ is winning in the modified game with assumptions $\Gamma_P$.
\end{lemma}

\begin{proof}
  We prove this by induction on the sequence of functions calls. Initially, on the first call \fnname{Explore}($C$, $\vec{0}$, $[]$, $(\emptyset,\emptyset)$, $(\emptyset,\emptyset)$), the property vacuously holds since $\Delta_\exists = \Delta_\forall = \emptyset$. Now, assuming that the property holds when a function is called, we show that it holds also on every invocation performed by such function.

  Assume that the property holds when \fnname{Explore}($C$, $\vec{k}$, $\rho$, $\Gamma$, $\Delta$) is called. The only invocation where the property could possibly fail is \fnname{Backtrack}($\oppo{\owner{C}}$, $C$, $\rho$, $\Gamma$, $\Delta$) after $(C,\vec{k})$ has been added to the decisions for $\oppo{\owner{C}}$, when $\moves{C} = \emptyset$. However we can immediately see that $\oppo{\owner{C}}$ wins from $C$ since the opponent $\owner{C}$ cannot move (the strategy is always winning from $C$). On all the other calls the property is preserved since all decisions are unchanged and no assumption has been removed.

  Assume that the property holds when \fnname{Backtrack}($P$, $C$, $\rho$, $\Gamma$, $\Delta$) is called. There are only two invocations to check. Clearly the property is preserved on the first one, i.e., \fnname{Explore}($C''$, $\vec{k}''$, $\rho$, $\Gamma$, $\Delta$), since all decisions and assumptions are unchanged.
  The second case is instead more complex. This is when the function \fnname{Backtrack}($P$, $C'$, $t$, $\Gamma$, $\Delta$) is invoked. Let us analyse the strategy for one player at a time. First, consider the opponent $\oppo{P}$. Even though the assumption $(C',\vec{k}')$ might have been removed from $\Gamma_{\oppo{P}}$, all decisions in $\Delta_{\oppo{P}}$ depending on such assumption have been removed as well via the function \fnname{Forget}($\Delta_{\oppo{P}}$, $\Gamma_{\oppo{P}}$, $(C',\vec{k}')$). Let $\Delta'_{\oppo{P}}$ be the remaining decisions. By Lemma~\ref{le:forget-strat} we know that the strategy with assumptions $\Gamma_{\oppo{P}} \smallsetminus \{(C',\vec{k}')\}$ from decisions $\Delta'_{\oppo{P}}$ is winning as long as the strategy with assumptions $\Gamma_{\oppo{P}}$ from decisions $\Delta_{\oppo{P}}$ was winning. Then by inductive hypothesis the property still holds for ${\oppo{P}}$.
  Now we need to prove the property for player $P$ as well. That is, the strategy $s_P$ with assumptions $\Gamma_P \smallsetminus \{(C',\vec{k}')\}$ from decisions $\Delta_P \cup \{(C',\vec{k}')\}$ is winning in the modified game with assumptions $\Gamma_P \smallsetminus \{(C',\vec{k}')\}$. To do this we just need to show that for every position $\hat{C} \in \mathsf{C}(\Delta_P \cup \{(C',\vec{k}')\})$, every complete path in the tree of plays $\tau^{\hat{C}}_{s_P}$ is a play won by $P$. First, recall that every finite complete path in $\tau^{\hat{C}}_{s_P}$ terminates in a position of an assumption or a truth. In both cases such a finite play is always won by $P$ since in the modified game assumptions and truths correspond to positions owned by the opponent with no available moves.
  By inductive hypothesis we know that the strategy $s'_P$ with assumptions $\Gamma_P$ from decisions $\Delta_P$ was winning in the modified game with assumptions $\Gamma_P$. Notice that the two strategies can only differ on the position $C'$ of the new decision $(C',\vec{k}')$. It may be that $s'_P$ was not defined on $C'$, if there was no decision or assumption for such position before now. Anyway, this means that if $C'$ never occurs along the path, then the play must be won by $P$ since $s_P$ and $s'_P$ coincide on all the positions in the path and $s'_P$ was winning by inductive hypothesis.
  Therefore we just need to check those paths containing $C'$. If $C'$ appears just finitely many times along the path, consider the subpath starting from the successor $C''$ of the last occurrence of $C'$. Such subpath does not contain $C'$ and it is still infinite. Recalling that all positions in infinite paths must come from decisions and $C'' \neq C'$, then the subpath must be one of the complete paths in the tree of plays $\tau^{C''}_{s'_P}$. Thus, by inductive hypothesis the subpath, as well as the initial one, must be a play won by $P$.
  Otherwise, $C'$ appears infinitely many times along the path. Consider every subpath between two consecutive occurrences of $C'$, including only the first one. In such subpath let $C'' \neq C'$ be the last position, which is the predecessor of the second occurrence of $C'$. Observe that no decision $(C',\vec{k})$ could have been added after exploring $(C',\vec{k}')$ and before now, because we would necessarily have either $\vec{k} <_P \vec{k}'$ or $\vec{k} <_{\oppo{P}} \vec{k}'$, thus satisfying the condition of the third \cmd{if} branch of function \fnname{Explore}, in which case the exploration would have stopped and $(C',\vec{k})$ would have never been added as a decision. Furthermore, any decision $(C',\vec{k})$ added before exploring $(C',\vec{k}')$ must be such that $\vec{k}' < \vec{k}$, because otherwise the exploration would have stopped satisfying the second \cmd{if} branch of function \fnname{Explore} and $(C',\vec{k}')$ would have never been added as a decision. Therefore we must have $d_P(C') = \vec{k}'$ and, if $C' \in \mathsf{C}(\Delta_P) \smallsetminus \mathsf{C}(\Gamma_P)$ hence $s'_P$ is defined on $C'$, $d_P(C') <_P d'_P(C')$ since $d'_P(C')$ is the minimum $\vec{k}$ among the decisions for $C'$ added before $(C',\vec{k}')$. Moreover, in the latter case, by Lemma~\ref{le:high-prior-path}(a) we obtain that $d_P(C') <_P d'_P(C') \leq_P \mathit{next}(d_P(C''),\priori{C''})$ since $C'$ succeeds $C''$. If instead $C' \notin \mathsf{C}(\Delta_P) \smallsetminus \mathsf{C}(\Gamma_P)$, then we must have that $(C',\vec{k'}) \in \Gamma_P$, since $C' \in s_P(C'') = s'_P(C'') \subseteq \mathsf{C}(\Delta_P \cup \Gamma_P)$ and $C' \in \mathsf{C}(\Delta_P \cup \{(C',\vec{k}')\}) \smallsetminus \mathsf{C}(\Gamma_P \smallsetminus \{(C',\vec{k}')\})$ because $s_P(C') \neq \emptyset$. In fact, by inspecting the code it can be seen that $C'$ must have been added as an assumption after exploring $C''$, which then became a decision $(C'',d_P(C''))$, and it must have held $\vec{k}' <_P \mathit{next}(d_P(C''),\priori{C''})$ as required by the third \cmd{if} branch in the function \fnname{Explore}. Thus, in both cases we have $\vec{k}' = d_P(C') <_P \mathit{next}(d_P(C''),\priori{C''})$. And so by Lemma~\ref{le:high-prior-path}(b) we know that $P = \exists$ iff $\eta_h = \nu$, where $h$ is the highest priority appearing along the subpath.
  For now assume $P = \exists$. Since this holds for all subpaths between two consecutive occurrences of $C'$, and there are infinitely many of them, which sequenced form the initial infinite path, then there must exist a priority $h$ s.t.\ $\eta_h = \nu$ and it is the highest priority appearing infinitely many times along the complete path. A dual reasoning holds for $P = \forall$. Recalling that an infinite play is won by player $\exists$ (resp.\ $\forall$) if the highest priority $h \in \interval{m}$ appearing infinitely often is s.t.\ $\eta_h = \nu$ (resp.\ $\mu$), we deduce that the path is won by $P$, whoever $P$ is. And so we conclude that $s_P$ is indeed winning in the modified game with assumptions $\Gamma_P \smallsetminus \{(C',\vec{k}')\}$.
\end{proof}

Now we can finally present the correctness result.
\end{toappendix}

\begin{theoremrep}[correctness]
  \label{th:alg-correctness}
  Given a fixpoint game, if a call \textup{\fnname{Explore}}($C$,
  $\vec{0}$, $[]$, $(\emptyset,\emptyset)$, $(\emptyset,\emptyset)$)
  returns a player $P$, then $P$ wins the game from $C$.
\end{theoremrep}

\begin{proof}
  Assume that the call \fnname{Explore}($C$, $\vec{0}$, $[]$,
  $(\emptyset,\emptyset)$, $(\emptyset,\emptyset)$) returns some
  player $P$. Since the only instruction returning a value is in the
  function \fnname{Backtrack} and it is reached only when $\rho = []$,
  then \fnname{Backtrack}($P$, $C'$, $[]$, $\Gamma$, $\Delta$) must
  have been invoked for some $\Gamma$ and $\Delta$. Furthermore,
  $C' = C$ because $\rho = []$ is the list of positions from the root
  $C$ to the current node $C'$. Also, by Lemma~\ref{le:alg-assumpt} we
  have that $\Gamma_P = \emptyset$. Thus, by
  Lemma~\ref{le:alg-backtrack-pos} we have that
  $(C,\vec{k}) \in \Delta_P$ for some counter $\vec{k}$. And so by
  Lemma~\ref{le:alg-inv} we can immediately conclude that $P$ wins the
  game from $C$, since the modified game with no assumptions coincides
  with the original one.
\end{proof}

Notice that it is unnecessary to prove the converse implication, that is, if $P$ wins the game from $C$, then the call \fnname{Explore}($C$, $\vec{0}$, $[]$, $(\emptyset,\emptyset)$, $(\emptyset,\emptyset)$) returns $P$. Indeed, since the game can never result in a draw, this is equivalent to show that if the call \fnname{Explore}($C$, $\vec{0}$, $[]$, $(\emptyset,\emptyset)$, $(\emptyset,\emptyset)$) returns $\oppo{P}$, then $\oppo{P}$ wins the game from $C$. And this already holds by Theorem~\ref{th:alg-correctness}.

\subsubsection{Using Up-To Techniques in the Algorithm}
\label{sec:up-to-algorithm}

In the literature about bisimilarity checking, up-to techniques have
been fruitfully integrated with local checking algorithm for speeding
up the computation (see,
e.g.,~\cite{h:mise-oeuvre-preuves-bisim}).
Here we show that a similar idea can be developed for our local algorithm for general systems of fixpoint equations.

Let $E$ be a system of $m$ equations of the kind
$\vec{x} =_{\vec{\eta}} \vec{f} (\vec{x})$ over a complete lattice $L$
and let $\vec{u}$ be a compatible tuple of up-to functions for $E$.
By Theorem~\ref{th:up-to-sys} we have that the system $E\bar{\vec{u}}$
with equations
$\vec{x} =_{\vec{\eta}} \vec{f}(\bar{\vec{u}} \cdot \vec{x})$ has the
same solution as $E$. Now, since $\bar{\vec{u}}$ is a tuple of functions obtained as least fixpoints (see Definition~\ref{de:lcu}), the system $E\bar{\vec{u}}$ can be ``equivalently'' written as the system of $2m$ equations that we denote by
$\sysupto{E}{\vec{u}}$, defined as follows: 
\begin{align*}
  \vec{y} & =_{\vec{\mu}} (\vec{u} \cdot \vec{y}) \sqcup \vec{x}\\
  \vec{x} & =_{\vec{\eta}} \vec{f}(\vec{y})
\end{align*}

More precisely, we can show the following result.

\begin{theoremrep}[preserving solutions with up-to]
  \label{th:up-to-sys-sol}
  Let $E$ be a system of $m$ equations of the kind
  $\vec{x} =_{\vec{\eta}} \vec{f} (\vec{x})$ over a complete lattice $L$.
  Let $\vec{u}$ be a $m$-tuple of up-to functions compatible for $E$
  (Definition~\ref{de:compatible-systems}). The solution of the system $\sysupto{E}{\vec{u}}$
  is $\sol{\sysupto{E}{\vec{u}}} = (\sol{E},\sol{E})$.
\end{theoremrep}

\begin{proof}
  We proceed by induction on the length $m$ of the original
  system. The base case is vacuously true since, for $m = 0$, both
  systems have empty solution. Then, for $m > 0$, assume that the
  property holds for systems of size $m-1$. By definition of solution
  we have that the solution of $x_m$ is
  \begin{eqnarray*}
    \sol[2m]{\sysupto{E}{\vec{u}}} & = & \eta_m (\lambda x.\, f_m(\sol[1,m]{\subst{\sysupto{E}{\vec{u}}}{x_m}{x}}))
  \end{eqnarray*}
  and the parametric solution of $y_m$ is the function $s' : L^m \to L$
  \begin{eqnarray*}
    s'(\vec{x}') & = & \sol[m]{\subst{\sysupto{E}{\vec{u}}}{\vec{x}}{\vec{x}'}} = \mu (\lambda y.\, u_m(y) \sqcup x'_m).
  \end{eqnarray*}
  Observe that since $s'(\vec{x}')$ depends only on $x'_m$, we can define the parametric solution of $y_m$ using just a function $s : L \to L$ instead of $s'$
  \begin{eqnarray*}
    s(x) & = & \mu (\lambda y.\, u_m(y) \sqcup x).
  \end{eqnarray*}
  Substituting the parametric solution of $y_m$ in the solution of $x_m$ we obtain
  \begin{eqnarray*}
    \sol[2m]{\sysupto{E}{\vec{u}}} & = & \eta_m (\lambda x.\, f_m(\sol[1,m-1]{\subst{\subst{\sysupto{E}{\vec{u}}}{x_m}{x}}{y_m}{s(x)}},s(x))).
  \end{eqnarray*}
  Let $h(x) =
  f_m(\sol[1,m-1]{\subst{\subst{\sysupto{E}{\vec{u}}}{x_m}{x}}{y_m}{s(x)}},s(x))$
  and $g_x(y) = u_m(y) \sqcup x$, so that
  $\sol[2m]{\sysupto{E}{\vec{u}}} = \eta_m (h)$ and $s(x) = \mu
  (g_x)$. Clearly $h$ and $g$ are both monotone (hence $s$ as
  well). The former because the solutions of a system (see
  \cite{BKMP:FPCL}) and $\vec{f}$ are monotone, the latter because
  both $u_m$ and the supremum are. Also notice that $s$ is an
  extensive function, i.e., $x \sqsubseteq s(x)$ for all $x$. In fact,
  since $s$ computes a (least) fixpoint we have that $s(x) = u_m(s(x))
  \sqcup x$, and clearly $x \sqsubseteq u_m(s(x)) \sqcup x$ by
  definition of supremum. Furthermore, we can prove that $s$ is
  compatible (wrt.\ $h$, i.e., $s(h(x)) \sqsubseteq h(s(x))$ for all
  $x$), continuous, and strict, whenever $u_m$ satisfies those
  conditions, respectively. First, if $u_m$ is continuous, then so is
  $g$ in both variables, since $\sqcup$ is continuous. Then, since
  $s(x)$ is the least fixpoint of $g_x$, it is immediate that $s$ is
  continuous as well. Recalling that $s(x) = g^\alpha_x(\bot)$ for some ordinal $\alpha$, both remaining properties can be proved by transfinite induction on $g^\alpha_x(\bot)$ for every $\alpha$. First we show that for all $x$, $g^\alpha_{h(x)}(\bot) \sqsubseteq h(s(x))$ for every ordinal $\alpha$ (hence $s(h(x)) \sqsubseteq h(s(x))$). For $\alpha = 0$, we have $g^0_{h(x)}(\bot) = \bot \sqsubseteq h(s(x))$. For a successor ordinal $\alpha = \beta + 1$, we have $g^{\beta+1}_{h(x)}(\bot) = g_{h(x)}(g^\beta_{h(x)}(\bot))$, and by inductive hypothesis we know that $g^\beta_{h(x)}(\bot) \sqsubseteq h(s(x))$. Then
  \begin{align*}
    &g_{h(x)}(g^\beta_{h(x)}(\bot)) \\
    \sqsubseteq &\qquad\mbox{[since $g$ is monotone]}\\
    &g_{h(x)}(h(s(x))) \\
    = &\qquad \mbox{[by definition of $g$]}\\
    &u_m(h(s(x))) \sqcup h(x) \\
    = &\qquad \mbox{[by definition of $h$]}\\
    &u_m(f_m(\sol[1,m-1]{\subst{\subst{\sysupto{E}{\vec{u}}}{x_m}{s(x)}}{y_m}{s^2(x)}},s^2(x))) \sqcup h(x) \\
    \sqsubseteq &\qquad\mbox{[by compatibility of $\vec{u}$]} \\
    &f_m(\vec{u} \cdot (\sol[1,m-1]{\subst{\subst{\sysupto{E}{\vec{u}}}{x_m}{s(x)}}{y_m}{s^2(x)}},s^2(x))) \sqcup h(x) 
  \end{align*}
  Observe that $u_m(s(z)) \sqsubseteq s(z) = g_{z}(s(z)) = u_m(s(z)) \sqcup z$ for all $z$. A similar reasoning applies to the other solutions as well, obtaining that $u_i(\sol[i]{\subst{\subst{\sysupto{E}{\vec{u}}}{x_m}{s(x)}}{y_m}{s^2(x)}}) \sqsubseteq \sol[i]{\subst{\subst{\sysupto{E}{\vec{u}}}{x_m}{s(x)}}{y_m}{s^2(x)}}$ for all $i \in \interval{m-1}$. Therefore we have
  \begin{align*}
    & f_m(\vec{u} \cdot (\sol[1,m-1]{\subst{\subst{\sysupto{E}{\vec{u}}}{x_m}{s(x)}}{y_m}{s^2(x)}},s^2(x))) \sqcup h(x) & \\
    \sqsubseteq & \qquad\mbox{[since $f_m$ is monotone]}\\
    & f_m(\sol[1,m-1]{\subst{\subst{\sysupto{E}{\vec{u}}}{x_m}{s(x)}}{y_m}{s^2(x)}},s^2(x)) \sqcup h(x) \\
    = & \qquad \mbox{[by definition of $h$]}\\
    & h(s(x)) \sqcup h(x) \\ 
    \sqsubseteq & \qquad \mbox{[since $h$ is monotone]}\\
    & h(s(x) \sqcup x) \\
    = & \qquad \mbox{[since $s$ is extensive]} \\
    & h(s(x)) 
  \end{align*}
  And so we established that $g^{\beta+1}_{h(x)}(\bot) \sqsubseteq h(s(x))$. For $\alpha$ limit ordinal, by inductive hypothesis we immediately have that $g^\alpha_{h(x)}(\bot) = \lub\limits_{\beta < \alpha} g^\beta_{h(x)}(\bot) \sqsubseteq \lub h(s(x)) = h(s(x))$. Now we show that $g^\alpha_\bot(\bot) = \bot$ for every ordinal $\alpha$. For $\alpha = 0$, we have $g^0_\bot(\bot) = \bot$. For $\alpha = \beta + 1$, by inductive hypothesis we have that $g^{\beta+1}_\bot(\bot) = g_\bot(g^\beta_\bot(\bot)) = g_\bot(\bot)$. And in turn, $g_\bot(\bot) = u_m(\bot) \sqcup \bot = \bot$, since $u_m$ is strict. For $\alpha$ limit ordinal, by inductive hypothesis we obtain that $g^\alpha_\bot(\bot) = \lub\limits_{\beta < \alpha} g^\beta_\bot(\bot) = \lub \bot = \bot$. Now we have two different cases depending on $\eta_m$.
  \begin{itemize}
    \item $\eta_m = \nu$\\
    In this case $\sol[2m]{\sysupto{E}{\vec{u}}} = h^\alpha(\top)$ for some ordinal $\alpha$. Here we show that actually $s(h^\alpha(\top)) = h^\alpha(\top)$ for every ordinal $\alpha$. Since as we mentioned above $s$ is extensive, we just need to prove that $s(h^\alpha(\top)) \sqsubseteq h^\alpha(\top)$ for every ordinal $\alpha$. We proceed by transfinite induction on $\alpha$. For $\alpha = 0$, we have $s(h^0(\top)) \sqsubseteq \top = h^0(\top)$. If $\alpha$ is a successor ordinal $\beta+1$, assuming the property holds for $\beta$, we show that $s(h^{\beta+1}(\top)) \sqsubseteq h^{\beta+1}(\top)$. Since $h$ is monotone, by inductive hypothesis we have that $h(s(h^\beta(\top))) \sqsubseteq h(h^\beta(\top)) = h^{\beta+1}(\top)$. Recalling that $s(h(x)) \sqsubseteq h(s(x))$ for all $x$, we also have that $s(h^{\beta+1}(\top)) = s(h(h^{\beta}(\top))) \sqsubseteq h(s(h^\beta(\top)))$. When $\alpha$ is a limit ordinal we have that $h^\alpha(\top) = \glb\limits_{\beta < \alpha} h^\beta(\top)$. Since $s$ is monotone, we have that $s(h^\alpha(\top)) = s(\glb\limits_{\beta < \alpha} h^\beta(\top)) \sqsubseteq \glb\limits_{\beta < \alpha} s(h^\beta(\top))$. And since by inductive hypothesis $s(h^\beta(\top)) \sqsubseteq h^\beta(\top)$ for all $\beta < \alpha$, we conclude also that $\glb\limits_{\beta < \alpha} s(h^\beta(\top)) \sqsubseteq \glb\limits_{\beta < \alpha} h^\beta(\top)$.
    \item $\eta_m = \mu$\\
    In this case $\sol[2m]{\sysupto{E}{\vec{u}}} = h^\alpha(\bot)$ for some ordinal $\alpha$. Recall also that since $\eta_m = \mu$, by hypothesis we know that $u_m$ is continuous and strict. In such case, as shown above, $s$ is continuous and strict as well. Again, we already know that $s$ is extensive, so we just prove by transfinite induction that $s(h^\alpha(\bot)) \sqsubseteq h^\alpha(\bot)$ for every ordinal $\alpha$. For $\alpha = 0$, we have $s(h^0(\bot)) = s(\bot) = \bot$, since $s$ is strict. If $\alpha$ is a successor ordinal $\beta+1$, assuming the property holds for $\beta$, we show that $s(h^{\beta+1}(\bot)) \sqsubseteq h^{\beta+1}(\bot)$. Since $h$ is monotone, by inductive hypothesis we have that $h(s(h^\beta(\bot))) \sqsubseteq h(h^\beta(\bot)) = h^{\beta+1}(\bot)$. Recalling that $s(h(x)) \sqsubseteq h(s(x))$ for all $x$, we also have that $s(h^{\beta+1}(\bot)) = s(h(h^{\beta}(\bot))) \sqsubseteq h(s(h^\beta(\bot)))$. When $\alpha$ is a limit ordinal we have that $h^\alpha(\bot) = \lub\limits_{\beta < \alpha} h^\beta(\bot)$. Since $s$ is continuous, we have that $s(h^\alpha(\bot)) = s(\lub\limits_{\beta < \alpha} h^\beta(\bot)) = \lub\limits_{\beta < \alpha} s(h^\beta(\bot))$. And since by inductive hypothesis $s(h^\beta(\bot)) \sqsubseteq h^\beta(\bot)$ for all $\beta < \alpha$, we conclude also that $\lub\limits_{\beta < \alpha} s(h^\beta(\bot)) \sqsubseteq \lub\limits_{\beta < \alpha} h^\beta(\bot)$.
  \end{itemize}
  So in both cases we have $s(h^\alpha(\top)) = h^\alpha(\top)$ or $s(h^\alpha(\bot)) = h^\alpha(\bot)$), respectively, for every ordinal $\alpha$. Consider the function $h'(x) = f_m(\sol[1,m-1]{\subst{\subst{\sysupto{E}{\vec{u}}}{x_m}{x}}{y_m}{x}},x)$. The previous fact implies that actually $\eta_m (h') = \eta_m (h) = \sol[2m]{\sysupto{E}{\vec{u}}}$. Furthermore, for the same reason we have that $s(\sol[2m]{\sysupto{E}{\vec{u}}}) = \sol[2m]{\sysupto{E}{\vec{u}}}$. Since $\sol[2m]{\sysupto{E}{\vec{u}}}$ is the solution of $x_m$ and by definition of solution $s(\sol[2m]{\sysupto{E}{\vec{u}}}) = \sol[m]{\sysupto{E}{\vec{u}}}$ is that of $y_m$, this means that $x_m$ and $y_m$ have the same solution in $\sysupto{E}{\vec{u}}$.
  So we can rewrite the solutions of $x_m$ and $y_m$ as $\eta_m (h')$, that is
  \begin{eqnarray*}
    \sol[2m]{\sysupto{E}{\vec{u}}} = \sol[m]{\sysupto{E}{\vec{u}}} & = & \eta_m (\lambda x.\, f_m(\sol[1,m-1]{\subst{\subst{\sysupto{E}{\vec{u}}}{x_m}{x}}{y_m}{x}},x)).
  \end{eqnarray*}
  Now, observe that the system
  $\subst{\subst{\sysupto{E}{\vec{u}}}{x_m}{x}}{y_m}{x}$ is actually
  $\sysupto{\subst{E}{x_m}{x}}{\vec{u}_{1,m-1}}$. Therefore, since
  $\subst{E}{x_m}{x}$ has size $m-1$, by inductive hypothesis we know
  that
  \begin{align*}
    \sol[1,m-1]{\subst{\subst{\sysupto{E}{\vec{u}}}{x_m}{x}}{y_m}{x}}
    &=
    \sol[m,2m-2]{\subst{\subst{\sysupto{E}{\vec{u}}}{x_m}{x}}{y_m}{x}} \\
    &= \sol{\subst{E}{x_m}{x}}.
  \end{align*}
  Thus, substituting these solutions in those of $x_m$ and $y_m$
  above, we obtain
  \begin{eqnarray*}
    \sol[2m]{\sysupto{E}{\vec{u}}} = \sol[m]{\sysupto{E}{\vec{u}}} & = & \eta_m (\lambda x.\, f_m(\sol{\subst{E}{x_m}{x}},x))
  \end{eqnarray*}
  which is also the definition of the solution of $x_m$ in $E$. Which means that $\sol[2m]{\sysupto{E}{\vec{u}}} = \sol[m]{\sysupto{E}{\vec{u}}} = \sol[m]{E}$. Then, the remaining solutions are
  \begin{align*}
    & (\sol[1,m-1]{\sysupto{E}{\vec{u}}}, \sol[m+1,2m-1]{\sysupto{E}{\vec{u}}}) & \\
    & = \sol{\subst{\subst{\sysupto{E}{\vec{u}}}{x_m}{\sol[2m]{\sysupto{E}{\vec{u}}}}}{y_m}{\sol[m]{\sysupto{E}{\vec{u}}}}} & \mbox{[by definition of solution]}\\
    & = \sol{\subst{\subst{\sysupto{E}{\vec{u}}}{x_m}{\sol[m]{E}}}{y_m}{\sol[m]{E}}} & \\
    & = (\sol{\subst{E}{x_m}{\sol[m]{E}}}, \sol{\subst{E}{x_m}{\sol[m]{E}}}) & \mbox{[by inductive hypothesis]}\\
    & = (\sol[1,m-1]{E}, \sol[1,m-1]{E}) & \mbox{[by definition of solution]}
  \end{align*}
  This and the previous fact allow us to conclude that
  \[ \sol{\sysupto{E}{\vec{u}}} = (\sol[1,m-1]{E}, \sol[m]{E},
    \sol[1,m-1]{E}, \sol[m]{E}), \] that is indeed
  $\sol{\sysupto{E}{\vec{u}}} = (\sol{E}, \sol{E})$.
\end{proof}

By relying on Theorem~\ref{th:up-to-sys-sol} we can derive an
algorithm that exploits the up-to function $\vec{u}$. It is
obtained by instantiating the general algorithm discussed before to
the system $\sysupto{E}{\vec{u}}$ and suitably restricting the moves
considered in the exploration. Roughly, the idea is to allow the use of the up-to function only when it leads immediately to an assumption or a decision.
This is in some sense similar to what is done for bisimilarity
checking in~\cite{h:mise-oeuvre-preuves-bisim}, where the up-to
function is used only to enlarge the set of states which are
considered bisimilar.
More precisely, when the exploration is in a position $(b,i)$
corresponding to one of the added equations
$y_i =_\mu u_i(y_i) \sqcup x_i$, according to the definition of the
game, a possible move would be any $2m$-tuple of sets
$(\vec{Y},\vec{X})$ such that
$b \sqsubseteq u_i(\bigsqcup Y_i) \sqcup \bigsqcup X_i$. First of all,
since only the $i$-th and $(m+i)$-th components $Y_i$ and $X_i$ play a
role and we can restrict to minimal moves (see
\S\ref{se:game-characterization}), we can assume
$X_j=Y_j=\emptyset$ for $j \neq i$. Moreover, for $X_i$ and $Y_i$, we
only allow two types of moves:
\begin{enumerate}

\item  $X_i=\{b\}$ and $Y_i=\emptyset$, which means that we keep the focus on element $b$ and just jump to the ``original'' equation $x_i=_{\eta_i}f(y_i)$, or

\item $X_i=\emptyset$ and all positions in $Y_i$ will immediately
  become assumptions or decisions when explored.
\end{enumerate}

At the level of the pseudocode, this only means that the action
``pick'' needs to be refined. Instead of
simply choosing randomly a move in $\moves{C}$, in some cases it has
to perform a constrained choice. This is made precise below.

\begin{definition}[up-to algorithm]
  \label{de:alg-upto}
  Let $E$ be a system of $m$ fixpoint equations over the complete
  lattice $L$ and let $\vec{u}$ be a compatible tuple of up-to
  function for $E$.
  The up-to algorithm for $E$ based on $\vec{u}$ is just the algorithm
  in Fig.~\ref{fi:alg} applied to the system $\sysupto{E}{\vec{u}}$,
  where, in function \fnname{Explore}($C$, $\vec{k}$, $\rho$,
  $\Gamma$, $\Delta$), when $C= (b,i)$ with $i \in \interval{m}$,
  the action ``pick'' can select only moves $C' = (\vec{Y},\vec{X})$ such that
  $Y_j=X_j =\emptyset$ for $j \neq i$ and $X_i, Y_i$ complying with
  either of the following conditions
  \begin{enumerate}
  \item
    \label{alg-upto-cond-a}
    $Y_i = \emptyset$ and $X_{i} = \{b\}$ or
  \item
    \label{alg-upto-cond-b}
    $X_i = \emptyset$ and for all $b' \in Y_i$ it
    holds
    \begin{enumerate}
    \item
      \label{alg-upto-cond-b1}
      $((b',i),\vec{k}') \in \Delta_\exists$ with
      $\vec{k}' \leq_\exists \mathit{next}(\vec{k},i)$ or
    \item
      \label{alg-upto-cond-b2}
      $((b',i),\vec{k}',\pi) \in \rho$ with
      $\vec{k}' <_\exists \mathit{next}(\vec{k},i)$.
    \end{enumerate}
  \end{enumerate}
\end{definition}

Condition (\ref{alg-upto-cond-a}) has been already clarified above. Condition (\ref{alg-upto-cond-b}) is a formal translation of the fact that $Y_i$ can contain only positions for which there are usable decisions (case (\ref{alg-upto-cond-b1})) or that will immediately become assumptions (case \ref{alg-upto-cond-b2})).


Clearly the modification does not affect termination on finite
lattices (in fact, we just restrict the possible moves of a procedure
which is known to be terminating). We next show that the up-to
algorithm is also correct.

\begin{theoremrep}[correctness with up-to]
  \label{th:alg-up-to-correct}
  Let $E$ be a system of $m$ equations of the kind
  $\vec{x} =_{\vec{\eta}} \vec{f} (\vec{x})$ over a complete lattice
  $L$. Let $\vec{u}$ a compatible $m$-tuple of up-to functions for
  $E$. Then the up-to algorithm associated with the system
  $\sysupto{E}{\vec{u}}$ as given in Definition~\ref{de:alg-upto} is
  correct, i.e., if a call \textup{\fnname{Explore}}($C$, $\vec{0}$,
  $[]$, $(\emptyset,\emptyset)$, $(\emptyset,\emptyset)$) returns a
  player $P$, then $P$ wins the game from $C$.
\end{theoremrep}

\begin{proof}
  Let $G$ be the fixpoint game associated with the initial system $E$,
  $G_u$ be the one associated with the modified system
  $\sysupto{E}{\vec{u}}$, and $G'_u$ be the game obtained from $G_u$
  by restricting the moves of player $\exists$ from positions
  associated with variables $y_i$ to only those satisfying either
  condition~(\ref{alg-upto-cond-a})
  or~~(\ref{alg-upto-cond-b}). Observe that the moves from every
  position controlled by player $\exists$ of $G$ are included in the
  moves from the corresponding position in $G'_u$ since they satisfy
  condition~(\ref{alg-upto-cond-a}), since in $E$ there are no up-to
  functions. Therefore, every winning strategy for $\exists$ in $G$
  can be easily converted into a winning strategy for the same player
  in $G'_u$. So the winning positions of player $\exists$ in $G$ are
  necessarily included in those of $G'_u$. Furthermore, the same
  clearly happens between $G'_u$ and $G_u$ since the moves of
  $\exists$ in $G'_u$ are defined as a restriction of those in
  $G_u$. Then, calling $W_\exists(G)$ the set of winning positions of
  player $\exists$ in the corresponding $G$, we have that
  $W_\exists(G) \subseteq W_\exists(G'_u) \subseteq W_\exists(G_u) =
  W_\exists(G)$, where the last equality holds by
  Theorem~\ref{th:up-to-sys-sol}. Since in our case
  every position not winning for $\exists$ is necessarily winning for
  $\forall$, this means that even if we restrict certain moves of
  player $\exists$, thus playing in the game $G'_u$, we still have the
  same exact winning positions for both players.
\end{proof}

The proof is based on the observation that any winning strategy for player  $\exists$
in the game associated with the original system $E$ can be replicated in the game associated with the modified system $\sysupto{E}{\vec{u}}$, even when the moves are restricted as in Definition~\ref{de:alg-upto}. This is done by choosing always moves corresponding to case (\ref{alg-upto-cond-a}) in Definition~\ref{de:alg-upto}.
Then strategies in the constrained game for 
$\sysupto{E}{\vec{u}}$ are also valid in the unconstrained game.  We conclude since, by Theorem~\ref{th:up-to-sys-sol}, we know that winning positions for  player $\exists$ are the same in the game for $E$ and in the game for $\sysupto{E}{\vec{u}}$.

Further optimizations of the up-to algorithm are possible by
exploiting the fact that a variable $y_i$ has the same solution of the
corresponding $x_i$ in the system $\sysupto{E}{\vec{u}}$. Intuitively,
decisions and assumptions for positions associated with a variables
$y_i$ could be used as decisions and assumptions for the corresponding
positions of variable $x_i$, and the other way around.

\begin{example}[model-checking $\mu$-calculus up-to bisimilarity]
  \label{ex:mc-mu-bisim}
  In Example~\ref{ex:mc-mu} we showed how the algorithm would solve a
  model-checking problem by exploring the corresponding fixpoint
  game. Suppose that this time we also want to use up-to bisimilarity
  as an up-to technique to answer the same question, that is, whether
  the state $a \in \mathbb{S}$ satisfies the formula $\phi = \mu
  x_2.((\nu x_1.(p \land\Box x_1))\lor\Diamond x_2)$. In
  Example~\ref{ex:up-to-bisim} we presented the up-to function $u_\sim
  : \Pow{\mathbb{S}} \to \Pow{\mathbb{S}}$ corresponding to up-to
  bisimilarity defined as $u_\sim(X) = \{s \in \mathbb{S}\,\mid\,s
  \sim_T s'\ \land\ s' \in X\}$. In order to apply the procedure
  described above, first we need to build the system
  $\sysupto{E}{(u_\sim,u_\sim)}$, which is
  \begin{align*}
    y_1 & =_{\mu} u_\sim(y_1) \cup x_1 & \qquad
    x_1 & =_{\nu} \{b, d, e\} \cap \sembox_T y_1 \\
    y_2 & =_{\mu} u_\sim(y_2) \cup x_2 & \qquad
    x_2 & =_{\mu} y_1 \cup \semdia_T y_2
  \end{align*}
  Then, to check whether the state $a$ satisfies the formula $\phi$ we invoke the function \fnname{Explore}($(a,4)$, $\vec{0}$, $[]$, $(\emptyset, \emptyset)$, $(\emptyset, \emptyset)$), where the index $4$ is that of the variable $x_2$ in the system $\sysupto{E}{(u_\sim,u_\sim)}$. Then, the algorithm behaves in similar fashion to what described in Example~\ref{ex:mc-mu}. However, this time the exploration of position $(d,1)$ with counter $(0,0,1,2)$ is pruned by using the up-to function. Recalling that position $(b,1)$ occurred in the past, hence it is included in the playlist, with counter $(0,0,0,2)$, we have that condition~(\ref{alg-upto-cond-b}) above holds here for the move $(\{b\},\emptyset,\emptyset,\emptyset)$ since $d \sim b$, hence $d \in u_\sim(\{b\}) \cup \emptyset$, and $(0,0,0,2) <_\exists \mathit{next}((0,0,1,2),1) = (1,0,1,2)$. This leads to making an assumption for $(b,1)$ and then backtracking up to the root. The same happens when exploring the other branch, that is position $(e,1)$, since also $e \sim b$. Similarly to the previous example, the last invocation \fnname{Backtrack}($\exists$, (a,4), [], $\Gamma$, $\Delta$) returns player $\exists$. Indeed, $\exists$ wins starting from position $(a,4)$ since the state $a$ satisfies the formula $\phi$.
\end{example}

\section{Conclusion}
\label{se:conclusion}

Our contribution is based on the notion of
approximation as formalised in abstract
interpretation~\cite{cc:ai-unified-lattice-model,CC:SDPAF}.
Due to the intimate connection of Galois connections and closure
functions, there is a close correspondence with up-to techniques for
enhancing coinduction proofs~\cite{p:complete-lattices-up-to,ps:enhancements-coinductive},
originally developed for CCS~\cite{Mil:CC}. However, as far
as we know, recent research has only started to explore this
connection: \cite{bggp:sound-up-to-complete-abstract} explains the
relation between sound up-to techniques and complete abstract
domains in the setting where the semantic function has an
adjoint. This adjunction or Galois connection plays a different role
than the abstractions: it gives the existential player a unique best
move, a concept explored in \S\ref{sec:on-the-fly-special}.

Fixpoint equation systems largely derive their interest from
$\mu$-calculus model-checking~\cite{bw:mu-calculus-modcheck}.
Evaluating $\mu$-calculus formulae on a transition system can be
reduced to solving a parity game and the exact complexity of this task
is still open. Progress measures, introduced
in~\cite{j:progress-measures-parity}, allow one to solve parity games
with a complexity which is polynomial in the number of states and
exponential in (half of) the alternation depth of the
formula. Recently quasi-polynomial algorithms for parity
games~\cite{CJKLS:DPGQPT,jl:success-progress-measures,l:modal-mu-parity-quasi-polynomial}
have been devised.  Instead of improving the complexity bounds, our
aim here is to introduce heuristics, based on a local algorithm
and up-to functions that are known to achieve good efficiency in
practice, in particular we explained how to combine the up-to
technique with $\mu$-calculus model-checking algorithm.

Many papers deal with abstraction in the setting of $\mu$-calculus
model checking. We noted that the results on simulation-based abstraction in~\cite{lgsbb:property-preserving-abstractions} can be obtained as an instance of our framework. The abstraction of the $\mu$-calculus along a Galois connection and its soundness is discussed in~\cite{BG:CBAS}.
A general framework for abstract interpretation of temporal calculi and logics is developed in~\cite{CC:TLA}. In particular, an abstract calculus for expressing nested fixpoint expressions is studied,  parametric with respect to the basic operators. The calculus is interpreted over complete boolean lattices and conditions ensuring the soundness and the completeness of the abstraction along a Galois connection are singled out. Such results are closely related to those in Section~\ref{se:abstractions}. The main differences reside in the fact that we work with general complete lattices, rather than with boolean lattices. In addition, we treat separately soundness and completeness, and, in order to establish a connection with up-to techniques, we distinguish two forms of completeness (for the abstraction and for the concretisation).

We showed -- for a special case -- how local algorithms
inspired by
\cite{BP:NFA,bggp:sound-up-to-complete-abstract,h:proving-up-to,h:mise-oeuvre-preuves-bisim}
for a single (greatest) fixpoint equation can be adapted to the case
of general lattices. For the general case of arbitrary fixpoint
equation systems a considerably more complex generalisation along the
lines of \cite{ss:practical-modcheck-games} is possible, but omitted due to lack of space.

The use of assumptions as stopping conditions in the algorithm is
reminiscent of parameterized coinduction
\cite{sm:precongruences-param-coinduction,HNDV:PPCP}, closely
related to up-to-techniques, as spelled out in
\cite{Pou:CAWU}.

\smallskip

The notion of progress measures that has been
studied in \cite{BKMP:FPCL} can be adapted to the game for arbitrary
complete (rather than just continuous) lattices, introduced in this
paper. A natural question
is whether the local
algorithm arises as an instance of the single equation algorithm
instantiated with the progress measure fixpoint equation.

With respect to the applications, we believe that our case study on
abstractions respectively simulations for $\mu$-calculus
model-checking can also be generalised to modal respectively
mixed transition systems
\cite{s:relations-abstraction-refinement,dgg:ai-reactive,lt:modal-process-logic}
or to abstraction for the full $\mu$-calculus as studied in
\cite{glls:abstraction-full-mu-calculus} by combining both under-
and over-approximations.
Furthermore, we plan to further study over-approximations for
fixpoint equations over the reals, closely connected to
probabilistic logics. In particular, we will investigate under
which circumstances one can obtain guarantees to be close to the
exact solution or to compute the exact solution directly.
Another interesting area is the use of up-to techniques for
behavioural metrics \cite{bkp:up-to-behavioural-metrics-fibrations}.

\begin{toappendix}
  \section{Comparison to the Bonchi/Pous Algorithm}
  \label{sec:bonchi-pous}

  In a seminal paper \cite{BP:NFA} Bonchi and Pous revisited the
  question of checking language equivalence for non-deterministic
  automata and presented an algorithm based on an up-to congruence
  technique that behaves very well in practice.

  We will here give a short description of this algorithm and then
  explain how it arises as a special case of the algorithm developed
  in \S\ref{sec:on-the-fly-special}.

  We are given a non-deterministic finite automaton
  $(Q,\Sigma,\delta,F)$, where $Q$ is the finite set of states,
  $\Sigma$ is the finite alphabet,
  $\delta\colon Q\times \Sigma\to \Pow{Q}$ is the transition function
  and $F\subseteq Q$ is the set of final states.  Note that we omit
  initial states. Given $a\in \Sigma$, $X\subseteq Q$ we define
  $\delta_a(X) = \bigcup_{q\in X} \delta(q,a)$.

  Given $q_1,q_2\in Q$, the aim is to show whether $q_1,q_2$ accept
  the same language (in the standard sense).

  In order to do this, the algorithm performs an on-the-fly
  determinization and constructs a bisimulation relation
  $R\subseteq \Pow{Q}\times \Pow{Q}$ on the determinized
  automaton. This relation has to satisfy the following properties:
  \begin{itemize}
  \item $\{q_1\}\,R\,\{q_2\}$
  \item Whenever $X_1\,R\,X_2$, then
    \begin{itemize}
    \item $\delta_a(X_1)\,R\,\delta_a(X_2)$ for all $a\in\Sigma$
      (\emph{transfer property})
    \item and $X_1\cap F\neq\emptyset\iff X_2\cap F\neq\emptyset$
      (\emph{one set is accepting iff the other is accepting})
    \end{itemize}
  \end{itemize}
  Due to the up-to technique there is no need to fully enumerate $R$.
  Instead in the second item above, it suffices to show that
  $\delta_a(X_1)\,c(R)\,\delta_a(X_2)$ where $c(R)$ is the congruence
  closure of $R$, i.e., the least relation $R'$ containing $R$ that is
  an equivalence and satisfies that $X_1\,R\,X_2$ implies
  $X_1\cup X\,R\,X_2\cup X$ (for $X_1,X_2,X\subseteq Q$). A major
  contribution of \cite{BP:NFA} is an algorithm for efficiently
  checking whether two given sets are in the congruence closure of a
  given relation. Here we will simply assume that this procedure is
  given and use it as a black box.

  \smallskip

  We will now translate this into our setting: the lattice is
  $L = \Pow{\Pow{Q}\times \Pow{Q}}$ (the lattice of all relations over
  the powerset of states) with inclusion as partial order. The basis
  $B$ consists of all singletons $\{(X_1,X_2)\}$ where
  $X_1,X_2\subseteq Q$. That is, we consider the setting of of
  \S\ref{sec:on-the-fly-special}.

  The behaviour map $f$ is given as follows: $f(R) = f^*(R)\cap C$
  where
  \begin{eqnarray*}
    f^*(R) & = & \{(X_1,X_2) \mid (\delta_a(X_1),\delta_a(X_2))\in R
    \text{ for all $a\in \Sigma$}\} \\
    C & = & \{(X_1,X_2) \mid X_1\cap F=\emptyset \iff X_2\cap F =
    \emptyset \}
  \end{eqnarray*}
  We want to solve a single fixpoint equation $R =_\nu f(R)$ where we
  are interested in the greatest fixpoint. In particular, we want to
  check whether $(Q_1,Q_2)\in R$ (where $Q_1 = \{q_1\}$,
  $Q_2 = \{q_2\}$) or alternatively $I = \{(Q_1,Q_2)\}\subseteq R$.

  Since we have determinized the automaton, $f^*$ has a left adjoint
  $f_*$, given as
  \[ f_*(R) = \{(\delta_a(X_1),\delta_a(X_2))\mid (X_1,X_2)\in R,
    a\in\Sigma\}. \] Now we can start exploring the game
  positions. Starting with $I = \{(Q_1,Q_2)\}\subseteq F$, the only
  move of $\exists$ is to play
  $\{\{(X_1,X_2)\} \mid (X_1,X_2)\in f_*(I)\}$, then it is the turn of
  $\forall$ who can choose any singleton set $\{(X_1,X_2)\}$ and one
  has to explore all those singletons. This continues until one
  encounters a singleton $\{(X_1,X_2)\}\not\subseteq C$ (which implies
  that $\exists$ has no move and loses) or one finds a set
  $\{(X_1,X_2)\}$ where one can cut off a branch due to the up-to
  technique -- more concretely $(X_1,X_2)\in c(W)$ where $W$ is the
  collection of all pairs visited so far on all paths and $c(W)$ is
  its congruence closure. One can conclude that $\exists$ wins if all
  encountered pairs are in $C$. This is a straightforward instance of
  the more general algorithm, enriched with an up-to technique, as
  explained in \S\ref{sec:on-the-fly-special}.
\end{toappendix}

\bibliography{references}

\begin{thebibliography}{10}

\bibitem{AJ:DT}
Samson Abramsky and Achim Jung.
\newblock Domain theory.
\newblock In Samson Abramsky, Dov Gabbay, and Thomas Stephen~Edward Maibaum,
  editors, {\em Handbook of Logic in Computer Science}, pages 1--168. Oxford
  University Press, 1994.

\bibitem{BKMP:FPCL}
Paolo Baldan, Barbara K\"onig, Christina Mika-Michalski, and Tommaso Padoan.
\newblock Fixpoint games in continuous lattices.
\newblock In Stephanie Weirich, editor, {\em Proc. of POPL'19}, volume~3, pages
  26:1--26:29. ACM, 2019.

\bibitem{bkp:abstraction-up-to-games-fixpoint}
Paolo Baldan, Barbara K\"onig, and Tommaso Padoan.
\newblock Abstraction, up-to techniques and games for systems of fixpoint
  equations.
\newblock In {\em Proc. of CONCUR '20}, volume 171 of {\em {LIPIcs}}, pages
  25:1--25:20. Schloss Dagstuhl -- Leibniz Center for Informatics, 2020.

\bibitem{BG:CBAS}
Gourinath Banda and John~P. Gallagher.
\newblock Constraint-based abstract semantics for temporal logic: {A} direct
  approach to design and implementation.
\newblock In Edmund~M. Clarke and Andrei Voronkov, editors, {\em {LPAR} 2010},
  volume 6355 of {\em Lecture Notes in Computer Science}, pages 27--45.
  Springer, 2010.

\bibitem{BdA:MCPNS}
Andrea Bianco and Luca de~Alfaro.
\newblock Model checking of probabalistic and nondeterministic systems.
\newblock In {\em Proc. of FSTTCS '95}, volume 1026 of {\em Lecture Notes in
  Computer Science}, pages 499--513. Springer, 1995.

\bibitem{bggp:sound-up-to-complete-abstract}
Filippo Bonchi, Pierre Ganty, Roberto Giacobazzi, and Dusko Pavlovic.
\newblock Sound up-to techniques and complete abstract domains.
\newblock In {\em Proc. of LICS '18}, pages 175--184. ACM, 2018.

\bibitem{bkp:up-to-behavioural-metrics-fibrations}
Filippo Bonchi, Barbara K{\"o}nig, and Daniela~Petri\c san.
\newblock Up-to techniques for behavioural metrics via fibrations.
\newblock In {\em Proc. of CONCUR '18}, volume 118 of {\em {LIPIcs}}, pages
  17:1--17:17. Schloss Dagstuhl -- Leibniz Center for Informatics, 2018.

\bibitem{BP:NFA}
Filippo Bonchi and Damien Pous.
\newblock Checking {NFA} equivalence with bisimulations up to congruence.
\newblock In {\em Proc. of POPL '13}, pages 457--468. {ACM}, 2013.

\bibitem{bw:mu-calculus-modcheck}
Julian Bradfield and Igor Walukiewicz.
\newblock The mu-calculus and model checking.
\newblock In Edmund~M. Clarke, Thomas~A. Henzinger, Helmut Veith, and Roderick
  Bloem, editors, {\em Handbook of Model Checking}, pages 871--919. Springer,
  2018.

\bibitem{CJKLS:DPGQPT}
Cristian~S. Calude, Sanjay Jain, Bakhadyr Khoussainov, Wei Li, and Frank
  Stephan.
\newblock Deciding parity games in quasipolynomial time.
\newblock In {\em Proc. of STOC '17}, pages 252--263. {ACM}, 2017.

\bibitem{cks:faster-modcheck-mu}
Rance Cleaveland, Marion Klein, and Bernhard Steffen.
\newblock Faster model checking for the modal mu-calculus.
\newblock In {\em Proc. of CAV 1992}, volume 663 of {\em Lecture Notes in
  Computer Science}, pages 410--422. Springer, 1992.

\bibitem{Cous:PCAFC}
Patrick Cousot.
\newblock Partial completeness of abstract fixpoint checking.
\newblock In Berthe~Y. Choueiry and Toby Walsh, editors, {\em Proceedings of
  the Fourth International Symposium on Abstraction, Reformulations and
  Approximation, SARA'2000}, volume 1864 of {\em Lecture Notes in Computer
  Science}, pages 1--25. Springer, 26--29 July 2000.

\bibitem{cc:ai-unified-lattice-model}
Patrick Cousot and Radhia Cousot.
\newblock Abstract interpretation: A unified lattice model for static analysis
  of programs by construction or approximation of fixpoints.
\newblock In {\em Proc. of POPL '77}, pages 238--252. ACM, 1977.

\bibitem{CC:SDPAF}
Patrick Cousot and Radhia Cousot.
\newblock Systematic design of program analysis frameworks.
\newblock In {\em Proc. of POPL '79}, pages 269--282. ACM, 1979.

\bibitem{CC:TLA}
Patrick Cousot and Radhia Cousot.
\newblock Temporal abstract interpretation.
\newblock In Mark~N. Wegman and Thomas~W. Reps, editors, {\em Proc. of POPL
  '00}, pages 12--25. {ACM}, 2000.

\bibitem{CC:CCLR}
Radhia Cousot and Patrick Cousot.
\newblock A constructive characterization of the lattices of all retractions,
  preclosure, quasi-closure and closure operators on a complete lattice.
\newblock {\em Portugaliae Mathematica}, 38(1-2):185--198, 1979.

\bibitem{CC:CCVTFP}
Radhia Cousot and Patrick Cousot.
\newblock Constructive versions of {T}arski's fixed point theorems.
\newblock {\em Pacific Journal of Mathematics}, 82(1):43--57, 1979.

\bibitem{dgg:ai-reactive}
Dennis Dams, Rob Gerth, and Orna Grumberg.
\newblock Abstract interpretation of reactive systems.
\newblock {\em ACM Trans.\ Program.\ Lang.\ Syst.}, 19(2):253--291, 1997.

\bibitem{GRS:MAIC}
Roberto Giacobazzi, Francesco Ranzato, and Francesca Scozzari.
\newblock Making abstract interpretations complete.
\newblock {\em Journal of the ACM}, 47(2):361--416, 2000.

\bibitem{glls:abstraction-full-mu-calculus}
Orna Grumberg, Martin Lange, Martin Leucker, and Sharon Shoham.
\newblock When not losing is better than winning: Abstraction and refinement
  for the full mu-calculus.
\newblock {\em Information and Computation}, 205(8):1130--1148, 2007.

\bibitem{hsc:lattice-progress-measures}
Ichiro Hasuo, Shunsuke Shimizu, and Corina C\^{\i}rstea.
\newblock Lattice-theoretic progress measures and coalgebraic model checking.
\newblock In {\em Proc. of POPL '16}, pages 718--732. ACM, 2016.

\bibitem{h:proving-up-to}
Daniel Hirschkoff.
\newblock Automatically proving up to bisimulation.
\newblock In {\em Proc. of MFCS '98 Workshop on Concurrency}, number~18 in
  Electronic Notes in Theoretical Computer Science, pages 75--89. Elsevier,
  1998.

\bibitem{h:mise-oeuvre-preuves-bisim}
Daniel Hirschkoff.
\newblock {\em Mise en oeuvre de preuves de bisimulation}.
\newblock PhD thesis, Ecole Nationale des Ponts et Chaussées, 1999.

\bibitem{HNDV:PPCP}
Chung{-}Kil Hur, Georg Neis, Derek Dreyer, and Viktor Vafeiadis.
\newblock The power of parameterization in coinductive proof.
\newblock In {\em Proc. of POPL '13}, pages 193--206. {ACM}, 2013.

\bibitem{hk:quantitative-analysis-mc}
Michael Huth and Marta Kwiatkowska.
\newblock Quantitative analysis and model checking.
\newblock In {\em Proc. of LICS '97}, pages 111--122. IEEE, 1997.

\bibitem{j:progress-measures-parity}
Marcin Jurdzi\'{n}ski.
\newblock Small progress measures for solving parity games.
\newblock In {\em Proc. of STACS '00}, volume 1770 of {\em Lecture Notes in
  Computer Science}, pages 290--301. Springer, 2000.

\bibitem{jl:success-progress-measures}
Marcin Jurdzinski and Ranko Lazic.
\newblock Succinct progress measures for solving parity games.
\newblock In {\em Proc. of LICS '17}, pages 1--9. {ACM/IEEE}, 2017.

\bibitem{k:prop-mu-calculus}
Dexter Kozen.
\newblock Results on the propositional $\mu$-calculus.
\newblock {\em Theoretical Computer Science}, 27(3):333--354, 1983.

\bibitem{lt:modal-process-logic}
Kim~Guldstrand Larsen and Bent Thomsen.
\newblock A modal process logic.
\newblock In {\em Proc. of LICS '88}, pages 203--210. IEEE, 1988.

\bibitem{l:modal-mu-parity-quasi-polynomial}
Karoliina Lehtinen.
\newblock A modal {\(\mu\)} perspective on solving parity games in
  quasi-polynomial time.
\newblock In {\em Proc. of LICS '18}, pages 639--648. {ACM/IEEE}, 2018.

\bibitem{lgsbb:property-preserving-abstractions}
Claire Loiseaux, Susanne Graf, Joseph Sifakis, Ahmed Bouajjani, and Saddek
  Bensalem.
\newblock Property preserving abstractions for the verification of concurrent
  systems.
\newblock {\em Formal Methods in System Design}, 6:1--35, 1995.

\bibitem{mm:quantitative-mu}
Annabelle McIver and Carroll Morgan.
\newblock Results on the quantitative $\mu$-calculus qm$\mu$.
\newblock {\em ACM Trans.\ Comp.\ Log.}, 8(1:3), 2007.

\bibitem{Mil:CC}
Robin Milner.
\newblock {\em Communication and Concurrency}.
\newblock Prentice Hall, 1989.

\bibitem{Min:AIT}
Antoine Min\'e.
\newblock Tutorial on static inference of numeric invariants by abstract
  interpretation.
\newblock {\em Foundations and Trends in Programming Languages},
  4(3-4):120--372, 2017.

\bibitem{MS:MS}
Matteo Mio and Alex Simpson.
\newblock {\L}ukasiewicz {\(\mu\)}-calculus.
\newblock {\em Fundamenta Informaticae}, 150(3-4):317--346, 2017.

\bibitem{p:complete-lattices-up-to}
Damien Pous.
\newblock Complete lattices and up-to techniques.
\newblock In {\em Proc. of APLAS '07}, pages 351--366. Springer, 2007.
\newblock {LNCS} 4807.

\bibitem{Pou:CAWU}
Damien Pous.
\newblock Coinduction all the way up.
\newblock In {\em Proc. of LICS'16}, pages 307--316. {ACM}, 2016.

\bibitem{ps:enhancements-coinductive}
Damien Pous and Davide Sangiorgi.
\newblock Enhancements of the bisimulation proof method.
\newblock In Davide Sangiorgi and Jan Rutten, editors, {\em Advanced Topics in
  Bisimulation and Coinduction}. Cambridge University Press, 2011.

\bibitem{s:bisimulation-coinduction}
Davide Sangiorgi.
\newblock {\em Introduction to Bisimulation and Coinduction}.
\newblock Cambridge University Press, 2011.

\bibitem{SM:PBUT}
Davide Sangiorgi and Robin Milner.
\newblock The problem of ``weak bisimulation up to''.
\newblock In W.R. Cleaveland, editor, {\em Proc. of CONCUR'92}, pages 32--46.
  Springer, 1992.

\bibitem{s:relations-abstraction-refinement}
David~A. Schmidt.
\newblock Binary relations for abstraction and refinement.
\newblock In {\em Workshop on Refinement and Abstraction}. Elsevier Electronic,
  2000.

\bibitem{Scott:CL}
Dana Scott.
\newblock Continuous lattices.
\newblock In F.~W. Lawvere, editor, {\em Toposes, Algebraic Geometry and
  Logic}, Lecture Notes in Mathematics, pages 97--136. Springer, 1972.

\bibitem{s:fast-simple-nested-fixpoints}
Helmut Seidl.
\newblock Fast and simple nested fixpoints.
\newblock {\em Information Processing Letters}, 59(6):303--308, 1996.

\bibitem{sm:precongruences-param-coinduction}
David Sprunger and Lawrence~S. Moss.
\newblock Precongruences and parametrized coinduction for logics for behavioral
  equivalence.
\newblock In {\em Proc. of CALCO '17}, volume~72 of {\em LIPIcs}, pages
  23:1--23:15. Schloss Dagstuhl -- Leibniz-Zentrum fuer Informatik, 2017.

\bibitem{ss:practical-modcheck-games}
Perdita Stevens and Colin Stirling.
\newblock Practical model-checking using games.
\newblock In {\em Proc. of TACAS '98}, volume 1384 of {\em Lecture Notes in
  Computer Science}, pages 85--101. Springer, 1998.

\bibitem{s:local-modcheck-games}
Colin Stirling.
\newblock Local model checking games.
\newblock In {\em Proc. of CONCUR '95}, volume 962 of {\em Lecture Notes in
  Computer Science}, pages 1--11. Springer, 1995.

\bibitem{t:lattice-fixed-point}
Alfred Tarski.
\newblock A lattice-theoretical fixpoint theorem and its applications.
\newblock {\em Pacific Journal of Mathematics}, 5:285--309, 1955.

\end{thebibliography}

\appendix

\end{document}